\pdfoutput=1

\documentclass[letterpaper,11pt]{article}

\usepackage{times}
\usepackage{amsmath}
\usepackage{amsfonts}
\usepackage{amssymb}
\usepackage{amsthm}
\usepackage{graphicx}
\usepackage[margin=1in]{geometry}
\usepackage{paralist}
\usepackage{xspace}
\usepackage[linesnumbered,algoruled,
vlined]{algorithm2e}
\SetKwInOut{Input}{input}
\SetKwInOut{Output}{output}
\SetKwComment{Comment}{// }{}
\SetCommentSty{textit}
\SetEndCharOfAlgoLine{}

\usepackage{color}
\usepackage{soul}

\newtheorem{theorem}{Theorem}[section]
\newtheorem{lemma}[theorem]{Lemma}
\newtheorem{definition}[theorem]{Definition}
\newtheorem{corollary}[theorem]{Corollary}

\newcommand{\namedref}[2]{\hyperref[#2]{#1~\ref*{#2}}}
\newcommand{\sectionref}[1]{\namedref{Section}{#1}}
\newcommand{\appendixref}[1]{\namedref{Appendix}{#1}}

\newcommand{\theoremref}[1]{\namedref{Theorem}{#1}}
\newcommand{\defref}[1]{\namedref{Definition}{#1}}
\newcommand{\figureref}[1]{\namedref{Figure}{#1}}

\newcommand{\lemmaref}[1]{\namedref{Lemma}{#1}}

\newcommand{\corollaryref}[1]{\namedref{Corollary}{#1}}

\newcommand{\appref}[1]{\namedref{Appendix}{#1}}

\newcommand{\algref}[1]{\namedref{Algorithm}{#1}}

\newcommand{\lineref}[1]{\namedref{Line}{#1}}
\newcommand{\equalityref}[1]{\hyperref[#1]{Equality~\eqref{#1}}}
\newcommand{\inequalityref}[1]{\hyperref[#1]{Inequality~\eqref{#1}}}

\newcommand{\N}{\mathbb{N}}
\newcommand{\R}{\mathbb{R}}

\newcommand{\BO}{\mathcal{O}}
\newcommand{\sO}{\tilde{\mathcal{O}}}
\newcommand{\sOmega}{\tilde{\Omega}}
\newcommand{\M}{\mathcal{M}}
\newcommand{\Comp}{\lambda}
\newcommand{\Set}[1]{\left\{#1\right\}}
\newcommand{\Congest}{\textsc{congest}\xspace}

\def\cA{\mathcal{A}}

\DeclareMathOperator{\act}{act}
\DeclareMathOperator{\moat}{rad}
\DeclareMathOperator{\Wd}{wd}
\DeclareMathOperator{\WD}{WD}
\DeclareMathOperator{\reg}{Reg}
\DeclareMathOperator{\vor}{Vor}
\DeclareMathOperator*{\argmin}{argmin}

\DeclareMathOperator{\polylog}{polylog}
\DeclareMathOperator{\cP}{{\cal P}}

\DeclareMathOperator{\sent}{sent}
\DeclareMathOperator{\unsent}{list}
\def\rW{\hat{W}}
\newcommand{\true}{\mathbf{true}}
\newcommand{\false}{\mathbf{false}}
\newcommand{\sfcr}{\textsc{dsf-cr}\xspace}
\newcommand{\sfic}{\textsc{dsf-ic}\xspace}

\renewcommand{\paragraph}[1]{\smallskip\par\noindent\textbf{#1}}

\usepackage[colorlinks,linkcolor=blue,filecolor=blue,citecolor=blue,urlcolor=blue
]{hyperref}

\begin{document}
\setcounter{tocdepth}{2}

\title{\textbf{Improved Distributed Steiner Forest Construction}
}

\author{
Christoph Lenzen\thanks{MIT CSAIL, The Stata Center, 32 Vassar Street,
Cambridge, MA 02139, USA.
\hbox{Email: {\tt clenzen@csail.mit.edu}.} Phone: +1 617-253-4632.
Supported by the Deutsche Forschungsgemeinschaft (DFG, reference
number Le 3107/1-1).}
\and 
Boaz Patt-Shamir\thanks{School of Electrical Engineering, Tel Aviv
  University, Tel Aviv 69978, Israel. Email:
  \texttt{boaz@eng.tau.ac.il}. Supported in part by Israel
  Ministry for Science and Technology.}
}

\date{}

\begin{titlepage}

\setcounter{page}{0}

\maketitle



\begin{abstract}
  We present new distributed algorithms for constructing a Steiner Forest in the
  \Congest model. Our deterministic algorithm finds, for any given constant
  $\varepsilon>0$, a $(2+\varepsilon)$-approximation in
  $\sO(sk+\sqrt{\min\Set{st,n}})$ rounds, where $s$ is the ``shortest path
  diameter,'' $t$ is the number of terminals, and $k$ is the number of terminal
  components in the input. Our randomized algorithm finds, with high
  probability, an $\BO(\log n)$-approximation in time $\sO(k+\min\Set{s,\sqrt
  n}+D)$, where $D$ is the unweighted diameter of the network. We prove a
  matching lower bound of $\sOmega(k+\min\Set{s,\sqrt
  n}+D)$on the running time of any distributed approximation
  algorithm for the Steiner Forest problem. The best previous algorithms were
  randomized and obtained either an $\BO(\log n)$-approximation in $\sO(sk)$
  time, or an $\BO(1/\varepsilon)$-approximation in $\sO((\sqrt
  n+t)^{1+\varepsilon}+D)$ time.
\end{abstract} 
\thispagestyle{empty}
\end{titlepage}

\section{Introduction}
Ever since the celebrated paper of Gallager, Humblet, and Spira \cite{GHS-83},
the task of constructing a minimum-weight spanning tree (MST) continues to be
a rich source of difficulties and ideas that drive network algorithmics (see,
e.g., \cite{Elkin-MST,GarayKP-98,LotkerPP,PelegR-00}).
The \emph{Steiner Forest} (SF) problem is a strict generalization of MST: We are
given a network with edge weights and some disjoint node subsets called
\emph{input components}; the task is to find a minimum-weight edge set which
makes each component connected. MST is a special case of SF, and so are the
Steiner Tree and shortest $s$-$t$ path problems.
The general SF problem is well motivated by many practical situations involving
the design of networks, be it physical (it was famously posed as a problem of
railroad design), or virtual (e.g., VPNs or streaming multicast). The problem
has attracted much attention in the classical algorithms community, as detailed
on the dedicated website \cite{Steiner-site}.

The first network algorithm for
SF in the \Congest model (where a link can deliver $\BO(\log n)$ bits in a
time unit---details in \sectionref{sec:model}) was presented by
Khan \textit{et al.}\ \cite{KKMPT-12}. 
It provides $\BO(\log n)$-approximate solutions in time $\sO(sk)$, where $n$ is
the number of nodes, $k$ is the number of components, and $s$ the \emph{shortest
path diameter} of the network, which is
(roughly---see \sectionref{sec:model})  the maximal number of edges 
in a weighted shortest path. 
Subsequently, in \cite{LenzenP13}, it was shown that for any given
$0<\varepsilon\le1/2$, an $\BO(\varepsilon^{-1})$-approximate solution to SF can
be found in time $\sO((\sqrt n+t)^{1+\varepsilon}+D)$, where $D$ is the diameter
of the unweighted version of the network, and $t$ is the number of
\emph{terminals}, i.e., the total number of nodes in all input components.  The
algorithms in \cite{KKMPT-12,LenzenP13} are both randomized.

\paragraph{Our Results.}
In this paper we improve the results for SF in the \Congest model in two ways.
First, we show that for any given constant $\varepsilon>0$, a
$(2+\varepsilon)$-approximate solution to SF can be computed by a deterministic
network algorithm in time $\sO(sk+\sqrt{\min\Set{st,n}})$.
Second, we show that an $\BO(\log n)$-approximation can be attained by a
randomized algorithm in time $\sO(k+\min\Set{s,\sqrt n}+D)\subseteq\sO(s+k)$.
On the other hand, we show that any algorithm in the \Congest model that
computes a solution to SF with non-trivial approximation ratio has running time
in $\sOmega(k+\min\Set{s,\sqrt n}+D)$. If the input is not given by
indicating to each terminal its input component, but rather by
\emph{connection
requests} between terminals, i.e., informing each terminal which
terminals  it must be
connected to, an  $\sOmega(t+\min\Set{s,\sqrt n}+D)$ lower bound
holds.  (It is easy to transform connection requests into equivalent
input components in $\BO(t+D)$ rounds.)

\paragraph{Related work.}
The Steiner Tree problem (the special case of SF where there is one
input component) has a remarkable history, starting with Fermat, who
posed the geometric 3-point on a plane problem circa 1643, including
Gauss (1836), and culminating with a popularization in 1941 by Courant
and Robbins in their book ``What is Mathematics'' \cite{CourantR-41}. An
interesting account of these early developments is given in
\cite{SteinerHistory}.
The contribution of Computer Science to the history of the problem apparently
started with the inclusion of Steiner Tree as one of the original 21 problems
proved NP-complete by Karp \cite{Karp-72}. There are quite a few variants of the
SF problem which are algorithmically interesting, such as Directed Steiner Tree,
Prize-Collecting Steiner Tree, Group Steiner Tree, and more. The site
\cite{Steiner-site} gives a continuously updated state of the
art results for many variants. Let us mention results for just the most common
variants: For the Steiner Tree problem, the best (polynomial-time) approximation
ratio known is $\ln 4+\varepsilon\approx1.386+\varepsilon$ for any constant
$\varepsilon>0$ \cite{ByrkaGRS-10}. For Steiner Forest, the best approximation
ratio known is $2-1/(t-k)$ \cite{AgrawalKR-95}. It is also known that the
approximation ratio of the Steiner Tree (or Forest) problem is  at least
$96/95$, unless P=NP~\cite{ChlebikC-08}.

Regarding distributed algorithms, there are a few relevant results. First, the
special case of minimum-weight spanning tree (MST) is known to have time
complexity of $\tilde\Theta(D+\sqrt n)$ in the \Congest model
\cite{DHKNPPW-11,Elkin-MST,GarayKP-98,KuttenP-98,PelegR-00}. In \cite{CF05}, a
2-approximation for the special case of Steiner Tree is presented, with time
complexity $\sO(n)$. The first distributed solution to the Steiner Forest
problem was presented by Khan \textit{et al.}~\cite{KKMPT-12}, where a
randomized algorithm is used to embed the instance in a virtual tree with
$\BO(\log n)$ distortion, then finding the optimal solution on the tree (which
is just the minimal subforest connecting each input component), and finally
mapping the selected tree edges back to corresponding paths in the original
graph. The result is an $\BO(\log n)$-approximation in time $\sO(sk)$.
Intuitively, $s$ is the time required by the Bellman-Ford algorithm to compute
distributed single-source shortest paths
, and the virtual tree of
\cite{KKMPT-12} is computed in $\sO(s)$ rounds.
A second distributed algorithm for Steiner Forest is presented in
\cite{LenzenP13}.
Here, a sparse spanner for the metric induced on the set of terminals and a
random sample of $\tilde{\Theta}(\sqrt{n})$ nodes is computed, on which the
instance then is solved centrally. To get an
$\BO(\varepsilon^{-1})$-approximation, the algorithm runs for $\sO(D+(\sqrt
n+t)^{1+\varepsilon})$ rounds. For approximation ratio $\BO(\log n)$,
the running time 
is $\sO(D+\sqrt n+t)$. 

\paragraph{Main Techniques.}
Our lower
bounds are
derived by the standard technique of reduction from results on $2$-party
communication complexity.
Our deterministic algorithm is an adaptation of the ``moat growing'' algorithm
of Agrawal, Klein, and Ravi \cite{AgrawalKR-95} to the \Congest model. It
involves determining the times in which ``significant events'' occur (e.g., all
terminals in an input component becoming connected by the currently selected
edges)
and extensive usage of pipelining. The algorithm generalizes
the MST algorithm from~\cite{KuttenP-98}
: for the special case of a Steiner Tree (i.e.,
$k=1$), one can interpret the output as the edge set induced by an MST of the
complete graph on the terminals with edge weights given by the terminal-terminal
distances, yielding a factor-$2$ approximation; specializing further to the MST
problem, the result is an exact MST and the running time becomes
$\sO(\sqrt{n}+D)$.

Our randomized algorithm is based on the embedding of the graph into a tree
metric from \cite{KKMPT-12}, but we improve the complexity of finding a Steiner
Forest. A key insight is that while the least-weight paths in the original
graph corresponding to virtual tree edges might intersect, no node participates
in more than $\BO(\log n)$ distinct paths. Since the union of
all least-weight paths ending at a specific node induces a tree, letting each
node serve routing requests corresponding to different destinations in a
round-robin fashion achieves a pipelining effect reducing the complexity to
$\sO(s+k)$.  If $s>\sqrt{n}$, the virtual tree and the corresponding
solution are constructed only partially, in time $\sO(\sqrt{n}+k+D)$,
and the partial result is used to create another instance with
$\BO(\sqrt n)$ terminals 
that captures the remaining connectivity demands; we solve it using
the algorithm from~\cite{LenzenP13}, obtaining an $\BO(\log n)$-approximation.



\paragraph{Organization.}
In \sectionref{sec:model} we define the model, problem and basic
concepts.  \sectionref{sec-lb} contains our lower bounds.  
In \sectionref{sec-alg1} and \sectionref{sec-alg2} we  present our
deterministic and randomized algorithms, respectively. 
We only give a high-level overview in this extended abstract. Proofs
are deferred to the appendix.

\section{Model and Notation}
\label{sec:model}
\paragraph{System Model.}
We consider the $\Congest(\log n)$ or simply the \Congest model as
specified in~\cite{Peleg:book}, briefly described as follows. The
distributed system is represented 
by a weighted graph $G=(V,E,W)$ of $n:=|V|$ nodes. The weights $W:E\to
\N$ are polynomially bounded in $n$ (and therefore polynomial sums of
weights can be encoded with $\BO(\log n)$ bits). Each node initially
knows its unique identifier of $\BO(\log n)$ bits, the identifiers of
its neighbors, the weight of its incident edges, and the local
problem-specific input specified below. Algorithms proceed in
synchronous rounds, where in each round, (i) nodes perform arbitrary,
finite local computations,%
\footnote{All our algorithms require polynomial computations only.} (ii) may
send, to each neighbor, a possibly distinct message of $\BO(\log n)$ bits, and
(iii) receive the messages sent by their neighbors. For randomized algorithms,
each node has access to an unlimited supply of unbiased, independent random
bits. Time complexity is measured by the number of rounds until all nodes
(explicitly) terminate. 

\paragraph{Notation.}
We use the following conventions and graph-theoretic notions.
\begin{compactitem}
\item The \emph{length} or number of \emph{hops} of a path $p=(v_0,\ldots,v_{\ell(p)})$ in $G$ is $\ell(p)$.
\item The weight of such a path is
  $W(p):=\sum_{i=1}^{\ell(p)}W(v_i,v_{i-1})$. For notational
  convenience, we assume w.l.o.g.\ that different paths have different
  weight (ties broken lexicographically).
\item By $\cP(v,w)$ we denote the set of all paths between $v,w\in V$ in $G$, i.e., $v_0=v$ and $v_{\ell(p)}=w$.
\item The (unweighted) \emph{diameter} of $G$ is \hfill\\$D:=\max_{v,w\in V}\{\min_{p\in \cP(v,w)}\{\ell(p)\}\}$.
\item The (weighted) \emph{distance} of $v$ and $w$ in $G$ is $\Wd(v,w):=\min_{p\in \cP(v,w)}\{W(p)\}$.
\item The \emph{weighted diameter} of $G$ is $\WD:=\max_{v,w\in V}\{\Wd(v,w)\}$.
\item Its \emph{shortest-path-diameter} is $s:=\max_{v,w\in V}\{\min \{\ell(p)\,|\,p\in \cP(v,w)\wedge W(p)=\Wd(v,w)\}\}$.
\item For $v\in V$ and $r\in \R^+_0$, we use $B_G(v,r)$ to denote the
  ball of radius $r$ around $v$ in $G$, which includes all nodes and
  edges at weighted distance at most $r$ from $v$. The ball may
  contain edge fractions: for an edge $\{w,u\}$ for which $w$ is in
  $B_G(v,r)$, the $(r-\Wd(v,w))/\Wd(v,w)$ fraction of the edge closer to $w$
  is considered to be within $B_G(v,r)$, and the remainder is considered
  outside $B_G(v,r)$.
\end{compactitem}
We use ``soft''  asymptotic notation.
Formally, given functions
$f$ and $g$, define (i) $f\in \tilde{\BO}(g)$ iff there is some $h\in
\polylog n$ so that $f\in \BO(gh)$, (ii) $f\in \tilde{\Omega}(g)$ iff
$g\in\sO(f)$, 
and (iii)
$f\in \tilde{\Theta}(g)$ iff $f\in \tilde{\BO}(g)\cap
\tilde{\Omega}(g)$.
By ``w.h.p.,'' we abbreviate ``with probability $1-n^{-\Omega(1)}$'' for a
sufficiently large constant in the $\Omega(1)$ term.

\paragraph{The Distributed Steiner Forest Problem.}
In the Steiner Forest problem, the output is a set of edges. We
require that the output edge set $F$ is represented distributively,
i.e., each node 
can locally answer which of its adjacent 
edges are in the output.
%
The input may be represented by two alternative methods, both are
justified and are common in the
literature. We give the two definitions.
\begin{definition}
[Distributed Steiner Forest with Connection Requests (\sfcr)]\ 
\begin{compactitem}
\item[\textbf{Input:}] At each node $v$, a set of \emph{connection
    requests} $R_v\subseteq V$.
\item[\textbf{Output:}] An edge set $F\subseteq E$ such that
  for each connection request $w\in R_v$, $v$ and $w$ are connected by
  $F$.
\item[\textbf{Goal:}] Minimize $W(F)=\sum_{e\in F}W(e)$.
\end{compactitem}
\end{definition}
\noindent

The set of \emph{terminal} nodes is defined to be $T=\Set{w\mid w\in
  R_v\mbox{ for some }v\in V}\cup \Set{v\mid R_v\ne\emptyset}$, i.e., the
set of nodes $v$ for which there is some connection request $\{v,w\}$.

\begin{definition}[Distributed Steiner Forest with Input Components (\sfic)]\ 
\begin{compactitem}
\item[\textbf{Input:}] At each node $v$, $\Comp(v)\in \Lambda \cup
  \{\bot\}$, where $\Lambda$ is the set of \emph{component
    identifiers}. The set of \emph{terminals} is $T:=\{v\in
V\mid\Comp(v)\neq \bot\}$. An \emph{input component} $C_{\Comp}$ for
$\lambda\ne\bot$ is the set of terminals with label $\Comp$.
\item[\textbf{Output:}] An edge set $F\subseteq E$ such that all terminals
  in each
  {input component}  are connected by $F$.
\item[\textbf{Goal:}] Minimize $W(F)=\sum_{e\in F}W(e)$.
\end{compactitem}
\end{definition}
An instance of \sfic is \emph{minimal}, if $|C_{\Comp}|\neq 1$ for all
$\Comp\in \Lambda$. We assume that the labels $\Comp\in \Lambda$ are encoded
using $\BO(\log n)$ bits. 
We define $t:=|T|$ and $k:=|\Lambda|\leq
t$, i.e., the number of terminals and input components, respectively. 

We say that any two instances of the above problems on the same
weighted graph, regardless of the way the input is given, are
\emph{equivalent} if the set of feasible outputs for the two instances
is identical.
\begin{lemma}\label{lemma:transform_to_input}
  Any instance of \sfcr can be transformed into an
  equivalent instance of \sfic in $\BO(D+t)$ rounds.
\end{lemma}

\begin{lemma}\label{lemma:transform_to_minimal}
Any instance of \sfic can be transformed into
an equivalent minimal instance of \sfic in $\BO(D+k)$ rounds.
\end{lemma}

\section{Lower Bounds}
\label{sec-lb}

In this section we state our lower bounds (for proofs and more
discussion, see \appendixref{app-lb}.) 
As our first result,
we show that applying \lemmaref{lemma:transform_to_input} to
instances of \sfcr comes at no penalty in asymptotic running time (a lower bound
of $\Omega(D)$ is trivial). 
\begin{lemma}
\label{lem-lb1}
  Any distributed algorithm for \sfcr with finite approximation ratio
  has time complexity $\Omega(t/\log n)$. This is true even in graphs
  with diameter at most $4$ and no more than two input components.
\end{lemma}

The main result of this section is the following theorem.
\begin{theorem}
\label{thm-lb}
  Any algorithm for the distributed Steiner Forest problem with non-trivial
  approximation ratio has worst-case time complexity in
  $\tilde{\Omega}(\min\{s,\sqrt{n}\}+k+D)$ in expectation.
\end{theorem}
The proof of \theoremref{thm-lb} in fact consists of proving the following two
separate lower bounds.


\begin{lemma}\label{lem:lower_k}
  Any distributed algorithm for \sfic with finite approximation ratio
  has time complexity $\Omega(k/\log n)$. This is true even for
  unweighted graphs of diameter 3.
\end{lemma}

\begin{lemma}\label{lem:lower_s}
  Any distributed algorithm for \sfic or \sfcr with finite approximation ratio
  has running time $\Omega(s/\log n)$ for $s\in
  \BO(\sqrt{n})$. This holds  even for instances with $t=2$, $k=1$,
  and $D\in \BO(\log n)$.
\end{lemma}

We remark that the proofs of Lemmas~\ref{lem-lb1} and  \ref{lem:lower_k},
are by reductions from Set Disjointness
\cite{KushilevitzN-book}.
In Lemmas~\ref{lem-lb1} and~\ref{lem:lower_k}, it is trivial to increase the
other parameters, i.e., $D$, $s$, $t$, or $n$, so we may apply
Lemmas~\ref{lemma:transform_to_input} and~\ref{lemma:transform_to_minimal} to
obtain a minimal instance of \sfic without affecting the asymptotic time
complexity.

\section{Deterministic Algorithm}
\label{sec-alg1}
In this section we describe our deterministic algorithm. We start by
reviewing the moat growing algorithm of \cite{AgrawalKR-95}, and then
adapt it to the \Congest model.

\paragraph{Basic Moat Growing Algorithm}
(pseudocode  
 in \algref{algo:centralized}). The algorithm proceeds by 
``moat growing'' and ``moat merging.'' 
A \emph{moat} of radius $r$ around a terminal
$v$ is a set that contains all
nodes and edges within distance $r$ from $v$, where edges may be
included fractionally: for example, if the only edge
incident with $v$ has 
weight $3$, then the moat of radius $2$ around $v$ contains $v$ and the 
$2/3$ of the edge closest to $v$. 
\emph{Moat growing} is a process in which multiple moats
increase their radii at the same rate. 

The algorithm proceeds as follows. All terminals, in parallel, grow moats around
them until two moats intersect. When this happens, 
\begin{inparaenum}[(1)]
\item moat growth is temporarily suspended,
\item the edges of a shortest path connecting two terminals in the meeting moats
are output (discarding edges that close cycles), and
\item the meeting  moats are contracted into a single node. 
\end{inparaenum}
This is called a \emph{merge step} or simply \emph{merge}. Then moat growing
resumes, where the newly formed node is considered an active terminal if some
input component is contained partially (not wholly) in the contracted region,
and otherwise the new node is treated like a regular (non-terminal) node. If the
new node is an active terminal, it resumes the moat-growing with initial radius
$0$. The algorithm terminates when no active terminals remain.

Formal details and analysis are provided
in \appendixref{app-basic}. The bottom line is as follows.

\begin{theorem}\label{theorem:2approx}
  \algref{algo:centralized} outputs a $2$-approximate Steiner forest.
\end{theorem}

\paragraph{Rounded Moat Radii.}
To reduce the
number of times the moat growing is suspended due to moats meeting, 
we defer moat merging  to the next integer power of
$1+\varepsilon/2$, where $\varepsilon$ is a given
parameter. Pseudo-code is given in 
\algref{algo:central_approx} in the Appendix.
Obviously, the 
number of distinct radii in which merges may occur in this algorithm
is now
bounded by $\BO(\log_{1+\varepsilon/2}\WD)
\subseteq\BO(\log n/\varepsilon)$ by our assumption that all edge weights, and hence
the weighted diameter, are bounded by a polynomial in
$n$. Furthermore, approximation deteriorates only a little, as the
following result  states (proof in
\appref{app-epsilon}). 

\begin{theorem}\label{theorem:2+eps_approx}
\algref{algo:central_approx} outputs a $(2+\varepsilon)$-approximate Steiner
forest.
\end{theorem}

\subsection{Distributed Moat-Growing Algorithm}
\label{ssec-dist}
Our goal in this section is to derive a distributed implementation of the
centralized \algref{algo:centralized}. To do this, it is sufficient to follow
the order in which moats merge in the sequential algorithm. The first main
challenge we tackle is to achieve pipelining for the merges that do not change
the activity status of terminals; since all active moats grow at the same rate,
we can compute the merge order simply by finding the distances between moats and
ordering them in increasing order. When the active status of some terminal
changes, we recompute the distances.

We start by defining \emph{merge phases}. Intuitively, a merge phase is a
maximal subsequence of merges in which no active terminal turns inactive and no
inactive terminal is merged with an active one.

\begin{definition}
\label{def:merge}
Consider a run of \algref{algo:centralized}, and let 
  $i_1,\ldots,i_{j_{\max}}$ be the values of $i$ in which $\act_{i_j+1}(v)\neq
  \act_{i_j}(v)$ for some $v\in T$, where $i_0=0$. Steps
  $i_{j-1}+1,\ldots,i_j$ are called \emph{merge phase $j$}, and we denote
  $\act^{(j)}(v):=\act_{i_{j-1}+1}(v)$, i.e., node $v$'s activity status
  throughout merge phase $j$. We use $j(i):=\min \{j\in
  \{1,\ldots,j_{\max}\}\mid i_j\geq i\}$ to denote the \emph{phase of merge
  $i$}.
\end{definition}
\begin{lemma}\label{lem-numphases}
  The number of merge phases is at most $2k$. 
\end{lemma}
Next, we define \emph{reduced weights}, formalizing moat contraction.
We use the following notation.

\smallskip\noindent\textbf{Notation.} For a terminal $v$ and merge step $i$, 
$B_i(v)= B_G(v,\moat_{i}(v))$.

\begin{definition}
  Given merge phase $j$ of \algref{algo:centralized}, define the
  \emph{reduced weight} of an edge $e$ by $\rW_j(e)=W(e)-W(e\cap
  \bigcup_{v\in T}B_{i_{j-1}}(v))$, where fractionally contained edges lose
  weight accordingly.
\end{definition}
Note that $\hat W_j$ is determined by the state of the moats just
before phase $j$ starts. We now define the Voronoi decomposition for phase $j$.
\begin{definition}
  Let $G=(V,E,W)$ be a graph with non-negative edge weights, and let
  $C=\Set{c_1,\ldots,c_k}$ be a set of nodes called \emph{centers},
  with positive distances
  between any two centers. The \emph{Voronoi decomposition} of $G$ w.r.t.\ $C$
  is a partition of the nodes and edges into $k$ subsets called \emph{Voronoi
  regions}, where region $i$ contains all nodes and all edge parts
  whose closest center is $c_i$ (ties broken lexicographically).
\end{definition}

In each phase $j$, we consider the Voronoi decomposition using reduced
weights $\hat W_j$ and active terminals
as centers. 
Let $\vor_j(v)$ denote
the Voronoi region of a node $v$ under this decomposition.
Since we need to consider inactive moats too,  the concept we actually
use is the following. 
\begin{definition}
  The \emph{region} of a
  terminal $v$ in phase $j$, denoted $\reg_j(v)$, is defined as
  follows. 
  $\reg_0(v):=B_0(v)=\{v\}$, and for $j>0$,
$$
\reg_j(v):=\reg_{j-1}(v)\cup
\begin{cases}
  \emptyset~,&\text{if } \neg\act^{(j)}(v)\\
  B_{i_j}(v)\cap\left(\vor_j(v)\setminus \bigcup_{u\in
      T}B_{i_{j-1}}(u)\right)~,&\text{if }\act^{(j)}(v)
\end{cases}
$$
The $j^{th}$ terminal decomposition is given by a collection of
shortest-path-trees spanning, for each $v\in T$, $\reg_j(v)$. We require that
the tree of $\reg_{j}(v)$ extends the tree of $\reg_{j-1}(v)$.
\end{definition}

In other words, $\reg_j(v)$ is obtained from $\reg_{j-1}(v)$ by growing all
active moats at the same rate, but only into uncovered parts of the graph; this growth
stops at the end of a merge phase. Given the $(j-1)^{st}$ terminal
decomposition, it 
is straightforward to compute $\vor_j$ and the required spanning trees
using the Bellman-Ford algorithm, as the following lemma states.

\begin{lemma}\label{lemma:partition}
  Suppose that each node $u\in V$
 knows the
  following about the $(j-1)^{th}$ terminal decomposition:
  \begin{compactitem}
  \item the node $v\in T$ for which $u\in \reg_{j-1}(v)$;
  \item $\act^{(j)}(v)$;
  \item the parent in the shortest-path-tree spanning $\reg_{j-1}(v)$  (unless
  $u=v$ is the root);
  \item $\Wd(v,u)-\moat_{i_{j-1}}(v)$.
  \end{compactitem}
  Then, in $\BO(s)$ rounds we can compute 
  shortest-path-trees rooted
  at nodes $v\in T$, that extend the given trees and span
  $\reg_{j-1}(v)\cup\left(\vor_j(v)\setminus \bigcup_{w\in
  T}B_{i_{j-1}}(w)\right)$ for active $v$ (trees of inactive terminals
remain unchanged). By the end of the computation, each node
  knows: 
  \begin{compactitem}
  \item the node $v\in T$ in whose tree $u$ participates;
  \item the parent in the shortest-path-tree rooted at $v$ (unless $u=v$
    is the root);
  \item for each edge incident to $u$, the fraction of it contained in
  the tree rooted at $v$;
  \item $\Wd(v,u)-\moat_{i_{j-1}}(v)$.
  \end{compactitem}
\end{lemma}
Note that \lemmaref{lemma:partition} says that we can ``almost'' compute
the $j^{th}$ terminal
decomposition (the $B_{i_j}(v)$ remain unknown).
What justifies the trouble of computing decompositions is the
following key observation.

\begin{lemma}\label{lemma:path_in_region}
  For $i=1,\ldots,i_{\max}$, let $v_i$ and $w_i$ be the 
  terminals whose moats are joined in the $i^{th}$ merge of
  \algref{algo:centralized}. Let $p$ be a shortest path connecting
  them. Then $p\subseteq\reg_{j(i)}(v_i)\cup
  \reg_{j(i)}(w_i)$.
\end{lemma}
\lemmaref{lemma:path_in_region} implies that each merging path is ``witnessed''
by the nodes of the respective edge crossing the boundary between the regions.
By the construction from \lemmaref{lemma:partition}, these nodes will be
able to correctly determine the reduced weight of the path. This
motivates the following 
definition.

\begin{definition}\label{def-induced}
  For each $v\in T$, fix a shortest-paths tree on $\reg_{j_{\max}}(v)$.
  Suppose that $e=\{x,y\}$ is an edge so that $x\in\reg_{j_{\max}}(v)$ and
  $y\in \reg_{j_{\max}}(w)$ for some terminals $v\neq w$. Then $e$
  \emph{induces} the unique path $p_{vew}$ that is the concatenation of the
  shortest path from $v$ to $x$ given by the terminal decomposition with $(x,y)$
  and the path from $y$ to $w$ given by the terminal decomposition.
\end{definition}

Since the witnessing nodes cannot determine locally whether ``their''
path is the next merging path, 
they need to
encapsulate and communicate the salient information about the witnessed path.

\begin{definition}\label{def-cands}
  Suppose that $e=\{x,y\}$ is an edge satisfying $x\in \reg_j(v)$ and
  $y\in\reg_j(w)$ with $v\neq w$, $e\subseteq
  \reg_j(v)\cup \reg_j(w)$, $e\not \subseteq  \reg_{j-1}(v)\cup \reg_{j-1}(w)$, 
   and $\act^{(j)}(x)=\true$. Then $e$ is said to induce a
  \emph{candidate merge} $\left(\{v,w\},j,\hat{W}(p_{vew}\cap
  \reg_j(v)),e\right)$ in phase $j$ with \emph{associated path} $p_{vew}$.
\end{definition}
$\hat{W}(p_{vew}\cap \reg_j(v))$ specifies the
increment of the moat radius of the (active) terminal $v$ before the
respective balls intersect.
To order candidate merges we need the following additional concept.

\begin{definition}\label{def-cand-graph}
The \emph{candidate multigraph} is defined as
$G_c:=(T,E_c)$, where for each candidate merge
$\left(\{v,w\},j,\hat{W}(p_{vew}\cap \reg_j(v)),e\right)$ there is an edge
$\{v,w\}\in E_c$.
\end{definition}

We can now  relate the paths selected by
\algref{algo:centralized} to the candidate merges.

\begin{lemma}\label{lemma:equivalent}
Consider the sequence of candidate merges ordered in
ascending lexicographical order: 
first by phase index, then by reduced weight, and finally break ties by
identifiers.  Discard
each merge that closes a 
cycle (including parallel edges) in $G_c$. Let $F_c\subseteq E_c$ be
the resulting forest in $G_c$. Then union of the paths corresponding
to  $F_c$  is exactly the set $F_{i_{\max}}$
computed by \algref{algo:centralized} (with the same tie-breaking rules).
\end{lemma}

\lemmaref{lemma:equivalent} implies that, similarly to Kruskal's
algorithm, it suffices to scan the candidate
merges in ascending order and filter out cycle-closing edges.
Using the technique introduced for MST \cite{GarayKP-98,KuttenP-98},
the filtering 
procedure can be done concurrently with collecting the merges, achieving full
pipelining effect. 
For later development,
we show a general
statement that allows for multiple merge phases to be handled concurrently and
out-of-order execution of a subset of the merges.
\begin{lemma}\label{lemma:filtering}
  Denote by $E_c^{(j)}$ the subset of candidate merges in phase $j$
  and set $F_c^{(j)}:=E_c^{(j)}\cap F_c$. For a set $F_c'\subseteq
  \bigcup_{j'=1}^{j}F_c^{(j')}$, assume that each node $u\in V$ is
  given a set $E_c(u)$ of candidate merges such that 
  $\bigcup_{j'=1}^j F^{(j')}\setminus F_c'
  \subseteq
  \bigcup_{u\in V}E_c(u)
  \subseteq 
  \bigcup_{j'=1}^j E_c^{(j')}$. Finally, assume
  that for each $u\in V$, each candidate merge in $E_c(u)$ is tagged
  by the connectivity components of its terminals in the subgraph
  $(T,F_c')$ of $G_c$. Then $\bigcup_{j'=1}^j F_c^{(j')}\setminus
  F_c'$ can be made known to all nodes in $\BO(D+|\bigcup_{j'=1}^j
  F_c^{(j')}\setminus F_c'|)$ rounds.
\end{lemma}

When emulating \algref{algo:centralized} distributively, we may
overrun the end of the phase if the causing event occurs
remotely. This may lead to spurious merges, which should be
invalidated later.

\begin{definition}
A \emph{false candidate} is a tuple $(\{v,w\},j,\hat{W},e)$ with $v,w\in T$,
$j\in \N$, $2\hat{W}\in \N_0$, and $e\in E$ that is not a candidate
merge. Candidate merges' order is extended to false candidates in the
natural way.
\end{definition}

Fortunately, false candidates originating from the $j^{th}$ Voronoi
decomposition given by \lemmaref{lemma:partition} will always have larger
weights than candidate merges in phase $j$, since they are 
induced by edges outside $\bigcup_{v\in 
T}\reg_j(v)=\bigcup_{v\in T}B_{i_j}(v)$ (see \lemmaref{lemma:decomp}). This
motivates the following corollary.

\begin{corollary}\label{coro:filtering}
Let $E_c^{(j)}$ denote the set of candidate merges in phase $j$ and set
$F_c^{(j)}:=E_c^{(j)}\cap F_c$. Suppose
$\bigcup_{j'=1}^{j-1}F_c^{(j')}$ is globally known, as well
$\Comp(v)$, for all $v\in T$. If each node $u\in V$ 
is given a set $E_c(u)$ of candidate merges and false candidates so that
$E_c^{(j)}\subseteq \bigcup_{u\in V}E_c(u)$ and each false candidate has larger
weight than all candidate merges in $E_c^{(j)}$, then $F_c^{(j)}$ can be made
globally known in $\BO(D+|F_c^{(j)}|)$ rounds.
\end{corollary}

We can now describe the algorithm (see pseudocode in
\appref{app-distalg}). The algorithm
proceeds in merge phases. In each phase, it constructs the
$j^{th}$ terminal decomposition except for knowing the $B_{i_j}(v)$ values
(\lemmaref{lemma:partition}). Using this decomposition, nodes propose
candidate merges, of which some are false candidates. The filtering
procedure from \corollaryref{coro:filtering} is applied to determine
$|F_c^{(j)}|$
. The weight of the last merge
is the increase in moat radii during phase $j$, setting
$B_{i_j}(v)$ and thus $\reg_j(v)$ for each $v\in T$, which allows us
to proceed to the next phase. Finally, the algorithm computes
the minimal subforest of the computed forest,  as in 
\algref{algo:centralized}. We summarize the analysis with the
following statement.

\begin{theorem}\label{theorem:2_distributed}
\sfic can be solved deterministically with approximation factor $2$ in
$\BO(ks+t)$ rounds.
\end{theorem}

\subsection{Achieving a Running Time that is Sublinear in
\texorpdfstring{$t$}{t}}
\label{sec:sublinear}

The additive $t$ term in \theoremref{theorem:2_distributed} can be
avoided. We do this by generalizing a technique first used for MST 
construction~\cite{GarayKP-98,KuttenP-98}. 
Roughly, the idea is to allow moats to grow locally until
they are ``large,'' and then use centralized filtering. A new
threshold that distinguishes ``large'' from ``small'' in this case is $\sqrt{st}$.


\begin{definition}
Define $\sigma=\sqrt{\min\{st,n\}}$. A moat is called \emph{small} if when
formed, its connected component using edges that were selected to the output up
to that point contains fewer than $\sigma$ nodes. A moat which is not small
is called \emph{large}.
%
\end{definition}


To reduce the time complexity, we implement
\algref{algo:central_approx}, where moats change
their ``active'' status only between \emph{growth phases}. In each growth phase,
the maximal moat radius grows by a factor of $1+{\varepsilon/2}$. 
The key insight here is that all we need is to determine at which
moat size the first inactive moat gets merged, because all active terminals
keep growing their moats throughout the
entire growth phase. 

We first slightly
adapt the definition of merge phases.

\begin{definition}
For an execution of \algref{algo:central_approx}, denote by $i_j$,
$j=1,\ldots,j_{\max}$, the merges for which either the if-statement in
\lineref{line:growth_phase} is executed or one of the moats participating in the
merge is inactive. Then the merges $i_j+1,\ldots,i_{j+1}$ constitute the
\emph{$j^{th}$ merge phase}. For $g\in 1,\ldots,g_{\max}$, denote by $j_g$ the
index so that $i_{j_g}$ is the $g^{th}$ merge for which the if-statement in
\lineref{line:growth_phase} is executed. Then the merges
$i_{j_g+1},\ldots,i_{j_{g+1}}$ constitute the \emph{$g^{th}$ growth phase} and
we define that $k_g:=j_{g+1}-j_g$. For convenience, $i_0:=0$ and
$j_0:=0$.
\end{definition}

For constant $\varepsilon$, 
the number of growth phases is in
$\BO(\log n)$ (see \lemmaref{lemma:number_merge_growth_phase}).

\paragraph{Algorithm overview.}
The algorithm is specified in \appref{app:sublinear-code}, except for the final
pruning step, which is discussed
below. The main loop runs over growth phases:
first, regions and terminal decompositions are computed. Then, each
small moat proposes its least-weight candidate merge. 
To avoid long chains of merges, we run a matching algorithm with
small moats as nodes and proposed merges as edges, and then add the
candidate merges proposed by the unmatched small moats. 
After a
logarithmic number of iterations of this procedure, at most $\sigma$
moats remain that may participate in further merges in the growth
phase; the filtering procedure from \lemmaref{lemma:filtering} then
selects the remaining merges in $\BO(\sigma+D)$ rounds. Finally, the
activity status 
for the next growth phase is computed;
small moats are handled by communicating over the edges connecting
them, and large moats rely on pipelining communication over a BFS
tree.


\paragraph{Analysis overview.}
The analysis is given in \appref{app:sublinear-analysis}. We only review the
main points here.  First, \lemmaref{lemma:large_moats} shows that
small moats have strong diameter at most $\sigma$, and that the number
of large moats is bounded by $\sigma$.
We show, in \lemmaref{lemma:growth_3b}, that the set $F_g$ the
algorithm selected by the end of growth phase $g$ is identical to that selected
by an execution of \algref{algo:central_approx} on the same instance of \sfic.
To this end, \lemmaref{lemma:decomp_correct} first shows
that the terminal decompositions are computed correctly in 
 $\BO(sk_g)$ rounds.
 Finally, we prove in \lemmaref{lemma:growth_phase} that the growth
 phase is completed in $\sO(k_gs+\sigma)$ rounds and, if it was not
 the last phase, it provides the necessary information to perform the
 next one.  We summarize the results of this subsection as follows.

\begin{corollary}\label{coro:growth}
For any instance of \sfic, a distributed algorithm can compute a solving forest
$F$ in $\sO(sk+\sigma)$ rounds that satisfies that its minimal subforest solving
the instance is optimal up to factor $2+\varepsilon$.
\end{corollary}

\paragraph{Fast Pruning Algorithm.}
\label{sec:prune}
After computing $F$
, it remains to select the minimal
subforest solving the given instance of problem \sfic: we may have
included merges with non-active moats that need to be pruned.
Simply collecting $F_c$ and $\Comp$ at a single node takes
$\Omega(t)$ rounds, and the depth of (the largest tree in) $F$ can be
$\Omega(st)$ in the worst case. Thus, we employ some of the strategies for
computing $F$ again. First, we grow clusters to size $\sigma$ locally, just like
we did for moats, and then solve a derived instance on the clusters to decide
which of the inter-cluster edges to select. Subsequently, the subtrees inside
clusters have sufficiently small depth to resolve the remaining demands by a
simple pipelining approach. Details are provided in \appref{app:prune}. We
summarize as follows.
\begin{corollary}\label{coro:2+eps_distributed}
For any constant $\varepsilon>0$, a deterministic distributed algorithm can
compute a solution for problem \sfic that is optimal up to factor
$(2+\varepsilon)$ in $\sO(s\min\{k_0,\WD\}+\sqrt{\min\{st,n\}}+k+D)$ rounds,
where $k_0$ is the number of input components with at least two terminals.
\end{corollary}
\
\section{Randomized Algorithm}
\label{sec-alg2}
In \cite{KKMPT-12}, Khan \textit{et al.}\ propose a randomized
algorithm for \sfic that constructs an 
expected $\BO(\log n)$-approximate solution in $\tilde{\BO}(sk)$ time
w.h.p.
In this section we
show how to modify it so as to reduce the running time to
$\sO(k+\min(s,\sqrt n))$ while keeping the approximation ratio
in $\BO(\log n)$.

\paragraph{Overview of the algorithm in \cite{KKMPT-12}.}
The algorithm consists of two main steps. First, a
virtual tree is constructed and embedded in the network, where each
physical node is a virtual leaf.  Then the algorithm selects, for each
input component $\Comp$, the minimal subtree containing all terminals
labeled $\lambda$, and adds, for each virtual edge in these
subtrees, the physical edges of the corresponding path in $G$. Since
the selected set of virtual
edges corresponds to an optimal solution in the tree topology, and
since it can be shown that the expected stretch factor of the embedding is in
$\BO(\log n)$, the result follows.

In more detail, the virtual tree is constructed as follows.  Nodes
pick IDs independently at random.  Each node of the graph is a leaf in
the tree, with ancestors $v_0,\ldots,v_L$, where $L$ the base-2 logarithm of
the weighted diameter (rounded up). The $i^{th}$ ancestor $v_i$ is the
node with the largest ID within distance $\beta2^i$ from $v$, for a
global parameter $\beta$ picked uniformly at random from $[1,2]$. The
weight of the virtual edge $(v_{i-1},v_{i})$ is defined to be
$\beta2^i$.  We note that the embedding in $G$ is via a shortest path
from each node $v$ to each of its $L+1$ ancestors (and not from
$v_{i-1}$ to $v_i$), implemented by ``next hop'' pointers along the
paths.  It is shown that w.h.p., at most $\BO(\log n)$ such distinct
paths pass through any physical node.

Now, consider the second phase. 
Let $T_{\Comp}$, for an input component $\Comp$,  denote the minimal
subtree that contains all terminals of $\Comp$ as leaves.
Clearly, $\bigcup_{\Comp \in \Lambda}T_\Comp$ is the optimal solution
to \sfic on the virtual tree. Thus, all that needs to be done is to select
for each virtual edge in this solution a path in $G$ (of weight smaller or
equal to the virtual tree edge) so that the nodes in $G$ corresponding to the
edge's endpoints get connected. However, since the embedding of the
tree may have paths of $\BO(s)$ hops, and since there are $k$ labels
to worry about,  the straightforward
implementation from~\cite{KKMPT-12} requires $\sO(sk)$ rounds to
select the output edges due to possible
congestion.

\paragraph{Overview of our algorithm.}
Our first idea is to improve  the second phase from~\cite{KKMPT-12} as
follows.  Each internal node $v_i$ is the root of a shortest paths
tree of weighted diameter $\beta2^i$.
For  any virtual tree edge
$\{v_i,v_{i-1}\}\in T_{\Comp}$, we
make sure that exactly one node $v$ in the virtual subtree rooted at $v_{i-1}$
includes the edges of a shortest physical path 
(in $G$) connecting $v$ and $v_i$ in the edge set $F$ output of the
algorithm. 
This is done by $v$ by sending  a message $(\Comp,v_i)$ to $v_i$ up
the shortest paths tree rooted at $v_i$, and these messages are
filtered along the way so that only the first 
$(\Comp,v_i)$ message is forwarded for each $\Comp\in \Lambda$.
This ensures that the only $\BO(s+k)$ steps are needed per
destination. 
Since there are $\BO(\log n)$ such destinations for each node, by
time-multiplexing 
we get running time of $\sO(s+k)$ (w.h.p.). 

When $s>\sqrt{n}$,%
\footnote{%
  W.l.o.g., we present the algorithm as if $s$ was known, because it
  can be determined in $\BO(D+\min\Set{s,\sqrt n})$ rounds as follows:
  Compute $n$ by convergecast, then run Bellman-Ford until
  stabilization or until $\sqrt n$ iterations have elapsed, whichever
  happens first. Since stabilization can be detected $\BO(D)$ time
  after it occurs, we are done.
%
} 
the running time can be improved further to $\sO(\sqrt{n}+k+D)$. The idea is as
follows. Let ${\cal S}$ be the set of the $\sqrt{n}$ nodes of highest rank.
We truncate each leaf-root path in the virtual tree at the first occurrence of a
node from $\cal S$: instead of connecting to that ancestor, the node $v$
connects to the \emph{closest} node from ${\cal S}$. This construction can be
performed in time $\sO(\sqrt{n}+k+D)$ w.h.p.\ (see \appendixref{sec:partial}).
Consider now the edge set $F$ returned by the procedure above: for each input
component $\Comp\in \Lambda$, the terminals labeled $\Comp$ will be partitioned
into connected components, each containing a node from $\cal S$ (if there is a
single connected component it is possible that it does not include any node from
$\cal S$). We view each such connected component as a ``super-terminal'' and
solve the problem by applying an algorithm from \cite{LenzenP13}. The output is
obtained by the set $F$ from the first virtual tree and the additional edges
selected by this algorithm. We show that the overall approximation ratio remains
$\BO(\log n)$ and that the total running time is
$\BO(k+\min\Set{\sqrt{n},s})$.

\paragraph{Detailed description.}
We present the construction for $s\leq \sqrt{n}$ and $s>\sqrt{n}$ in a unified
way. Detailed proofs for the claimed properties are given in
\appendixref{sec:tree_selection}. The first stage consists of the following
steps.
\begin{compactenum}
  \item If $s\leq \sqrt{n}$, set ${\cal S}:=\emptyset$. Otherwise, let ${\cal
  S}$ be the set of $\sqrt{n}$ nodes of highest rank. Delete from the virtual
  tree internal nodes mapped to nodes of $\cal S$. Compute the remaining part of
  the virtual tree and, if ${\cal S}\neq \emptyset$, let each node learn about
  its closest node from ${\cal S}$. In other words, each node 
  $v\notin {\cal S}$ learns the identity of and the shortest paths to
  $v_0,\ldots,v_{i_v-1},\tilde{v}_{i_v}$, where $\tilde{v}_{i_v}$ is
  the node closets to $v$ from $\cal S$. If $v\in {\cal S}$, $i_v=0$ and
  $\tilde{v}_{i_v}=v$.
  \item For each terminal $v\in T$, set $l(v):=\{\Comp(v)\}$. For all other
  terminals, $l(v):=\emptyset$.
  \item For $i\in \{0,\ldots,L\}$ phases:
  \begin{compactenum}
    \item Make for each $\Comp \in \Lambda$ known to all nodes whether it
    satisfies that there is only one terminal $v\in T$ with $\Comp\in\{l(v)\}$.
    If this is the case, delete $\Comp$ from $l(v)$.
    \item Each node $v\in V$ sets
    $\unsent:=\{(\Comp,v_i)\,|\,\Comp \in l(v)\}$ if $i<i_v$ and
    $\unsent:=\{(\Comp,\tilde{v}_{i_v})\,|\,\Comp \in l(v)\}$ otherwise. Then
    all nodes set $\sent:=\emptyset$, $l(v):=\emptyset$, and
    $\hat{l}(v):=\emptyset$.
    \item Repeat until no more messages are sent:
    \begin{compactitem}
      \item For each node $w$, do the following. Each node $v\in V$ for which
      $\unsent\setminus \sent\neq \emptyset$ picks some $(\Comp,w)\in
      \sent\setminus \unsent$ and sets $\sent:=\sent \cup \{(\Comp,w)\}$. If
      $v\neq w$, it sends $(\Comp,w)$ to the next node on the least-weight path
      to $w$ known from the tree construction, otherwise it sets
      $\hat{l}(v):=\hat{l}(v) \cup \{\Comp\}$.
      Each traversed edge is added to $F$.
      \item Each node $v$ that receives a message $(\Comp,w)$ sets
      $\unsent:=\unsent \cup \{(\Comp,w)\}$.
    \end{compactitem}
    \item Each node $w$ with $\hat{l}(w)\neq \emptyset$ selects a node $v$ that
    added, for some $\Comp$, $(\Comp,w)$ to its $\unsent$ variable in Step~3b.
    It sends all entries in its $\hat{l}(w)$ variable to $v$. The node $v$ and
    the routing path to $v$ are determined by backtracing a sequence of messages
    $(\Comp,w)$ from Step 3c. The receiving node $v$ sets $l(v):=\hat{l}(w)$.
  \end{compactenum}
  \item Return $F$.
\end{compactenum}
\emph{The Second Stage.}
If $s\leq \sqrt{n}$, $F$ is the solution.
Otherwise, we construct a new instance and solve it.
%
To define the new instance,
define, for each $v\in {\cal
  S}$, the node set
\begin{equation*}
T_v:=\left\{w\in T\mid\mbox{in $(V,F)$, $v$ is closest to $w$
among nodes from ${\cal S}$ and within $\tilde{\BO}(\sqrt{n})$ hops}\right\},
\end{equation*}
 ties  broken lexicographically. Let $V_r:=V\setminus\bigcup_{v\in\cal S}T_v$. 
The new instance is defined over the following graph.

\begin{definition}
  The
  \emph{$F$-reduced graph} $\hat{G}=(\hat{V},\hat{E},\hat{W})$  is
  defined as
  follows. 
  \begin{compactitem}
  \item $\hat{V}:=\Set{T_v\mid v\in{\cal S}}\cup V_r$
  \item $\hat E
         :=\Set{\{T_u,T_v\}\mid u\in T_u, v\in T_v, \{u,v\}\in E}
         ~\cup~\Set{u,T_v\mid u\in V_r, \{u,v\}\in E\text{ for some }
           v\in T_v}
         ~\cup \\~~~~~~~~~~\Set{\{u,v\}\mid u,v\in V_r}$
  \item $\hat{W}(\hat u,\hat v):=
\begin{cases}\min\Set{W(u,v)\mid u\in T_u, v\in T_v,
      \{u,v\}\in E}& \text{if } \hat u =T_u, \hat v=T_v\\
    \min\Set{W(u,v)\mid u\in T_u,
      \{u,v\}\in E}& \text{if } \hat u =T_u, \hat v\in V_r\\
    W(u,v)&\text{if  } \hat u,\hat v\in V_r
\end{cases}
$
  \end{compactitem}
\end{definition}

To complete the description of the new instance, we specify the new terminals
and labels.
Given an instance of \sfic and the edge set $F$ computed in the first
stage, the \emph{$F$-reduced instance} is defined over the 
$F$-reduced graph $\hat{G}$ as follows. The set of terminals is 
$\hat{T}:=\{T_v\,|\,v\in {\cal S}\wedge T_v\cap T\neq \emptyset\}$. To
construct the labels, define the helper graph $(\Lambda,E_{\Lambda})$,
where 
$$E_{\Lambda}:=\Set{\Set{\Comp,\Comp'}\mid
  \Comp(v)=\Comp,\Comp(u)=\Comp'\text{ for some }v,u\in T_w\text{ for
some }w\in{\cal S}}~.$$  
Now, let
$\hat{\Lambda}$ be the set of connected components of
$(\Lambda,E_{\Lambda})$, identified by  $\BO(\log n)$ bits
each. Finally, the
label $\hat\Comp(T_v)$ of a node $T_v$ in $\hat G$ is the identifier of the
connected component in $(\Lambda,E_{\Lambda})$ of any label $\Comp\in \Lambda$
which belongs to any node in $T_v$ ($\hat{\Comp}(\cdot)$ is well defined,
because all these labels belong to the same connected component of
$(\Lambda,E_{\Lambda})$). 

Since the reduced instance imposes fewer constraints, its optimum is at most
that of the original instance. We show that the reduced instance can be
constructed efficiently, within $\sO(\sqrt{n}+k+D)$ rounds, and then apply the
algorithm from~\cite{LenzenP13} to solve it with approximation factor $\BO(\log
n)$. For this approximation guarantee, the algorithm has time complexity
$\sO(\sqrt{n}+\hat{t}+D)$; since we made sure that the reduced instance has
$\hat{t}=\sqrt{n}$ terminals only, this becomes $\sO(\sqrt{n}+D)$.
The union of the returned edge set with $F$ then yields a solution of the
original instance that is optimal up to factor $\BO(\log n)$. Detailed proofs of
these properties and the following main theorem can be found in
\appendixref{sec:spanner}.
\begin{theorem}\label{theorem:fast}
There is an algorithm that solves \sfic in $\tilde{\BO}(\min\{s,\sqrt{n}\}+k+D)$
rounds within factor $\BO(\log n)$ of the optimum w.h.p.
\end{theorem}

\newpage


{\small

}

\section*{APPENDIX}
\appendix
\section{Preliminaries}
\begin{proof}[Proof of \lemmaref{lemma:transform_to_input}]
  We construct an (unweighted) breadth-first-search (BFS) tree rooted at
an arbitrary node, say the one with the largest identifier. Clearly,
this results in a tree of depth $\BO(D)$ this can be done in $\BO(D)$
rounds. For the first transformation, each node sends all connection
requests it initially knows or receives from its children and that do
not close cycles in $T$ to the root. Since any forest on $T$ has at
most $t-1$ edges, this takes at most $\BO(t+D)$ rounds using messages
of size $\BO(\log n)$. Subsequently, the remaining set of requests at
the root is broadcasted over the BFS tree to all nodes, also in time
$\BO(t+D)$. By transitivity of connectivity, a set $F$ is feasible in
the original instance iff it is feasible w.r.t.\ the remaining set of
connectivity requests. Since these are now global knowledge, the nodes
can locally compute the induced connectivity components (on the set of
terminals) and and unique labels for them: say, the smallest ID in the
component. Setting the label of
terminal $v$ to the label of its connectivity component, the resulting
instance with input components is equivalent as well. 
\end{proof}
\begin{proof}[Proof of \lemmaref{lemma:transform_to_minimal}]
As for the previous lemma, we construct a BFS tree rooted at some
node. Each terminal sends the message $(v,\Comp(v))$ to its parent in the
BFS tree. For each label $\Comp$, if a node ever learns about two
different messages $(v,\Comp)$, $(w,\Comp)$, it sends $(\true,\Comp)$ to
its parent and ignores all future messages with label $\Comp$. All other
messages are forwarded to the parent. Since for each label $\Comp$, no
node sends more than $2$ messages, this step completes in $\BO(D+k)$
rounds. Afterwards, for each $\Comp$ with $|C_{\Comp}|>1$, the root
has either received a message $(\true,\Comp)$, or it has received two messages
$(v,\Comp)$, $(w,\Comp)$, or it has received one message $(v,\Comp)$ and is
in input component $C_{\Comp}$ itself. On the other hand, if
$|C_{\Comp}|=1$, clearly none of these cases applies. Therfore, the
root can determine the subset of labels $\{\Comp\in
\Lambda\,|\,|C_{\Comp}|>1\}$ and broadcast it over the BFS tree,
taking another $\BO(D+k)$ rounds. The minimal instance is then
obtained by all terminals in singleton input components deleting their
label. 
\end{proof}

\section{Lower Bounds}
\label{app-lb}
\begin{proof}[Proof of \lemmaref{lem-lb1}]
Let $\cA$ be a distributed algorithm for \sfcr with approximation
ratio $\rho<\infty$. 
We reduce Set Disjointness (SD) to $\rho$-approximate \sfcr as follows.
Let $A,B\subseteq[n]$ be an instance of SD. Alice, who knows $A$,
constructs the following graph: the nodes are the set
$\Set{a_i}_{i=1}^n$ and two additional nodes denoted $a_0$ and
$a_{-1}$. 
All nodes corresponding to elements in $A$ are connected to $a_0$ and
all nodes corresponding to $[n]\setminus A$ are connected to
$a_{-1}$. Formally, define $E_A=\Set{(a_0,a_i)\mid i\in
  A}\cup\Set{(a_{-1},a_i)\mid i\notin A}$.
Similarly, Bob  constructs nodes $\Set{b_i}_{i=-1}^n$ and edges 
$E_B=\Set{(b_0,b_i)\mid i\in B}\cup\Set{(b_{-1},b_i)\mid i\notin B}$.
In addition to the edges $E_A$ and $E_B$, the graph contains the edges
$E_{AB}=\Set{(a_0,b_0),(a_{-1},b_{-1}),(a_0,b_{-1}),(a_{-1},b_0)}$. 
All edges, except $\Set{(a_0,b_{0}),(a_{-1},b_{-1})}$ have unit cost,
and the edges $\Set{(a_0,b_{0}),(a_{-1},b_{-1})}$ have cost
$W:=\rho(2 n+2)+1$.
This concludes the description of the graph (see \figureref{fig-lb} left). 
Finally, we define the
connection requests 
as follows: for each $i\in A$ we introduce the connection request
$b_i\in R_{a_i}$, and similarly for each $i\in B$ we introduce the
request $a_i\in R_{b_i}$. Note that we have $t\le n$ and $k\le 2$.

\begin{figure}[t]
  \centering
  \includegraphics[height=2.2in]{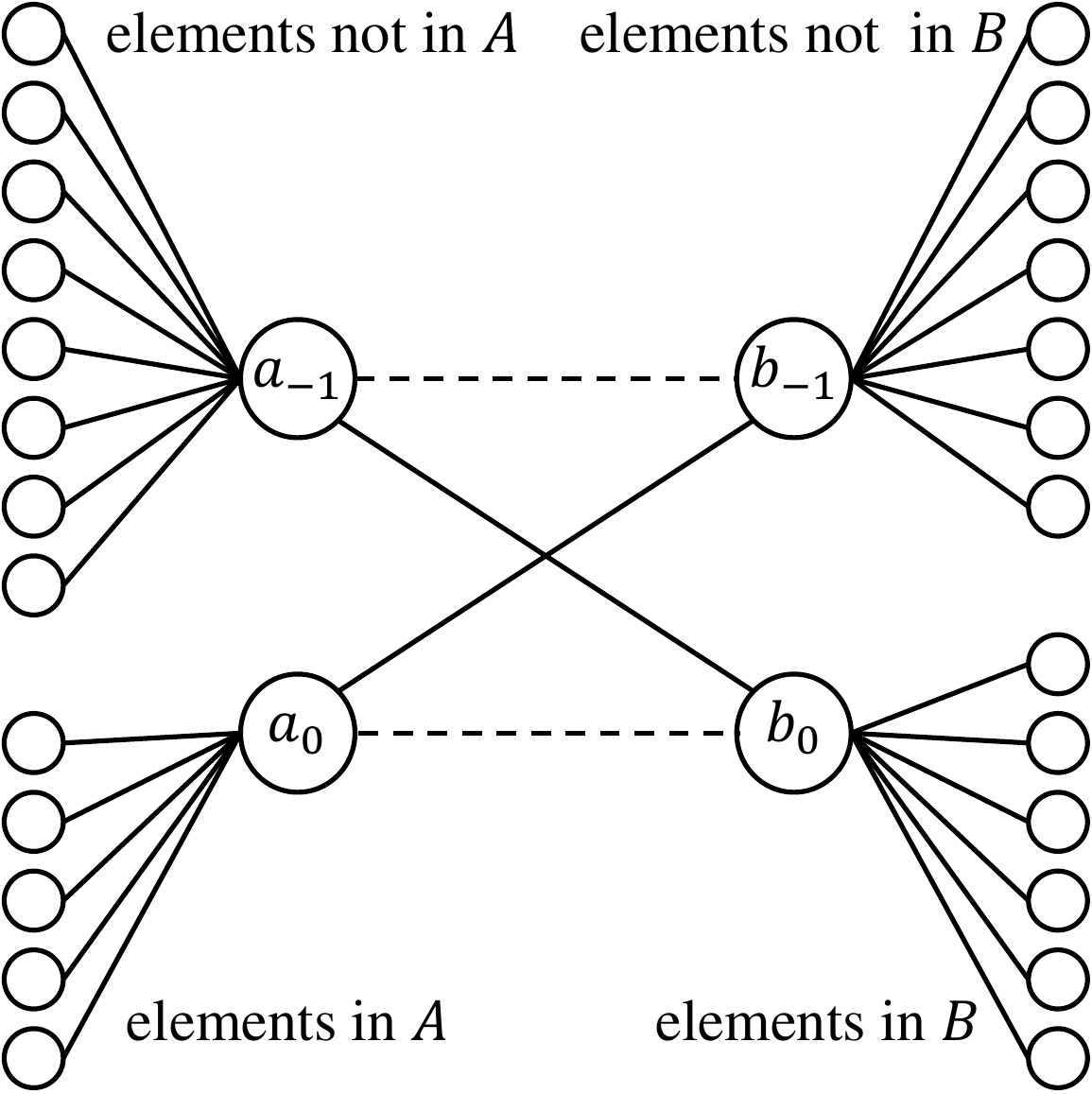}\hspace*{15mm}
  \includegraphics[height=2.2in]{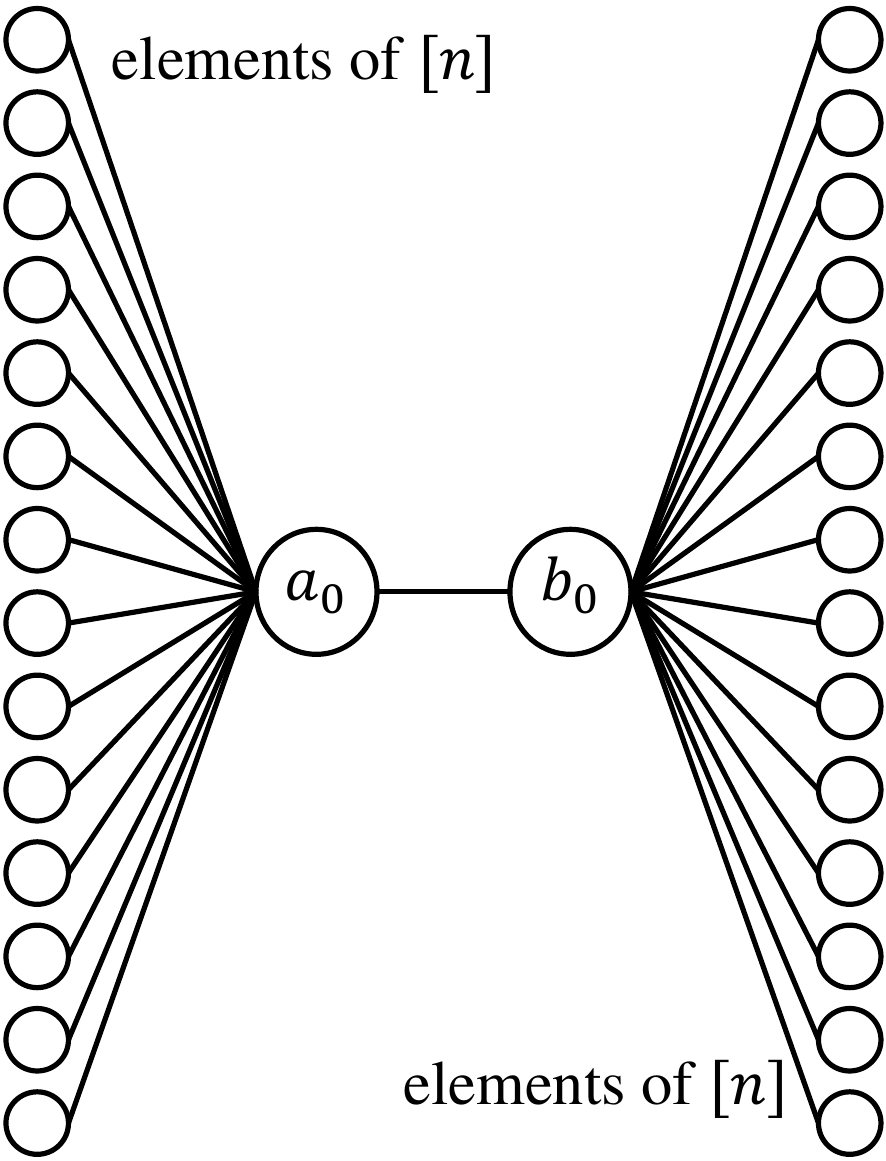}
  \caption{\small Reductions of Set Disjointness to Distributed Steiner
    Forest. Left: reduction to \sfcr (solid edges are light, dashed
    edges are heavy). Right: reduction to \sfic (all edges have unit weight).} 
  \label{fig-lb}
\end{figure}
This completes the description of the \sfcr instance. We now claim
that if $\cA$ computes a $\rho$-approxima\-tion to \sfcr, then we can output
the answer ``YES'' to
the original SD instance iff $\cA$ produces an output that does not
include neither of the heavy edges
$\Set{(a_0,b_{-1}),(a_{-1},b_0)}$. To see this, consider the optimal
solutions. If $A\cap B=\emptyset$, then all connection requests can be
satisfied using edges from $E_A\cup
E_B\cup\Set{(a_0,b_{-1}),(a_{-1},b_0)}$. Hence the optimal cost is at most
$2n+2$, which means that any $\rho$-approximate solution cannot include
a heavy edge; and if $A\cap B\ne\emptyset$, then any solution must
include at least one of the
heavy edges, and hence its weight is larger than $\rho(2n+2)$.

It follows that if $\cA$ is a $\rho$-approximate solution to \sfcr,
then the following algorithm solves SD: Alice and Bob construct the
graph based on their local input without any communication. Then Alice
simulates $\cA$ on the $\Set{a_i}$ nodes  and Bob simulates $\cA$ on
the $\Set{b_i}$ nodes. The only communication required between Alice
and Bob to run the simulation is the messages that cross the edges in
$E_{AB}$. Now,  solving SD requires exchanging $\Omega(n)$ bits
in the worst case (see,
e.g., \cite{KushilevitzN-book}). In the $\Congest(\ell)$ model,
at most $\BO(\ell)$ bits can cross $E_{AB}$ in a round, and hence it must
be the case that the running time of $\cA$ is in
$\Omega(n/\ell)\subseteq\Omega(t/\ell)$.
\end{proof}

\noindent\textbf{Remarks.}
\\$\bullet$  In the lower bound,  $n$ is a parameter describing the universe
  size  of the input
  to SD. Let $n'$ denote the number of nodes in the  corresponding
  instance of \sfcr. Note that we can set $n'$ to any number larger
  than $2n+2$ just by adding 
  isolated nodes. Similarly we can extend the diameter to any number
  larger than $3$ so long as it's smaller than $n'-2n+1$ by attaching
  a chain of $n'-(2n+2)$ nodes to $a_1$. Finally, we can also extend
  $k$ to any number larger than $2$ by adding pairs of nodes
  $\Set{(c_i,c_i')}$, each pair connected  by an edge, and have
  $R_{c_i}=\Set{c_i'}$.
\\$\bullet$  Since $D$ is a trivial
lower bound, we may apply
\lemmaref{lemma:transform_to_input} to convert any \sfcr instance with
$k\geq 2$ into an \sfic instance without losing worst-case performance
w.r.t.\ the 
set of the considered parameters. (If we are
guaranteed that $k=1$, the transformation is trivial, as all terminals
are to be connected.) 
\\$\bullet$  We note that in the hard instances of SD, $|A|,|B|\in\Theta(n)$
  and $|A\cap B|\le1$.
\\$\bullet$  The hardness result applies to \sfcr algorithms that do not require
  symmetric requests. More specifically, if the \sfcr algorithm works only
  for inputs satisfying  $\forall u,v(u\in R_v \iff v\in R_u)$, then the reduction from
  SD fails.
\\$\bullet$  The special case of MST ($t=n$ and $k=1$) can be solved
in time $\sO(\sqrt n+D)$ \cite{KuttenP-98}.

\begin{proof}[Proof of \lemmaref{lem:lower_k}]
As in \lemmaref{lem-lb1}, we reduce Set Disjointness (SD) to \sfic.
Specifically, 
the reduction is as follows.
Let $A,B$ be the input sets
to Alice and Bob, respectively, where $|A|,|B|\subseteq[n]$.
Alice constructs a star whose leaves are the nodes
$\Set{a_i}_{i=1}^n$, all connected to a center node $a_0$ (see
\figureref{fig-lb} right).  
 For each node $a_i$ Alice sets $\lambda(a_i)=i$ if $i\in A$ and
$\lambda(a_i)=\bot$ otherwise. Similarly Bob
constructs another star whose leaves are $\Set{b_i}_{i=1}^n$, all
connected to the center node $b_0$, 
and sets $\lambda(b_i)=i$ if $i\in B$ and $\lambda(b_i)=\bot$
otherwise. In addition the instance to \sfic contains the edge
$(a_0,b_0)$. All edges have unit weight.
Note that using \sfic terminology, we have that the number of input
components satisfies $k\le n$.

We now claim that given any $\rho$-approximation algorithm $\cA$ for \sfic,
the following algorithm solves SD: Alice and Bob construct the graph
(without any communication), and then they simulate 
$\cA$, where Alice simulates all the $\Set{a_i}$ nodes and Bob simulates
all the $\Set{b_i}$ nodes. The answer to SD is YES iff the
edge $(a_0,b_0)$ is not in the output of $\cA$.
To show the algorithm correct, consider two cases.
If the SD instance is  a NO instance, then there exists some $i\in
A\cap B$, which implies, by construction, that $a_i$ and $b_i$ must be
connected by the output edges, and, in particular, the edge  $(a_0,b_0)$ must 
be in the output of $\cal A$ (otherwise $\cal A$ did not produce a
valid output); and if the SD
instance was   a 
NO instance, then the optimal solution to the constructed \sfic instance
contains no edges, i.e., its weight is $0$, and therefore no
finite-approximation 
algorithm may include any edge, and in particular the edge $(a_0,b_0)$, in its output. This
establishes the correctness of the reduction. 

Finally, we note that 
the simulation of $\cA$ requires communicating only the messages
that are sent over the edge $(a,b)$. Since, as mentioned above, any
algorithm for SD 
requires communicating $\Omega(n)$ bits between Alice and Bob, we conclude that if $\cA$
guarantees finite approximation ratio, the number of bits  it must
communicate over $(a_i,b_i)$ is in $\Omega(n)\subseteq\Omega(k)$, and
since in the 
$\Congest(\ell)$ model only $\BO(\ell)$ bits can be communicated over
a single edge in each round,
it must be the case that the running time of $\cA$ is in
$\Omega(k/\ell)$.
\end{proof}

\begin{proof}[Proof of \lemmaref{lem:lower_s}]
Follows from the observation that the shortest $s$-$t$ path is a
special case of the Steiner Forest problem where $s$ and $t$ are the
only two terminals, belonging to the same component. Therefore the
lower bound of 
\cite{DHKNPPW-11} on distributed algorithms solving the shortest $s$-$t$
path problem applies.
\end{proof}

\section{Basic Moat Growing Algorithm}
\label{app-basic}
\begin{algorithm}[H]\small
\caption{Centralized Moat-Growing.}\label{algo:centralized}
\SetKwInOut{Input}{input}
\SetKwInOut{Output}{output}
\Input{$\forall v\in V: \Comp(v)\in \Lambda\cup \{\bot\}$\hfill // input
components}
\Output{feasible forest $F\subseteq E$\hfill // $2$-approximation}
$\M_1:=\{\{v\}\,|\,v\in T\}$ \hfill// moats partition $T$; for
$v\in T$, let $M_i(v)\in \M_i$ s.t.\ $v\in M_i(v)$\\
\For{each $v\in T$}{
  $\moat_0(v):=0$\hfill // by how much moats grew while $v$'s moat was active\\
  $\Comp_1(\{v\}):=\lambda(v)$ \hfill // input components are merged when
  moats merge\\
  $\act_1(\{v\}):=\true$ \hfill // satisfied components' moats become inactive
}
$F_0:=\emptyset$ \hfill // set of selected edges\\
$i:=0$\\
\While{$\exists M\in \M_{i+1}: \act_{i+1}(M)=\true$}{
  $i:=i+1$\\
  $\mu':=\min_{\mu\in \R^+_0}\exists v,w\in T:$\\
  $~~\act_i(M_i(v))=\act_i(M_i(w))=\true \wedge
  \Wd(v,w)=\moat_{i-1}(v)+\moat_{i-1}(w)+2\mu$\\
  $\mu'':=\min_{\mu\in \R^+_0}\exists v,w\in T:$\\
  $~~\act_i(M_i(v))\neq\act_i(M_i(w))=\false \wedge
  \Wd(v,w)=\moat_{i-1}(v)+\moat_{i-1}(w)+\mu$\\
  $\mu_i:=\min\{\mu',\mu''\}$ \hfill // minimal moat growth so
  that two moats touch\\
  \For{each $u\in T$ with $\act_i(M_i(u))=\true$}{
    $\moat_i(u):=\moat_{i-1}(u)+\mu_i$\hfill // grow moats
  }
  Denote by $v_i,w_i\in T$ a pair of terminals giving rise to
  $\mu_i$\nllabel{line:select_v_i_w_i}\\
  Let $E_p$ be the edge set of a least-weight path from $v_i$ to $w_i$
  (drop edges in cycles with $F_i$)\\
  $F_i:=F_{i-1}\cup E_p$\hfill // connect $M_i(v)$ and $M_i(w)$\\
  $\M_{i+1}:=\M_i\cup \{M_i(v_i)\cup M_i(w_i)\}\setminus
  \{M_i(v_i),M_i(w_i)$\}\hfill // merge moats\\
  \For{each $M\in \M_{i+1}$}{
    \If{$M=M_i(v_i)\cup M_i(w_i)$}{
      $\Comp_{i+1}(M):=\Comp_i(M_i(v_i))$
    }
    \ElseIf{$\Comp_i(M)=\Comp(M_i(w_i))$}{
      $\Comp_{i+1}(M):=\Comp_i(M_i(v_i))$\hfill // merge input components (if
      different)
    }
    \Else{
      $\Comp_{i+1}(M):=\Comp_i(M)$
    }
  }
  \If{$\{M\in \M_{i+1}\,|\,\Comp_{i+1}(M)=\Comp_i(M_i(v_i))\}=\{M_i(v_i)\cup
  M_i(w_i)\}$}{
    $\act_{i+1}(M_i(v_i)\cup M_i(w_i)):=\false$ \hfill// new moat's
    component connected by $F_i$
  }
  \Else{
    $\act_{i+1}(M_i(v_i)\cup M_i(w_i)):=\true$
  }
  \For{$M\in \M_{i+1}\setminus \{M_i(v_i)\cup M_i(w_i)\}$}{
    $\act_{i+1}(M):=\act_i(M)$
  }
}
\Return{minimal feasible subset of $F_i$}\hfill // may have selected useless
paths
\end{algorithm}

\begin{definition}[Merges]
Each iteration of the while-loop of \algref{algo:centralized} is
called  a
\emph{merge step}, or simply a \emph{merge}.  The total number of
merges is denoted $i_{\max}$. The number of active moats
during the $i^{th}$ merge  is denoted $\act_i$, i.e., $\act_i:=|\{M\in
\M_i\,|\,\act_i(M_i)=\true\}|$.
\end{definition}

\begin{lemma}\label{lemma:components}
For $i\in \{0,\ldots,i_{\max}\}$, the set $F_i$ computed by
\algref{algo:centralized} is an inclusion-minimal forest such that each $M\in
\M_{i+1}$ is the cut of $T$ with a component of $(V,F_i)$.
\end{lemma}
\begin{proof} 
We show the claim by induction on $i$. We have that $\M_1=\{\{v\}\,|\,v\in T\}$
and $F_0=\emptyset$, i.e., the claim holds for $i=0$. Now assume that it holds
for $i\in \{0,\ldots,i_{\max}\}$ and consider index $i+1$. The choice of
$F_{i+1}\setminus F_i$ guarantees that the joint moat
$M_{i+1}(v_{i+1})\cup M_{i+1}(w_{i+1})$ is subset of the same connectivity
component of $(V,F_{i+1})$. To see that no terminal from $T\setminus
(M_{i+1}(v_{i+1})\cup M_{i+1}(w_{i+1}))$ is connected to this component by
$F_{i+1}$, observe that a least-weight path from $v_{i+1}$ to $w_{i+1}$ contains
no terminal from $T\setminus M_{i+1}(v_{i+1})\cup M_{i+1}(w_{i+1})$ (otherwise
it is not of least weight or $\mu_{i+1}$ would not have been minimal). By the
induction hypothesis, this implies that $M_{i+1}(v_{i+1})\cup M_{i+1}(w_{i+1})$
is a maximal subset of $T$ that is in the same component $(V,F_{i+1})$.

It remains to show that $F_{i+1}$ is an inclusion-minimal forest with this
property. Since $F_{i+1}\setminus F_i$ closes no cycles, it follows from the
induction hypothesis that $F_{i+1}$ is a forest. From this and the
inclusion-minimality of $F_i$ it follows that deleting any edge from $F_i$ will
disconnect a pair of terminals in the same moat. Similarly, removing an edge
from $F_{i+1}\setminus F_i$ will disconnect the new moat $M_{i+1}(v_{i+1})\cup
M_{i+1}(w_{i+1})$.
\end{proof}

\begin{lemma}\label{lemma:feasible}
The output $F$ of \algref{algo:centralized} is a feasible forest.
\end{lemma}
\begin{proof}
By \lemmaref{lemma:components} and the fact that the algorithm terminates once
all moats are inactive, it is sufficient to show that an inactive moat contains
only complete input components.

Note that if the algorithm changes component identifiers, it does so by changing
them for all moats $M\in \M_i$ with $\Comp_i(M)=\lambda$ into some
$\Comp_{i+1}(M)=\lambda'$. Hence all terminals $v\in T$ which initially shared
the same value $\lambda(v)$ are always in moats with identical component
identifiers. Since initially for each $\Comp \in \Lambda$ there are at least two
distinct terminals $v,w\in T$ with $\Comp(v)=\Comp(w)$, for each $\Comp$
initially there are at least two moats $M\in \M_1$ with $\Comp_1(M)=\Comp$. A
merge between moats $M,M'\in \M_i$ assigns component identifier $\Comp_i(M)$ to
all moats with identifier $\Comp_i(M)$ or $\Comp_i(M')$. The merged moat (which
is a connectivity component of $(T,E_i)$) becomes inactive if and only if it is
the only remaining moat with label $\Comp_i(M)$. The statement of the lemma
follows.
\end{proof}

\begin{lemma}\label{lemma:cost} 
For any feasible output $F$, \algref{algo:centralized} satisfies that
\begin{equation*}
W(F)\geq \sum_{i=1}^{i_{\max}}\act_i\mu_i.
\end{equation*}
\end{lemma}
\begin{proof}
We show the statement by induction on $i_{\max}$. The statement is trivial for
$i_{\max}=0$ (i.e., no input components), so suppose it holds for
$i_{\max}\in \N_0$ and consider $i_{\max}+1$. We split up the weight function
$W$ into $W_1+W_2$ so that $W_1(F)\geq \act_1\mu_1$ and define a modified
instance to which we can apply the induction hypothesis, proving that $W_2(F)\geq
\sum_{i=2}^{i_{\max}+1}\act_i\mu_i$.

For each $e\in E$, define $W_1$ to be $W$ within $\bigcup_{v\in T}
B_G(v,\mu_1)$ and $0$ outside (boundary edges have the appropriate fraction of
their weight) and $W_2:=W-W_1$. Consider the edge set $F_C$ of a connectivity
component $C\subseteq T$ induced by $F$. We claim that if it contains $n_C\geq
2$ nodes, it must hold that $W_1(F_C)\geq n_C\mu_1$. To see this, note that the
choice of $\mu_1$ guarantees that the $B_{\mu_1}(v)$ are disjoint for all
$v\in T$. Moreover, by definition, any path connecting $v\in T$ to a node
outside $B_{\mu_1}(v)$ must contain edges of weight at least $\mu_1$ within
$B_{\mu_1}(v)$. The claim follows. Summing over all connectivity components
$C\subseteq T$ induced by $F$ (which satisfy $n_C\geq 2$ since by the problem
definition each terminal must be connected to at least one other terminal), we
infer that $W_1(F)\geq |T|\mu_1=\act_1\mu_1$.

Recall that $\M_1=\{\{v\}\,|\,v\in T\}$. We take the following steps:
\begin{compactitem}
  \item The algorithm replaces the moats $\{v_1\}$ and $\{w_1\}$ by the
  joint moat $\{v_1,w_1\}$. For the purpose of our induction, we simply
  interpret this as setting $T':=T\setminus \{w_1\}$ if the resulting moat is
  active.
  \item If the merge connected the only two terminals $v_1$ and $w_1$ sharing
  the same component identifier, the respective moat becomes inactive. In this
  case, we also remove $v_1$ from $T$, i.e., $T':=T\setminus \{v_1,w_1\}$.
  \item The algorithm assigns to all moats $M\in \M_1$ with
  $\Comp_1(M)=\Comp_1(M_1(w_1))$ the component identifier $\Comp(v_1)$, i.e.,
  $\Comp_2(M):=\Comp_1(M_1(v_1))$. Analogously, we set $\Comp'(v):=\Comp(v)$
  for all $v\in T'\setminus \{v\in T\,|\,\Comp(v)=\Comp(w_1)\}$ and
  $\lambda'(v):=\lambda(w_1)$ for $v\in T'\cap \{v\in
  T\,|\,\Comp(v)=\Comp(w_1)\}$.
  \item Note that the previous steps guarantee that for each terminal $v\in T'$,
  there is a terminal $v\neq w\in T'$ so that $\Comp'(v)=\Comp'(w)$.
  \item The new instance of the problem is now given by the graph
  $G'=(V,E,W_2)$, the terminal set $T'$, and the terminal component function
  $\Comp'$.
\end{compactitem}
Consider an execution of \algref{algo:centralized} on the new instance. We make
the following observations:
\begin{compactitem}
  \item For each $v\in T$ and any radius $r\in \R^+_0$, it holds that
  $B_{G'}(v,r)=B_G(v,r+\mu_1)$. 
  \item Since $B_{G'}(v_1,r)=B_{G'}(w_1,r)$ (as their distance in $G'$ is $0$),
  deleting $w_1$ from the set of terminals has the same effect as joining them
  into one moat.
  \item Hence, if the merged moat $\{v_1,w_1\}$ remains active and thus $v_1$ is
  part of the set of terminals of the new instance, we get a one-to-one
  correspondence between merges of the two instances, i.e., it holds that
  $\act_{i+1}=\act_i'$ and $\mu_{i+1}=\mu_i'$ for all $i\in
  \{1,\ldots,i_{\max}\}$ (where $'$ indicates values for the new instance).
  \item By the induction hypothesis, this implies that
  \begin{equation*}
  W_2(F)\geq \sum_{i=1}^{i_{\max}}\act_i'\mu_i'=
  \sum_{i=2}^{i_{\max}+1}\act_i\mu_i.
  \end{equation*}
  \item If $\{v_1,w_1\}$ became inactive, but never participates in a merge, the
  same arguments apply.
\end{compactitem}
Hence, suppose that $\{v_1,w_1\}\in \M_{i_0}$ participates in a merge in step
$i_0$. For all indices $i<i_0-1$, the above correspondence holds. Moreover,
since $\{v_1,w_1\}$ is inactive, (i) the moat $M\in \M_{i_0}$ with which it is
merged must satisfy that $\act_{i_0}(M)=\true$ and (ii) we have that
$\act_{i_0+1}(M\cup \{v_1,w_1\})$, i.e., the resulting moat is active (as
$\Comp_{i_0}(\{v_1,w_1\})=\Comp_{i_0}(M')$ for any $M'\in \M_{i_0}$ would
contradict the fact that $\{v_1,w_1\}$ is inactive). Thus, the merge does not
affect the number of active moats, i.e., $\act_{i_0+1}=\act_{i_0}$.
Furthermore, it holds that $\moat_{i_0}(v_1)=\moat_{i_0}(w_1)=\mu_1$, since
$\{v_1,w_1\}$ has been active only during merge $1$. We conclude that, for any
$r\in \R^+_0$,
\begin{equation*}
\bigcup_{v\in M\cup \{v_1,w_1\}}B_G(v,\moat_{i_0+1}(v)+r)
= \bigcup_{v\in M}B_{G'}(v,\moat_{i_0-1}'(v)+r),
\end{equation*}
as the moats of size $\mu_1$ around $v_1$ and $w_1$ at the end of the $i^{th}$
merge exactly compensate for the fact that the edges inside the respective
weighted balls in $G$ have no weight in $G'$. By induction on $i\in
\{i_0+1,\ldots,i_{\max}+1\}$, it follows that, for any $r\in \R^+_0$,
\begin{equation*}
\bigcup_{v\in M\cup \{v_1,w_1\}}B_G(v,\moat_i'(v)+r)
= \bigcup_{v\in M}B_{G'}(v,\moat_{i-2}'(v)+r),
\end{equation*}
and we can map the following merges of the two runs onto each other, i.e.,
$\mu_{i_0}+\mu_{i_0+1}=\mu_{i_0-1}'$ and, for $i\in
\{i_0,\ldots,i_{\max}-1\}$, $\mu_{i+2}=\mu_i'$ as well as
$\act_{i+2}=\act_i$. In particular,
\begin{equation*}
\act_{i_0}\mu_{i_0}+\act_{i_0+1}\mu_{i_0+1}
=\act_{i_0}(\mu_{i_0}+\mu_{i_0+1})=\act_{i_0-1}'\mu_{i_0-1}',
\end{equation*}
and the induction hypothesis yields that
\begin{equation*}
W_2(F)\geq \sum_{i=1}^{i_{\max}-1}\act_i'\mu_i'=
\sum_{i=2}^{i_{\max}+1}\act_i\mu_i.
\end{equation*}
Hence, in both cases $W(F)=W_1(F)+W_2(F)\geq
\sum_{i=1}^{i_{\max}+1}\act_i\mu_i$, and the proof is complete.
\end{proof}

\begin{proof}[Proof of \theoremref{theorem:2approx}]
By \lemmaref{lemma:feasible}, the output $F$ of the algorithm is a feasible
forest. With each merge, the algorithm adds the edges of a path of cost
$\moat_i(v_i)+\moat_i(w_i)$ to $F$. Hence
\begin{equation*}
W(F)\leq  \sum_{i=1}^{i_{\max}}\moat_i(v_i)+\moat_i(w_i)\\
=  \sum_{i=1}^{i_{\max}}\left(\sum_{\substack{j=1\\ \act^{(j)}(M_j(v_j))=\true}}^i
\mu_j +\sum_{\substack{j=1\\ \act^{(j)}(M_j(w_j))=\true}}^i \mu_j\right).
\end{equation*}

We construct $G'$ and $F'$ from $(V,F\cup F_{i-1})$ by contracting edges in
$B_G(v,\sum_{j=1}^{i-1}\mu_j)$ for all $v\in T$. If edges are ``partially
contracted'' since they are only fractionally part of
$B_G(v,\sum_{j=1}^{i-1}\mu_j)$ for some $v\in T$, their weight simply is
reduced accordingly; note that since $F$ is a forest, no edges are ``merged'',
i.e., the resulting weights are well-defined. By \lemmaref{lemma:components},
this process identifies for each moat $M\in \M_{i-1}$ its terminals. Note that
the edges from $F_{i-1}$ are completely contained in these balls. We interpret
the set of active moats $T'=\{M\in \M_i\,|\,\act_i(M)=\true\}$ (which after
contraction are singletons) as the set of terminals in $G'$. Since $F$ is
minimal w.r.t.\ satisfying all constraints, so is $F'$ (where in $G'$ two
terminals need to be connected if the corresponding moats contain terminals
that need to be connected). As only active moats contain terminals with
unsatisfied constraints (cf.~\lemmaref{lemma:feasible}), $F'$ is the union of at
most $|T'|-1=\act_i-1$ shortest paths between terminal pairs from $T'$ that
contain no other terminals.

Now consider the balls $B_{G'}(v,\mu_i)$ around nodes $v\in T'$. By the choice
of $\mu_i$, they are disjoint. For each such ball $B_{G'}(v,\mu_i)$, by
definition any least-weight path has edges of weight at most $\mu_i$ within the
ball. We claim that any path in $F'$ that connects nodes $v,w\in T'$, but
contains no third node $u\in T'\setminus \{v,w\}$, does not pass through
$B_{G'}(u,\mu_i)$ for any $u\in T'$. Otherwise, consider the subpath from $v$
to a node in $B_{G'}(u,\mu_i)$ for some $u\in T'\setminus \{v,w\}$ and
concatenate a shortest path from its endpoint to $u$. The result is a path
from $v$ to $u$ that smaller weight than the original path from $v$ to $w$.
Symmetrically, there is a path shorter than the one from $v$ to $w$ connecting
$w$ and $u$. However, together with the fact that the algorithm connects moats
incrementally using least-weight paths of ascending weight implies that the
pairs $\{v,u\}$ and $\{w,u\}$ must end up in the same moat \emph{before} the
path connecting $v$ and $w$ is added. By transitivity of connectivity this
necessitates that $v$ and $w$ are in the same moat when a path connecting them
is added, a contradiction. We conclude that indeed each of the considered paths
passes through the balls around its endpoints only.

Overall, we obtain that in the above double summation, for each index $i$, there
are at most $2(\act_i-1)$ summands of $\mu_i$: $2$ for each of the at
most $\act_i-1$ paths connecting nodes in $T'$ considered in the previous
paragraph (note that the contraction did not change weights of edges covered by
these summands). We conclude that
\begin{equation*}
W(F) \leq  \sum_{i=1}^{i_{\max}}\left(\sum_{\substack{j=1\\
\act^{(j)}(M_j(v_j))=\true}}^i \mu_j +\sum_{\substack{j=1\\
\act^{(j)}(M_j(w_j))=\true}}^i \mu_j\right)\\
\leq \sum_{i=1}^{i_{\max}}2(\act_i-1)\mu_i\\
< 2\sum_{i=1}^{i_{\max}}\act_i\mu_i.
\end{equation*}
By \lemmaref{lemma:cost}, this is at most twice the cost of any feasible
solution. In particular, the cost of $F$ is smaller than twice that of an
optimal solution.
\end{proof}

\section{Rounded Moat Radii}
\label{app-epsilon}

\begin{algorithm}[p!]\small
\caption{Centralized Approximate Moat-Growing with
approximation ratio $(2+\varepsilon)$.}\label{algo:central_approx}
\SetKwInOut{Input}{input}
\SetKwInOut{Output}{output}
\Input{$\forall v\in V: \Comp(v)\in \Lambda\cup \{\bot\}$\hfill // input
components}
\Output{feasible forest $F\subseteq E$\hfill // $2$-approximation}
$\M_1:=\{\{v\}\,|\,v\in T\}$ \hfill// moats partition $T$; for
$v\in T$, let $M_i(v)\in \M_i$ s.t.\ $v\in M_i(v)$\\
\For{each $v\in T$}{
  $\moat_0(v):=0$\hfill // by how much moats grew while $v$'s moat was active\\
  $\Comp_1(\{v\}):=\lambda(v)$ \hfill // input components are merged when
  moats merge\\
  $\act_1(\{v\}):=\true$ \hfill // satisfied components' moats become inactive
}
$F_0:=\emptyset$ \hfill // set of selected edges\\
$i:=0$\\
$\hat{\mu}:=1$\\
\While{$\exists M\in \M_{i+1}: \act_{i+1}(M)=\true$}{
  $i:=i+1$\\
  $\mu':=\min_{\mu\in \R^+_0}\exists v,w\in T:$\\
  $~~\act_i(M_i(v))=\act_i(M_i(w))=\true \wedge
  \Wd(v,w)=\moat_{i-1}(v)+\moat_{i-1}(w)+2\mu$\\
  $\mu'':=\min_{\mu\in \R^+_0}\exists v,w\in T:$\\
  $~~\act_i(M_i(v))\neq\act_i(M_i(w))=\false \wedge
  \Wd(v,w)=\moat_{i-1}(v)+\moat_{i-1}(w)+\mu$\\
  $\mu_i:=\min\{\mu',\mu''\}$ \hfill // minimal moat growth so
  that two moats touch\\
  \If{$\sum_{j=1}^i \mu_j\geq \hat{\mu}$\nllabel{line:growth_phase}}{
    $\mu_i:=\hat{\mu}-\sum_{j=1}^{i-1} \mu_j$ \hfill // stop moat growth at
    $\hat{\mu}$\\
    $F_i:=F_{i-1}$\hfill // no merge, just checking whether moats are active\\
    $\M_{i+1}:=\M_i$\\
    \For{each $M\in \M_i$}{
      $\Comp_{i+1}(M):=\Comp_i(M)$\\
      \If{$\{M'\in \M_i\,|\,\Comp_i(M')=\Comp_i(M)\}=\{M\}$}{
        $\act_{i+1}(M):=\false$ \hfill// moat's terminals satisfied
      }
      \Else{
        $\act_{i+1}(M):=\true$
      }
    }
    $\hat{\mu}:=(1+\varepsilon/2)\hat{\mu}$\hfill // threshold for next check
  }
  \Else{
    Denote by $v_i,w_i\in T$ a pair of terminals giving rise to $\mu_i$\\
    Let $E_p$ be the edges of a least-weight path from $v_i$ to $w_i$
    (drop edges in cycles with $F_i$)\nllabel{line:select_path}\\
    $F_i:=F_{i-1}\cup E_p$\hfill // connect $M_i(v)$ and $M_i(w)$\\
    $\M_{i+1}:=\M_i\cup \{M_i(v_i)\cup M_i(w_i)\}\setminus
    \{M_i(v_i),M_i(w_i)$\}\hfill // merge moats\\
    $\Comp_{i+1}(M_i(v_i)\cup M_i(w_i)):=\Comp_i(M_i(v_i))$\\
    $\act_{i+1}(M_i(v_i)\cup M_i(w_i)):=\true$\\
    \For{each $M\in \M_{i+1}\setminus \{M_i(v_i)\cup M_i(w_i)\}$}{
      \If{$\Comp_i(M)=\Comp(M_i(w_i))$}{
        $\Comp_{i+1}(M):=\Comp_i(M_i(v_i))$\hfill // merge input components (if
        different)
      }
      \Else{
        $\Comp_{i+1}(M):=\Comp_i(M)$
      }
      $\act_{i+1}(M):=\act_i(M)$
    }
  }
  \For{each $u\in T$ with $\act_i(M_i(u))=\true$}{
    $\moat_i(u):=\moat_{i-1}(u)+\mu_i$\hfill // grow moats
  }
}
\Return{minimal feasible subset of $F_i$}\hfill // may have selected useless
paths
\end{algorithm}
\clearpage

\begin{corollary}\label{coro:cost_approx}
For any solution $F$, \algref{algo:central_approx} satisfies that
\begin{equation*}
\left(1+\frac{\varepsilon}{2}\right)W(F)\geq \sum_{i=1}^{i_{\max}}\act_i\mu_i,
\end{equation*}
where $i_{\max}$ is the final iteration of the while-loop of the algorithm.
\end{corollary}
\begin{proof}
Denote by $u_i$ the number of unsatisfied moats in the $i^{th}$ iteration of the
while-loop of \algref{algo:central_approx}, i.e., the moats which can terminals
that need to be connected to terminals in different moats. Analogously to
\lemmaref{lemma:cost}, we have that
\begin{equation*}
W(F)\geq \sum_{i=1}^{i_{\max}}u_i\mu_i.
\end{equation*}
Now consider a satisfied moat $M_i\in \M_i$ that is formed in iteration $i-1$
out of two unsatisfied moats; we call such a moat bad. Denote by $j(M_i)\geq i$
the first iteration in which a moat $\bar{M}\supseteq M$ is unsatisfied or
inactive, whichever happens earlier. Since the minimal edge weight is $1$ and
$\hat{\mu}$ is increased by factor $1+\varepsilon/2$ whenever the algorithm
checks whether to inactivate moats, it holds that $\sum_{k=i}^{j(M_i)-1}
\mu_k\leq \varepsilon/2\cdot \sum_{k=1}^{i-1} \mu_k$. As an unsatisfied moat
can only be created by merging an unsatisfied moat (with a satisfied or
unsatisfied moat), there is a sequence of unsatisfied moats $M_0\subseteq
M_1\subseteq\ldots\subseteq M_{i-1}$ such that $M_{i-1}\subset M_i$.

We observe that if we pick a different moat $M'$ and merge $i'$ as above and
apply the same construction, the resulting sequence
$M_0'\subseteq\ldots\subseteq M_{i'-1}'$ must be disjoint from the sequence
$M_0\subseteq\ldots\subseteq M_{i-1}$, since for each $j\in
\{1,\ldots,i_{\max}\}$, the set of moats $\M_j$ forms a partition of $T$ and
the sequences contain no unsatisfied moats. We conclude that
\begin{equation*}
\sum_{i=1}^{i_{\max}}\act_i\mu_i \leq \sum_{i=1}^{i_{\max}}u_i\mu_i
+ \sum_{i=1}^{i_{\max}}\sum_{\substack{M_i\in \M_i\\ M_i \mbox{
bad}}}\sum_{k=i}^{j(M_i)-1}\mu_k
\leq \left(1+\frac{\varepsilon}{2}\right)\sum_{i=1}^{i_{\max}}u_i\mu_i
\leq \left(1+\frac{\varepsilon}{2}\right)W(F).
\end{equation*}
\end{proof}

\begin{proof}[Proof of \theoremref{theorem:2+eps_approx}]
Analogous to \theoremref{theorem:2approx}, except that the final bound on the
approximation ratio follows from \corollaryref{coro:cost_approx}.
\end{proof}

\section{\texorpdfstring{Proofs for \sectionref{ssec-dist}}{Proofs Concerning
the Distributed Moat Growing Algorithm}}
\label{app-dist}

\begin{proof}[Proof of \lemmaref{lem-numphases}]
  Clearly, the total number of times moats become inactive is at
  most $k$, because every input component becomes completely contained in a moat
  exactly once throughout the execution. When  an inactive moat
  merges, either all its 
  terminals become active again or a new inactive moat is formed. Hence, the
  total number of merges for which the activity status of some terminals change
  is at most $2k$.
\end{proof}

\begin{proof}[Proof of \lemmaref{lemma:partition}]
To compute the Voronoi decomposition in phase $j$, we use the single-source
Bellman-Ford algorithm, where active moats are sources.
All nodes in active moats are initialized with distance $0$, and the edge
weights are given by the reduced weight function $\hat{W}_j$ (which is known
locally, because the moat size is locally known). Messages are tagged by the
identifier of the closest source w.r.t.\ $\hat{W}_j$ (the ``old'' trees are not
touched, but simply extended). In $\BO(s)$ rounds, the Bellman-Ford algorithm
terminates, and the result is that the shortest paths trees are extended to
include all nodes in the respective Voronoi regions $\vor_j$ that are not in
$\reg_{j-1}(v)$ for a terminal $v\in T$ with $\act^{(j)}=\true$, and each node
knows its distance from the closest moat according to $\hat W_j$, i.e.,
$\Wd(v,u)-\moat_{i_{j-1}}(v)$. Finally, observe that nodes in $\reg_{j-1}(v)$
for some $v\in T$ with $\act^{(j)}(v)=\false$ simply can use the information
from the previous phase $j-1$.
\end{proof}

\begin{lemma}\label{lemma:decomp}
For each $j\in \{0,\ldots,j_{\max}\}$, it holds that $\bigcup_{v\in
T}\reg_j(v)=\bigcup_{v\in T}B_{i_j}(v)$.
\end{lemma}
\begin{proof}
We prove the statement by induction on $j$; it trivially holds for $j=0$, so
consider the induction step from $j-1$ to $j$. For any node (or part of an edge)
in $\bigcup_{v\in T}\reg_{j-1}(v)=\bigcup_{v\in T}B_G(v,\moat_{i_{j-1}}(v))$,
the statement trivially holds by the induction hypothesis. Hence, suppose a node
(or part of an edge) is outside $\bigcup_{v\in T}\reg_{j-1}(v)$ and
consider the least-weight path $p$ that leads to $\bigcup_{v\in T,
\act^{(j)}(v)=\true}\reg_{j-1}(v)$ (for simplicity, suppose it contains no
fractional edges; the general case follows by subdividing edges into lines).
Suppose $v\in T$ is the terminal in whose region $\reg_{j-1}(v)$ the path ends.
Then, by the definition of reduced weights and $\vor_j(v)$, the path is
contained in $\vor_j(v)\setminus \bigcup_{v\in T,\act^{(j)}(v)=\true}\reg_{j-1}(v)$.
Hence, if $W(p)\leq \moat_{i_j}(v)-\moat_{i_{j-1}}(v)$, i.e., the node (or part
of an edge) is contained in $B_{i_j}(v)$, it must be in $\bigcup_{v\in
T}\reg_j(v)$. The choice of $v$ implies that $p\subseteq B_{i_j}(v)$ is
equivalent to $p\subseteq \bigcup_{v\in T,\act^{(j)}(v)=\true}B_{i_j}(v)$. Because
the node (or part of an edge) is outside $\bigcup_{v\in
T}\reg_{j-1}(v)=\bigcup_{v\in T}\reg_j(v)$, this is equivalent to the node (or
part of an edge) being in $\bigcup_{v\in T}B_{i_j}(v)$. We
conclude that $\bigcup_{v\in T}\reg_j(v)=\bigcup_{v\in T}B_{i_j}(v)$, i.e., the
induction step succeeds.
\end{proof}

\begin{proof}[Proof of \lemmaref{lemma:path_in_region}]
Since $p$ is a least-weight path, $W(p)=\Wd(v_i,w_i)$. By the definition of
$\mu_i$, hence $W(p)=\moat_i(v_i)+\moat_i(w_i)$. By \lemmaref{lemma:decomp},
$\bigcup_{v\in T}\reg_{j(i)}(v)=\bigcup_{v\in T}B_{i_{j(i)}}(v))$.
Thus, any path $q$ between to terminals that enters the uncovered region in
phase $j(i)$ must have weight $W(q)>W(p)$; in particular, $p$ cannot enter the
uncovered region.

Hence, assume for contradiction that $p$ enters $\reg_{j(i)}(u)$ for some $u\in
T$. Denote by $p'$ a minimal prefix of $p$ ending at node $x\in \reg_{j(i)}(u)$
for some $u\in T$. We make a case distinction, where the first case is that
$u\in M_i(v_i)$. Consider the concatenation $q_1$ of the suffix of $p$ starting
at $x$ to a least-weight path from $u$ to $x$. By the definition of regions, we
have that
\begin{eqnarray*}
W(q_1)-\moat_{i_{j(i)-1}}(w_i)-\moat_{i_{j(i)-1}}(u) &=& W_{j(i)}(q_1)\\
&<& W_{j(i)}(p)\\
&=& W(p)-\moat_{i_{j(i)-1}}(w_i)-\moat_{i_{j(i)-1}}(v_i).
\end{eqnarray*}
By assumption $u$ and $v_i$ are in the same moat after merge $i-1$, which must
have been active. By the definition of merge phases, $u$ and $v_i$ thus were
both in active moats during all merges $i_{j(i)-1}+1,\ldots,i$. This entails
that their $\moat$ variables have been increased by the same value in each of
these merges, yielding that
\begin{equation*}
W(q_1)-\moat_{i-1}(w_i)-\moat_{i-1}(u)<W(p)-\moat_{i-1}(w_i)-\moat_{i-1}(v_i).
\end{equation*}
As $p$ is a least-weight path from $v_i$ to $w_i$, we conclude that
\begin{equation*}
\Wd(u,w_i)-\moat_{i-1}(w_i)-\moat_{i-1}(u)<\Wd(v_i,w_i)-\moat_{i-1}(w_i)-\moat_{i-1}(v_i).
\end{equation*}
This contradicts the minimality of $\mu_i$, since $u$ is in an active moat in
merge $i$.

Hence it must hold $u\notin M_i(v_i)$, which is the second case. Consider
the path $q_2$ which is the concatenation of a least-weight path between $x$ and
$u$ to $p'$. Similarly to the first case, we have that
\begin{equation*}
W(q_2)-\moat_{i_{j(i)-1}}(v_i)-\moat_{i_{j(i)-1}}(u) 
<W(p)-\moat_{i_{j(i)-1}}(v_i)-\moat_{i_{j(i)-1}}(w_i).
\end{equation*}
If $M_i(u)$ is active, $u$ is in active moats during merges $i\in
\{i_{j(i)-1}+1,\ldots,i\}$, and similarly to the first case we can infer that
\begin{equation*}
\Wd(v_i,u)-\moat_{i-1}(v_i)-\moat_{i-1}(u)<\Wd(v_i,w_i)-\moat_{i-1}(v_i)-\moat_{i-1}(w_i);
\end{equation*}
the same applies if $\act_i(M_i(w_i))=\false$. Again this contradicts the
minimality of $\mu_i$, as $M_i(v_i)$ is active. 

It remains to consider the possibility that $\act_i(M_i(u))=\false$ and
$\act_i(M_i(w_i))=\true$. Symmetrically to the first case, we can exclude that
$u\in M_i(w_i)$. Since $u$ is in inactive moats during phase $j(i)$, it holds that
$\moat_{i-1}(u)=\moat_{i_{j(i)-1}}(u)$. By definition of $q_1$ and $q_2$, we
thus have that
\begin{equation*}
\Wd(v_i,u)+\Wd(w_i,u)-2\moat_{i-1}(u)\leq W(q_1)+W(q_2)-2\moat_{i_{j(i)-1}}(u)
<W(p)=\Wd(v_i,w_i).
\end{equation*}
As $W(p)=\moat_i(v_i)+\moat_i(w_i)\geq \moat_{i-1}(v_i)+\moat_{i-1}(w_i)$, this
yields
\begin{equation*}
\Wd(v_i,u)+\Wd(w_i,u)-2\moat_{i-1}(u)<
2\Wd(v_i,w_i)-\moat_{i-1}(v_i)-\moat_{i-1}(w_i).
\end{equation*}
By the pidgeon hole principle, we obtain that
\begin{equation*}
\Wd(v_i,u)-\moat_{i-1}(u)-\moat_{i-1}(v_i)<
\Wd(v_i,w_i)-\moat_{i-1}(v_i)-\moat_{i-1}(w_i)
\end{equation*}
or that
\begin{equation*}
\Wd(w_i,u)-\moat_{i-1}(u)-\moat_{i-1}(w_i)<
\Wd(v_i,w_i)-\moat_{i-1}(v_i)-\moat_{i-1}(w_i).
\end{equation*}
As both $\act_i(M_i(v_i))=\act_i(M_i(w_i))=\true$ and $u\notin M_i(v_i)\cup
M_i(w_i)$, this contradicts the minimality of $\mu_i$. We conclude that all
cases lead to contradiction and therefore the claim of the lemma is true.
\end{proof}
\begin{proof}[Proof of \lemmaref{lemma:filtering}]
To specify the execution of \algref{algo:centralized}, the following symmetry
breaking rule is introduced: Among all feasible combinations of choices for
$v_i$ and $w_i$ in \lineref{line:select_v_i_w_i}, and paths $p$ in
\lineref{line:select_path}, the algorithm selects the path $p_{v_iew_i}$ such
that $\{v_i,w_i\}\cup e$ is minimal w.r.t.\ the order used in point (iii) of
\defref{def-cand-graph}.

For the respective execution, we show the claim by induction on the merges $i$.
We anchor the induction at $i=0$, for which $F_0=\emptyset$, which equals the
union of edges in the paths associated with $\emptyset$. Hence, consider merge
$i\in \{1,\ldots,i_{\max}\}$, assuming that the claim holds for the first $i-1$
merges/candidate merges in $F_c$. \lemmaref{lemma:path_in_region} shows that the
least-weight path $p_{v_iew_i}$ from $v_i$ to $w_i$ selected by
\algref{algo:centralized} in merge $i$ satisfies that $p_{v_iew_i}\in
\reg_{j(i)}(v_i)\cup \reg_{j(i)}(w_i)$. Since $\act^{(j(i))}(v_i)=\true$ and
$M_{i_{j(i)}+1}(v_i)\subseteq M_i(v_i)\neq M_i(w_i)\supseteq
M_{i_{j(i)}+1}(w_i)$, $e$ induces candidate merge
$(\{v_i,w_i\},j(i),\hat{W}_{j(i)}(p_{v_iew_i}\cap \reg_{j(i)}(v_i)),e)$.

We claim that this candidate merge is the next element of $F_c$ (according to
the order). Assuming otherwise for contradiction, the symmetry breaking rules
specified above imply that there is a candidate merge
$(\{v,w\},j,\hat{W}_j(p_{ve'w}\cap \reg_j(v)),e')$ which (i) satisfies that
$(j,\hat{W}_j(p_{ve'w}\cap \reg_j(v))<j,\hat{W}_j(p_{v_iew_i}\cap \reg_j(v))$
(lexicographically), (ii) closes no cycle with the first $i-1$ selected merges,
and (iii) satisfies that $\act^{(j)}(v)=\true$. By property (ii) and the
induction hypothesis, $M_i(v)\neq M_i(v)$. If $j'<j(i)$, the candidate merge
must have been selected as element $i'<j(i)\leq i$ into $F_c$, contradicting the
fact that no w.r.t.\ $G_c$ duplicate edges are selected into $F_c$. Therefore,
by (i), $j=j(i)$ and $\hat{W}_j(p_{ve'w}\cap \reg_j(v))<\hat{W}(p_{v_iew_i}\cap
\reg_j(v_i))$. By the definition of regions,\footnote{TODO: A bit of a leap
here, but should not be hard to show by a case distinction. Should be done at
some point\ldots} this implies that $\moat_i(v)+\moat_i(w)>\Wd(v,w)$.
It follows that $v$ and $w$ must satisfy that $M_i(v)=M_i(w)$, since otherwise
\algref{algo:centralized} would merge these moats instead in merge $i$. However,
the induction hypothesis and the facts that $F_c$ closes no cycles and contains
no duplicate edges entail that $M_i(v)\neq M_i(w)$, a contradiction; the claim
follows.

Because the path associated with candidate merge
$(\{v_i,w_i\},j(i),\hat{W}_{j(i)}(p_{v_iew_i}\cap \reg_{j(i)}(v_i)),e)$ is
$p_{v_iew_i}$, the induction hypothesis yields that the edge set of the union
of paths associated with the first $i$ elements of $F_c$ is a superset of $F_i$.
Since $p_{v_iew_i}\setminus \{e\}$ is contained in the shortest-path-trees at
$v_i$ and $w_i$, the respective edges close no cycles with the cut of $F_{i-1}$
with the trees rooted at $v_i$ and $w_i$, respectively. Since $M_i(v_i)\neq
M_i(w_i)$, $e$ does not close a cycle in $F_i$ either. We conclude that
\algref{algo:centralized} adds all edges in $p_{v_iew_i}$ to $F_{i-1}$ when
$p_{v_iew_i}$ does not close a cycle with $F_{i-1}$, implying that
constructing $F_i$. Hence, the the edge set of the union of paths associated
with the first $i$ elements of $F_c$ equals $F_i$, the induction step succeeds,
and the proof is complete.
\end{proof}

\begin{proof}[Sketch of Proof of \lemmaref{lemma:filtering}]
We use the edge elimination procedure introduced for MST
\cite{GarayKP-98,KuttenP-98}, which works as follows. We use an (unweighted) BFS
tree rooted at some node $R\in V$, which can be constructed in $\BO(D)$ rounds.
For round $r\in \N$, let $F_u(r)$ denote the set of candidate merges node $u\in
V$ holds at the end of round $r$, where $F_u(0):=E_c(u)$. In each round each
node executes the following convergecast procedure.
\begin{compactenum}
\item ${F}_u(r-1)$ is scanned in ascending weight order, and a merge that closes
a cycle in $G_c$ with the union of $F_c'$ and previous merges is deleted. (This
is possible because the merges are tagged by the connectivity components of
the terminals they join in $(T,F_c')$.)
\item The least-weight unannounced merge in $F_u(r-1)$ is announced by $u$ to
its parent ($R$ skips this step).
\item ${F}_u(r)$ is assigned the union of $F_u(r-1)$ with all merges received
from children.
\end{compactenum}
Once all sets stabilize (which can be detected at an overhead of $\BO(D)$
rounds), the set $F_R(r)$ equals $\bigcup_{j'=1}^j F_c^{(j')}\setminus F_c'$.
Perfect pipelining is achieved, leading to the stated running time bound.
\end{proof}

\begin{proof}[Proof of \lemmaref{coro:filtering}]
Set $F_c':=\bigcup_{j'=1}^{j-1}F_c^{(j')}$. Each node $u\in V$ locally
computes the connectivity components of $(T,F_c')$ and tags the elements of
$E_c(u)$ accordingly. We apply the same procedure as for
\lemmaref{lemma:filtering}, except that we need to detect termination
differently, as we would like to stop the routine once the root knows
$F_c^{(j)}$. The pipelining guarantees that after $D+i$ rounds of the routine,
the first $i$ elements of the ascending list of merges (whose sublist up to
element $|F_c^{(j)}|$ equals $F_c^{(j)}$) are known to the root. Since the root
knows $F_c'$ and, for each $v\in T$, $\Comp(v)$, it can locally compute the
variables $\act^{(j)}(v)$, $v\in V$, and will detect in round $D+|F_c^{(j)}|$ that
some terminal changes its activity status. This enables to determine when to
terminate the collection routine and which elements of $F_R(D+|F_c^{(j)}|)$
constitute $F_c^{(j)}$.
\end{proof}
We put the pieces of our analysis together to bound the time complexity of our
algorithm.
\begin{lemma}\label{lemma:2_time}
The above algorithm can be implemented such that it runs in $\BO(sk+t)$
rounds.
\end{lemma}
\begin{proof}
Clearly, Step 1 can be executed in $\BO(D)$ rounds. Step 2 consists of local
computations only. By \lemmaref{lemma:equivalent}, we have that
$F_c=F_c^{(j_{\max})}$, since at the end of merge phase $j_{\max}$, no active
terminals remain. We conclude that the loop in Step 3 of the above algorithm is
executed for $j_{\max}$ iterations. By \lemmaref{lem-numphases}, $j_{\max}\leq
2k_0$.

We claim that iteration $j\in \{1,\ldots,j_{\max}\}$ of the loop can be executed
in $\BO(s+|F_c^{(j)}|)$ rounds, which we show by induction on $j$. The induction
hypothesis is that, after $j-1$ iterations of the loop, the prerequisites of
\lemmaref{lemma:partition} are satisfied for index $j-1$,
$\moat^{(j-1)}(v)=\moat_{i_{j-1}}(v)$ for all $v\in T$, and the value of the
variable $\act^{(j)}(v)$ is correct for each $v\in T$. This is trivially satisfied
for $j=1$ by initialization, hence suppose the hypothesis holds for
$j\in\{1,\ldots,j_{\max}-1\}$. Under this assumption, \lemmaref{lemma:partition}
shows that Step 3a can be executed in $\BO(s)$ rounds, in the sense that the
trees become locally known as stated in the lemma. Clearly, this implies that
Step 3b can be executed in one round, by each node $u$ sending $v_u$ to each
neighbor.

Consider $(\{v_u,v_{u'}\},j,\hat{W},\{u,u'\})\in E_c(u)$. We have that
$\act^{(j)}(v_u)=\true$. For each entry, we have that $v_u\neq
v_{u'}$ and $\{u,u'\}\notin \reg_{j-1}(v_u)\cap \reg_{j-1}(v_{u'})$. Thus,
if $\{u,u'\}\in \reg_j(v_u)\cap \reg_j(v_{u'})$, the hypothesis that
$\moat_{i_{j-1}}(v_u)=\moat^{(j-1)}(v_u)$ implies that
\begin{equation*}
\hat{W}=\Wd(v_u,u)-\moat_{i_{j-1}}(v_u)+W(\{u,u'\}\cap T_u)
=W(p_{v_u\{u,u'\}v_{u'}})-\moat_{i_{j-1}}(v_u)=\hat{W}_j(p_{v_u\{u,u'\}v_{u'}})
\end{equation*}
and $(\{v_u,v_{u'}\},j,\hat{W},\{u,u'\})$. Hence,
$E_c^{(j)}\subseteq \bigcup_{u\in V}E_c(u)$ and an entry
$(\{v_u,v_{u'}\},j,\hat{W},\{u,u'\})\in E_c(u)$ is a candidate merge if and only
if $\{u,u'\}\in \reg_j(v_u)\cap \reg_j(v_{u'})$.

As $\act^{(j)}(v_u)=\true$, it holds that that
$\moat_{i_j}(v_u)-\moat_{i_{j-1}}(v_u):=\hat{W}_{\max}$ is identical for all
$v_u\in T$. We have that
\begin{equation*}
W(\{u,u'\}\cap T_u)\subseteq \reg_j(v_u)\Leftrightarrow
\Wd(v_u,u)+W(\{u,u'\}\cap T_u)\leq \moat_{i_j}(v_u)
\Leftrightarrow \hat{W}\leq \hat{W}_{\max}.
\end{equation*}
Similarly, if $\act^{(j)}(v_{u'})=\true$,
\begin{equation*}
W(\{u,u'\}\cap T_{u'})\subseteq \reg_j(v_{u'})\Leftrightarrow
\hat{W}\leq \hat{W}_{\max},
\end{equation*}
because $\moat_{i_{j-1}}(v_{u'})=\moat^{(j-1)}(v_{u'})$. It follows that
\begin{equation*}
W(\{u,u'\}\cap T_u)\subseteq \reg_j(v_u)\Leftrightarrow
W(\{u,u'\}\cap T_{u'})\subseteq \reg_j(v_{u'}).
\end{equation*}
On the other hand, if $\act^{(j)}(v_{u'})=\false$, the statement $W(\{u,u'\}\cap
T_{u'})\subseteq \reg_j(v_{u'})$ is trivially satisfied, because $T_{u'}$ spans
$\reg_j(v_{u'})$. We conclude that $(\{v_u,v_{u'}\},j,\hat{W},\{u,u'\})\in
E_c(u)$ is a candidate merge if and only if $\hat{W}\leq \hat{W}_{\max}$.

Therefore, each false candidate in $\bigcup_{u\in V}E_c(u)$ is of larger weight
than all candidate merges in $E_c^{(j)}$. We conclude that the prerequisites of
\corollaryref{coro:filtering} are satisfied for merge phase $j$, yielding that
Step 3c can be executed in $\BO(D+|F_c^{(j)}|)$ rounds.

Step 3d requires local computation only. We observe that:
\begin{compactitem}
\item For each $v\in T$, $\moat^{(j)}(v)=\moat_{i_j}(v)$, since we established
that $\mu^{(j)}=\hat{W}_{\max}=\moat_{i_j}(v)-\moat_{i_{j-1}}(v)$ for each $v\in
T$ with $\act^{(j)}(v)=\true$.
\item By \lemmaref{lemma:partition}, the local information available to the
nodes from Step 3a and the $\moat^{(j)}$ variables permit to determine, for each
$u\in V$, whether $u\in \reg_j(v_u)$ and the fraction of its incident edges
inside $\reg_j(v_u)$.
\item By \lemmaref{lemma:equivalent}, $\M^{(j+1)}=\M_{i_j+1}$, i.e., the moats
at the beginning of merge phase $j+1$.
\item The computed variables $\act^{(j+1)}(v)$, $v\in T$, are thus correct.
\end{compactitem}
This establishes the induction hypothesis for index $j+1$. The total time
complexity of the $j^{th}$ iteration of the loop in Step 3 is
$\BO(D+s+|F_c^{(j)})|)\subseteq \BO(s+|F_c^{(j)})|)$, yielding a total of
\begin{equation*}
\BO\left(\sum_{j=1}^{j_{\max}}s+|F_c^{(j)})|\right)=
\BO\left(j_{\max}s+|F_c|\right)\subseteq \BO(ks+|T|-1)=\BO(ks+t)
\end{equation*}
rounds to complete Step 3.

Step 4 requires local computations only. For Step 5, for an edge $\{x,y\}$
inducing a candidate merge from $F_{\min}$, $x$ and $y$ send a token to their
respective parents. Each node receiving a token for the first time forwards it,
other tokens will be ignored. Edge $\{x,y\}$ and all edges traversed by a token
are selected into $F$. Since the goal is to select for each edge $\{x,y\}$ the
edge and the paths from $x$ and $y$ to the roots in their respective trees, this
rule ensures that $F$ is computed correctly. Because the shortest-path-trees
have depth at most $s$ and there is no congestion, this implementation of Step 4
completes in $\BO(s)$ rounds (where termination is detected in
$\BO(D)\subseteq \BO(s)$ rounds over the BFS tree). Since Step 6 requires no
communication, summing up the time complexities for Steps 1 to 6 yields a total
running time bound of $\BO(D+t+ks+t+s)=\BO(ks+t)$.
\end{proof}
\begin{proof}[Proof of \theoremref{theorem:2_distributed}]
By \lemmaref{lemma:2_time}, the above algorithm can be executed within the
stated running time bound. By \lemmaref{lemma:equivalent}, the edge set $F'$ of
the union of paths associated with $F_c$ equals the set $F_{i_{\max}}$
computed by some execution of \algref{algo:centralized}. Hence, if we can show
that the set $F$ returned in Step 6 of the above algorithm is the minimal subset
of $F'$ that solves the instance, the theorem readily follows from
\theoremref{theorem:2approx}.

Recall that because \algref{algo:centralized} never closes a cycle,
$F'=F_{i_{\max}}$ is a forest, and so is $F$. By the minimality of $F_{\min}$,
any two terminals connected by $F_{\min}$ (viewed as forest in $G_c$) must be
connected by any subforest of $F'$ that is a solution. For any edge $\{x,y\}\in
F$, there is an element of $(\{v,w\},\cdot,\cdot,\cdot)\in F_{\min}$ such that
$\{x,y\}$ is on the associated path connecting $v$ and $w$. Deleting $\{x,y\}$
from $F$ will disconnect $v$ and $w$ (because $F$ is a forest), implying that
the resulting edge set does not solve the instance of \sfic. We conclude that
$F$ is indeed the edge set returned by \algref{algo:centralized}, and therefore
optimal up to factor $2$.
\end{proof}

\subsection{The distributed algorithm}
\label{app-distalg}
\begin{compactenum}
\item Construct a directed BFS tree, rooted at $R$. For each $v\in T$, broadcast
$(v,\Comp(v))$ to all nodes (via the BFS tree).
\item Set $j:=1$ (index of the merge phase) and $M_1:=\{\{v\}\,|\,v\in T\}$. For
each $v\in T$, set $\moat^{(0)}(v):=0$, $\act_1(v):=\true$, and
$\reg_0(v):=\{v\}$.
\item While $\exists v\in T$ with $\act^{(j)}(v)=\true$:
\begin{compactenum}
\item Compute the collection of shortest-path-trees spanning for each $v\in T$
with $\act^{(j)}=\false$ $\reg_j(v)$ and for each $v\in T$ with $\act^{(j)}=\true$
$\reg_{j-1}(v)\cup (\vor_j(v)\setminus \bigcup_{w\in T}B_{i_{j-1}}(w))$.
\item For each $u\in V$, denote by $v_u$ the root of the tree ${\cal T}_u$ it
participates in. For each $u\in V$ with $\act^{(j)}(v_u)=\true$, locally construct
$E_c(u)$ as follows. For each neighbor $u'$ of $u$ so that $v_{u'}\neq v_u$ and
$\{u,u'\}\notin \reg_{j-1}(v_u)\cup \reg_{j-1}(v_{u'})$, $u$
adds $(\{v_u,v_{u'}\},j,\Wd(v_u,u)-\moat^{(j-1)}(v_u)+W(\{u,u'\}\cap
{\cal T}_u),\{u,u'\})$ to $E_c(u)$. For all other nodes $u$,
$E_c(u):=\emptyset$.
\item Determine $F_c^{(j)}$ and make it known to all nodes.
\item Suppose the maximal merge in $F_c^{(j)}$ is
$(\cdot,\cdot,\mu^{(j)},\cdot)$. Each $u\in V$ locally computes:
\begin{compactitem}
\item for $v\in T$ with $\act^{(j)}(v)=\true$,
$\moat^{(j)}(v):=\moat^{(j-1)}(v)+\mu^{(j)}$; 
\item for $v\in T$ with $\act^{(j)}(v)=\false$,
$\moat^{(j)}(v):=\moat^{(j-1)}(v)$;
\item whether $u\in \reg_j(v_u)$ or not, and the fraction of its
incident edges inside $\reg_j(v_u)$;
\item the set $\M^{(j+1)}$ of connectivity components of the forest on $T$
induced by $\bigcup_{j'=1}^jF_c^{(j)}$ (for $v\in T$, denote by $M_v\in
\M^{(j+1)}$ the moat so that $v\in M_v$);
\item for $v\in T$ with $\exists w\in M_v, w'\in T\setminus M_v:\,\Comp(w)=
\Comp(w')$, $\act^{(j+1)}(v):=\true$;
\item for $v\in T$ with $\nexists w\in M_v, w'\in T\setminus M_v:\,\Comp(w)=
\Comp(w')$, $\act^{(j+1)}(v):=\false$.
\end{compactitem}
\item $j:=j+1$.
\end{compactenum}
\item Set $F_c:=\bigcup_{j'=1}^{j-1}F_c^{(j)}$.
Each node locally computes the minimal subset $F_{\min}\subseteq F_c$ such that
the induced forest on $T$ connects for each $\Comp\in \Lambda$ all terminals
$v\in T$ with $\Comp(v)=\Comp$.
\item $F:=\emptyset$. For each element of $F_{\min}$, suppose $e=\{x,y\}$ is the
inducing edge and $p_{vew}$ the associated path. Add $e$ to $F$ and also all
edges on the paths from $x$ to $w$ and $y$ to $w$ that are given by the
shortest-path-trees spanning $\reg_{j-1}(v)$ and $\reg_{j-1}(w)$,
respectively.
\item Return $F$.
\end{compactenum}

\section{Material for \texorpdfstring{\sectionref{sec:sublinear}}{Section
\ref*{sec:sublinear}}}
\label{app:sublinear}
\subsection{Specification of the Algorithm}
\label{app:sublinear-code}
\paragraph{Specification of the algorithm.}
\begin{compactenum}
\item Construct a directed BFS tree, rooted at $R$.
\item Set $j:=0$ (index of the merge phase), $F:=\emptyset$, and
$\hat{\mu}:=1$. At each $v\in T$, set $\moat^{(0)}(v):=0$, $\act_1(v):=\true$,
$\reg_0(v):=\{v\}$, $M_v:=\{v\}$, and $L(M_v):=v$ (the leader of moat $M_v$).
\item While $\exists v\in T:\, \act^{(j+1)}(v)=\true$:
\begin{compactenum}
\item While $\sum_{j'=1}^{j-1}\mu^{(j')}<\hat{\mu}$:
\begin{compactenum}
\item $j:=j+1$.
\item Compute the shortest-path-trees spanning for each $v\in T$
with $\act^{(j)}=\false$ $\reg_j(v)$ and for other terminals
$\reg_{j-1}(v)\cup (\vor_j(v)\setminus \bigcup_{w\in T}B_{i_{j-1}}(w))$.
\item For each $u\in V$, denote by $v_u$ the root of the tree ${\cal T}_u$ it
participates in. For each $u\in V$ with $\act^{(j)}(v_u)=\true$, check whether there
is a neighbor $u'\in V$ with $\act^{(j)}(v_{u'})=\false$. If so, set
\begin{equation*}
c_u:=\argmin\limits_{\substack{\{u,u'\}\in E\\ \act^{(j)}(v_{u'})=\false}}
\{(\{v_u,v_{u'}\},j,\Wd(v_u,u)-\moat^{(j-1)}(v_u)+W(\{u,u'\}\cap
{\cal T}_u),\{u,u'\})\},
\end{equation*}
i.e., $c_u$ is the least-weight candidate merge with an inactive terminal
induced by an edge incident to $u$. For all other nodes $u$, $c_u:=\bot$.
\item Over the BFS tree, determine
$(\{v_c,w_c\},j,\hat{W},\{u,u'\}):=\argmin_{u\in V}\{c_u\}$ and make it known to
all nodes. If there is no such candidate merge or
$\hat{W}+\sum_{j'=1}^{j-1}\mu^{(j')}>\hat{\mu}$, set
$\mu^{(j)}:=\hat{\mu}-\sum_{j'=1}^{j-1}\mu^{(j')}$. Otherwise,
$\mu^{(j)}:=\hat{W}$. All terminals $v\in T$ with $\act^{(j)}(v)=\true$ set
$\moat^{(j)}(v):=\moat^{(j-1)}(v)+\mu^{(j)}$. Other terminals set
$\moat^{(j)}(v):=\moat^{(j-1)}(v)$. Each terminal $v\in T$ broadcasts
$\moat^{(j)}(v)-\moat^{(j-1)}(v)$ over its current shortest-path-tree. Each node
$u\in V$ determines whether it is in $\reg_j(u_v)$ and the fraction of its
incident edges in $\reg_j(u_v)$.
\item If $\mu^{(j)}=\hat{W}$ (i.e., merge phase $j$ does not end the growth
phase), terminal $w_c$ (i.e., the one with $\act^{(j)}(w_c)=\false$) broadcasts
$L(M_w)$ over the BFS tree. All terminals $v$ with $L(M_v)=L(M_{w_c})$ set
$\act^{(j+1)}(v):=\true$. Each terminal $v\in T$ broadcasts $\act^{(j+1)}(v)$ over
its current shortest-path-tree.
\end{compactenum}
\item For $\lceil\log \sqrt{\min\{t/s,n\}}\rceil$ iterations:
\begin{compactenum}
\item Denote by $\M:=\{\{v\in T\,|\,L(M_v)=L\}\,|\,\exists w\in T:\,L=L(M_w)\}$
the set of current moats. Each small moat $M\in \M$ finds the smallest
candidate merge $(\{v,w\},j',\hat{W},e)$ satisfying that $j'\leq j$, $v\in
M$, and $w\notin M$ (if there is any). Denote the set of such candidate merges
by $F_C$.
\item Interpret $F_C$ as the edge set of a simple graph on the node set $\M$, by
reading each candidate merge $(\{v,w\},j,\hat{W},e)\in F_C$ as an edge
$\{M_v,M_w\}$.\footnote{This is well-defined, since the minimality of edges in
$F_C$ ensures that there can be only one edge between any pair of moats.}
Define $F_C':=\{(\{v,w\},j,\hat{W},e)\in F_C\,|\,\mbox{$M_v$ and $M_w$ are
small}\}$. Determine an inclusion-maximal matching $M\subseteq F_C'\subseteq
F_C$. Each (small) moat that is not incident to an edge in $M$, but added an
edge to $F_C$, adds the respective edge to $M$ again, resulting in a set of
candidate merges $F_+\subseteq F_C$.
\item For each $(\{v,w\},j,\hat{W},e)\in F_+$, add the edges of $p_{vew}$ to
$F$.
\item Denote by ${\cal C}$ the set of connectivity components of $(V,F)$. For
each $v\in T$, set $M_v:=T\cap C_v$, where $C_v\in {\cal C}$ is the component
such that $v\in C$. Each terminal $v\in T$ learns the identifier of $L(M_v)$,
the terminal with largest identifier among all terminals $w$ with $M_w=M_v$.
Each terminal $v\in T$ learns whether $M_v$ is small.
\item For each small $M_v$, make the complete set $M_v$ known to all its
terminals.
\end{compactenum}
\item For each $v\in T$, broadcast $L(M_v)$ to all nodes in $\reg_j(v)$ (over
its shortest-path-tree).
\item For each $u\in V$, locally construct $E_c(u)$ as follows. Starting from
$E_c(u):=\emptyset$, for each $j'\in \{1,\ldots,j\}$ with $\act_{j'}(v_u)=\true$
and each neighbor $u'$ of $u$ so that $v_{u'}\neq v_u$ and $\{u,u'\}\notin
\reg_{j'-1}(v_u)\cup \reg_{j'-1}(v_{u'})$, $u$ adds
$(\{v_u,v_{u'}\},j',\Wd(v_u,u)-\moat^{(j'-1)}(v_u)+W(\{u,u'\}\cap
\reg_{j'}(v_u)),\{u,u'\})$ to $E_c(u)$. Each candidate merge is tagged by the
identifiers of the moat leaders $L(M_{v_u})$ and $L(M_{v_{u'}})$.
\item Denote by $F_c'$ the set of candidate merges whose associated paths' edges
have been added to $F$ so far. Determine $F_+:=\bigcup_{j'=1}^j
F_c^{(j')}\setminus F_c'$.
\item For each $(\{v,w\},j,\hat{W},e)\in F_+$, add the edges of $p_{vew}$ to
$F$.
\item Denote by ${\cal C}$ the set of connectivity components of $(V,F)$. For
each $v\in T$, set $M_v:=T\cap C_v$, where $C_v\in {\cal C}$ is the component
such that $v\in C$. Each terminal $v\in T$ learns the identifier of $L(M_v)$,
the terminal with largest identifier among all terminals $w$ with $M_w=M_v$.
Each terminal $v\in T$ learns whether $M_v$ is small.
\item For each small $M_v$, make the complete set $M_v$ known to all its
terminals.
\item For each $v\in T$, determine whether there are $w\in M_v$ and $u\in
T\setminus M_v$ so that $\Comp(w)=\Comp(u)$. If this is the case, set
$\act^{(j+1)}(v):=\true$, otherwise set $\act^{(j+1)}(v):=\false$.
\end{compactenum}
\item Return $F$.
\end{compactenum}

\subsection{Proofs}
\label{app:sublinear-analysis}
\begin{lemma}\label{lemma:number_merge_growth_phase}
For $\varepsilon\in \BO(1)$ and any execution of \algref{algo:central_approx},
there are at most $\BO(\log n/\varepsilon)$ growth phases and
$\sum_{g=1}^{g_{\max}}k_g\in k+\BO(\log n/\varepsilon)$.
\end{lemma}
\begin{proof}
We claim that $\sum_{i=1}^{i_{\max}}\mu_i\leq \WD/2$. Assuming the
contrary, there must be some active moat $M\in \M_{i_{\max}-1}$. Since the moat
is active, there are terminals $v\in M$ and $w\in T\setminus M$ so that
$\Comp(v)=\Comp(w)$. Clearly, these terminals were not in the same moats after
any merge $i<i_{\max}$ and therefore remain active throughout the entire
execution of the algorithm. It follows that
$\moat_v(i_{\max})=\moat_w(i_{\max})=\sum_{i=1}^{i_{\max}}\mu_i> \WD/2$.
However, by definition $\Wd(v,w)\leq \WD$, implying that
\begin{equation*}
2\mu_{i_{\max}}\leq \Wd(v,w)-\moat_{i_{\max}-1}(v)-\moat_{i_{\max}-1}(w)
\leq \WD-\moat_{i_{\max}-1}(v)-\moat_{i_{\max}-1}(w).
\end{equation*}
Because $\mu_{i_{\max}}=\moat_{i_{\max}}(v)-\moat_{i_{\max}-1}(v)
=\moat_{i_{\max}}(w)-\moat_{i_{\max}-1}(w)$, this yields the contradiction
\begin{equation*}
0\leq \WD-\moat_{i_{\max}}(v)-\moat_{i_{\max}}(w)<0.
\end{equation*}
We conclude that indeed $\sum_{i=1}^{i_{\max}}\mu_i\leq \WD/2$. Since
$\hat{\mu}$ is initialized to $1$ and grows by factor $1+\varepsilon/2$ with
each growth phase, we obtain that the number of growth phases is bounded by
\begin{equation*}
1+\left\lceil \log_{1+\varepsilon/2}\left(\frac{\WD}{2}\right)\right\rceil
\leq 1+\left\lceil\frac{\log \WD}{\log(1+\varepsilon/2)}\right\rceil
\in \BO(\log n /\varepsilon),
\end{equation*}
where the last step exploits that for $\varepsilon\in \BO(1)$,
$\log(1+\varepsilon)\in \Omega(\varepsilon)$. The bound on the number of merge
phases follows from this bound and the definition of the $k_g$, since there are
at most $k$ merges which may result in inactive moats (i.e., input components
become satisfied), each of which can be merged only once.
\end{proof}


\begin{lemma}\label{lemma:large_moats}
At any stage of the above algorithm, the number of large moats is bounded by
$\sigma$ and the connectivity component of $(V,F)$ of a small moat has a hop
diameter of at most $\sigma$.
\end{lemma}
\begin{proof}
The bound on the hop diameter of small moats' components trivially follows from
the fact that they contain at most $\sigma$ nodes.

Suppose $st< n$. We claim that the connectivity component of a moat with $\tau$
terminals contains at most $1+(\tau-1)(s-1)$ nodes. This holds trivially for the
initial moats. Now suppose moats $M$ and $M'$ are merged. The merging path has
at most $s$ hops, implying that at most $s-2$ nodes are added. Hence the new
moat has at most $2+(|M\cap T|+|M'\cap T|-2)(s-1)+(s-2)\leq 1+(|(M\cup M')\cap
T|-1)(s-1)$ nodes. The claim follows. This entails that the total
number of nodes in moats' components is bounded by $st$.

We conclude that there are at most $\sigma^2$ nodes in moats' connectivity
components w.r.t.\ $F$, and therefore at most $\sigma$ large moats.
\end{proof}
\begin{lemma}\label{lemma:decomp_correct}
Suppose that after $g-1\in \{0,\ldots,g_{\max}-1\}$ growth phases, the
variables $\act^{(j_g+1)}(v)$, $\moat^{(j_g)}(v)=\moat_{i_{j_g}}(v)$,
the local representations of $\reg_j(v)$, $j\in \{1,\ldots,j_g\}$, and the trees
spanning them, membership of edges in $F=F_{i_{j_g}}$, and
$M_v=M_{i_{j_g}+1}(v)$ are identical to the corresponding values for an
execution of \algref{algo:central_approx}. Then in growth phase $g$, Step 3a
of the algorithm correctly computes the terminal decompositions $j\in
\{j_g+1,\ldots,j_{g+1}\}$, as well as the variables
$\moat^{(j)}(v)=\moat_{i_j}(v)$ and $\act^{(j+1)}(v)$. It can be completed in
$\BO(sk_g)$ rounds.
\end{lemma}
\begin{proof}
We prove the claim by induction on the iterations $j\in
\{j_g+1,\ldots,j_{g+1}\}$ of the loop in Step 3a, anchored at $j=j_g$. The
hypothesis is that all respective values for index $j$ are correct, which holds
for $j=j_g$ by assumption. For the induction step from $j-1$ to $j$, observe
that the hypothesis and \lemmaref{lemma:partition} yield that Step 3aii can be
performed in $\BO(s)$ rounds. Clearly, Step 3aiii requires one round of
communication only.

If $c_{\min}:=\argmin_{u\in V}\{c_u\}\neq \bot$, suppose
$c_{\min}=(\{v,w\},j,\hat{W},e)$. If $\hat{W}+\sum_{j'=1}^{j-1}\mu^{(j')}\leq
\hat{\mu}$, we claim that $c_{\min}\in F^{(j)}$ is the candidate merge
completing merge phase $j$. Otherwise (also if $c_{\min}=\bot$), $j=j_{g+1}$ and
active moats grow by exactly $\hat{\mu}-\sum_{j'=1}^{j-1}\mu^{(j')}$ during the
merge phase. To see this, recall that merge phase $j$ ends if (i) an active and
an inactive moat merge or (ii) active moats have grown by
$\hat{\mu}-\sum_{j'=1}^{j-1}\mu^{(j')}$. Note that, by
\lemmaref{lemma:partition} and the induction hypothesis,
\begin{equation*}
\hat{W}_j(p_{v_u\{u,u'\}v_{u'}}\cap
\reg_j(v_u))=\Wd(v_u,u)-\moat^{(j-1)}(v_u)+W(\cap\{u,u'\}\cap T_u)
\end{equation*}
for any $\{u,u'\}\in \reg_{j-1}(v_u)\cup (\vor_j(v_u)\setminus \bigcup_{w\in
T}B_{i_{j-1}}(w))\cup \reg_j(v_{u'})$ so that $\act^{(j)}(v_u)=\true$ and
$\act^{(j)}(v_{u'})=\false$. Moreover, $\{u,u'\}\notin \reg_{j-1}(v_u)\cup
\reg_{j-1}(v_{u'})$, as otherwise $v_u$ and $v_{u'}$ would have been connected
in an earlier merge phase and cannot satisfy that $\act^{(j)}(v_u)\neq
\act^{(j)}(v_{u'})$.

Suppose (i) applies, i.e., \algref{algo:central_approx} merges the moats of
terminals $v_{i_j}$ and $w_{i_j}$ in step $i_j$, and suppose it does so by the
path $p_{v_{i_j}ew_{i_j}}$ induced by $e= \{u,u'\}$ with $u\in \reg_j(v_{i_j})$
and $u'\in \reg_j(w_{i_j})$ (by \lemmaref{lemma:path_in_region}, we know that
such an edge exists). Since the merge phase ends due to this merge and terminals
can become inactive only at the end of a growth phase, it must hold that
$\true=\act_{i_j}(v_{i_j})\neq \act_{i_j}(w_{i_j})$. It follows that
$c_u=c_{\min}$, as any $c_{u''}<c_u$ would imply that another pair of terminals
from active and inactive moats would be merged earlier, ending the merge phase
at an earlier point. The same argument yields that in case of (ii), no
$c_u=(\cdot,j,\hat{W},\cdot)$ can exist with
$\hat{W}+\sum_{j'=1}^{j-1}\mu^{(j')}\leq \hat{\mu}$, as otherwise an active and
inactive terminal would get merged before the growth phase ends.

We conclude that the above claim holds. It follows that in Step 3aiv, which can
be completed in $\BO(D+s)=\BO(s)$ rounds, the correct variables
$\moat^{(j)}(v)$, $v\in T$, and therefore also regions $\reg_j(v)$ are
determined. If $j=j_{g+1}$, the induction halts. Otherwise, we know that the
merge $i_j$ connects an active and inactive moat. Because the input labels of
terminals in the inactive moat must be disjoint from those of other terminals
(as by the hypothesis the variables $\act^{(j)}$ have correct values), the resulting
moat must consist of active terminals; no terminals outside the new moat change
their activity status. By the prerequisites of the lemma, the terminals in the
inactive moat $M$ recognize their membership by the identifier of their leader
$L(M)$. Since any merge with an inactive moat makes its terminals active and no
terminals can become inactive except for the end of a growth phase, we conclude
that Step 3av results in the correct values of the variables $\act^{(j+1)}(v)$,
$v\in T$. Step 3av requires $\BO(D+s)=\BO(s)$ rounds, resulting in a total
complexity of $\BO(s)$ of the iteration of the while-loop in Step 3a.

The above establishes that, unless $j=j_{g+1}$, the induction hypothesis is
established for index $j+1$. Hence, the induction succeeds. We conclude that
there are $k_g=j_{g+1}-j_g$ iterations of the loop in Step 3a, for each of which
we observed that it can be implemented with running time $\BO(s)$.
\end{proof}

\begin{lemma}\label{lemma:growth_3b}
Suppose that the prerequisites of \lemmaref{lemma:decomp_correct} are satisfied
for growth phase $g$. Then, each candidate merge selected by the above
algorithm in Step 3b of growth phase $g$ is in $F_c$ for a (specific, for all
applications of the lemma to an instance fixed) execution of
\algref{algo:central_approx}. The step can be completed in $\sO(\sigma+s)$
rounds.
\end{lemma}
\begin{proof}
As in \lemmaref{lemma:equivalent}, we consider the execution of
\algref{algo:central_approx} employing the same tie breaking mechanism as we use
to order candidate merges.

We prove the claim by induction on the iterations of the loop in Step 3b. The
hypothesis is that all merges performed by the algorithm up to the beginning of
the current loop iteration correspond indeed to candidate merges from $F_c$ and
the moats $\M$ defined in Step 3bi implicitly given by the variable $L(M_v)$
known to each $v\in T$ are the moats induced by the union of edges of associated
paths. The induction is anchored by the assumptions of the lemma; hence consider
some iteration of the loop.

Suppose for a moat $M\in \M$, the smallest candidate merge is
$(\{v,w\},j,\hat{W},e)$. By \lemmaref{lemma:decomp_correct}, the nodes in $e$
can detect the existence of the candidate merge by communicating over $e$;
performing this concurrently for all nodes, this takes one round, since each
edge induces one candidate merge only. Since the moat is small, by
\lemmaref{lemma:large_moats}, the moat's component in $(V,F)$ has diameter at
most $\sigma$. Hence, a spanning tree rooted at the leader can be constructed
and used to determine the least-weight candidate merge as specified in Step 2bi
within $\BO(\sigma+s)$ rounds (the additive $s$ accounts for the depth of the
trees of the terminal decomposition).

In Step 2bii, only small moats $M\in \M$ need to participate in the computation.
We interpret the subgraph of the graph specified in Step 2bii induced by $F_C$
as a directed graph, where each small moat has one outgoing edge. We $3$-color
the graph by simulating the Cole-Vishkin algorithm~\cite{CV-86} on this graph,
where moat leaders take the role of the nodes and communication is routed
through the spanning trees of the moats. Observe that since nodes need to
receive messages only from their ``parent'' and send identical messages to their
children, the congestion is constant. Hence, each round of the Cole-Vishkin
algorithm can be simulated in $\BO(\sigma+s)$ rounds in $G$, the depth bound for
the trees constructed in Step 2bi. After $\BO(\log^* |\M|)\in \sO(1)$ rounds, a
$3$-coloring is computed, which in $3$ additional simulated rounds can be used
to determine a maximal matching. After another simulated round, each moat
leader in a small moat knows its incident edges from $F_C$. Consequently, Step
2biii requires another $\BO(\sigma+s)$ rounds.

Concerning Step 3biv, observe that the construction of $F_+$ ensures for each
connectivity component of $(\M,F_+)$, either all moats in the component are
small and it consists of two stars connected by a matching edge, or it is a star
centered at a large moat, whose leaves are all small moats. If the former
applies, \lemmaref{lemma:large_moats} shows that Step 3biv can be completed for
small moats within $\BO(\sigma+s)$ rounds using the edges from $F$ in the
respective component of $(V,F)$ only. Moreover, in this time a spanning tree can
be constructed and used to count the number of terminals or nodes, respectively,
determining whether the new moat is small. For the case where a large moat is
involved, the new leader will be the leader of the unique large moat in the
respective connectivity component of $(V,F)$. Since this leader is already known
to all terminals in the large moat, \lemmaref{lemma:large_moats} shows that its
identifier can be distributed to all nodes in the ``attached'' small moats in
$\BO(\sigma+s)$ rounds. Trivially, the resulting moat is large.

With respect to Step 3bv, we again apply \lemmaref{lemma:large_moats}, showing
that for each small moat, in $\BO(\sigma)$ rounds, a spanning tree with edges
from $F$ can be constructed that spans its component in $(V,F)$. This tree is
used to broadcast the terminal identifiers of its at most $\sqrt{\min\{st,n\}}$
terminals to all constituent nodes within $\BO(\sigma)$ rounds.

To complete the induction step, it thus remains to show that $F_+\subseteq F_c$
and therefore indeed all edges selected into $F$ in Step 3biii are also selected
by the execution of \algref{algo:central_approx} that selects the same merges
and the associated paths, and also that the computed moats are indeed the cuts
of $T$ with the connectivity components of $(V,F)$. Observe that for a candidate
merge added to $F_C$ by moat $M$, any cycle it might close in $G_c$ must contain
another candidate merge between terminals in $M$ and $T\setminus M$. Since any
candidate merge selected into $F_C$ is minimal among \emph{all} candidate merges
for $M$, it follows that it will never be filtered out. Therefore, it must
hold that $F_+\subseteq F_C\subseteq F_c$. As we already observed earlier,
$(\M,F_+)$ is a forest at the end of Step 3bii. Since for each
$(\{v,w\},\cdot,\cdot,\cdot)\in F_+$ the associated path is contained in
$\reg_j(v)\cup \reg_j(w)$ for some $j$, it connects exactly the moats
$M_v,M_w\in \M$. We conclude that the new moats are exactly those computed in
the iteration of the loop in Step 3b. We conclude that the induction hypothesis
for the next loop iteration is established, i.e., the induction succeeds. Since
there are $\BO(\log n)$ iterations, the total time complexity is
$\sO(\sigma+s)$.
\end{proof}

\begin{lemma}\label{lemma:growth_phase}
Suppose the prerequisites of \lemmaref{lemma:decomp_correct} are satisfied for a
growth phase $g\in \{1,\ldots,g_{\max}\}$. Then the growth phase can be
completed in $\sO(k_gs+\sigma)$ rounds and the prerequisites of
\lemmaref{lemma:decomp_correct} hold for index $g+1$.
\end{lemma}
\begin{proof}
By \lemmaref{lemma:decomp_correct}, the regions $\reg_{j_{g+1}}(v)$, $v\in V$,
have been determined in Step 3a, within $\BO(k_g s)$ rounds. By
\lemmaref{lemma:growth_3b}, the Step completes in $\sO(s+\sigma)$ rounds and
determines for $v\in T$ the variable $L(M_v)$ in accordance with $F$, where $F$
is the edge set of paths associated with a set $F_c'\subseteq F_c$. Step 3c can
thus be correctly executed in $\BO(s)$ rounds, and Step 3d, which requires local
computations only, will determine sets $E_c(u)$, $u\in V$, so that
$\bigcup_{j=1}^{j_{g+1}}E_c^{(j)}\subseteq \bigcup_{u\in V}E_c(u)$. Hence, the
preconditions of \lemmaref{lemma:filtering} are satisfied, permitting to perform
Step 3e in $\BO(D+|\bigcup_{j=1}^{j_{g+1}}F_c^{(j)}\setminus F_c'|)$ rounds.

We claim that $|\bigcup_{j=1}^{j_{g+1}}F_c^{(j)}\setminus F_c'|\leq \sigma$. To
see this, observe that in each iteration of the loop in Step 2b, each small moat
that has an incident candidate merge will be merge with some other moat. Hence,
the minimal number of terminals (if $st<n$) or nodes (if $st\geq n$) in a moat
that can still participate in a merge in the growth phase doubles in each
iteration of the loop. It follows that after Step 2b, any moat that can still
participate in a merge in merge phase $g$ is large. By
\lemmaref{lemma:large_moats}, there are at most $\sigma$ large moats. Since
$F_c$ (as edge set in $G_c$) contains neither cycles nor duplicate edges, the
claim follows. In particular, Step 3d completes within $\BO(D+\sigma)\subseteq
\BO(\sigma+s)$ rounds.

Since $F_+$ becomes known to all nodes, Step 3f can be performed in $\BO(s)$
rounds. For Step 3g, we collect for each candidate merge in $F_+$ the
identifiers of the merged moats' leaders over the BFS tree, in
$\BO(D+\sigma)\subseteq \BO(s+\sigma)$ rounds. The new leaders then can be
computed locally by all nodes, since $F_+$ is known by all nodes. For each new
moat, the number of terminals (or nodes) is then determined by pipelining the
respective additions on the BFS tree and broadcasting the result to all nodes,
again requiring $\BO(s+\sigma)$ rounds. This enables each node to determine
whether its moat is small or large. Step 3h is performed, for each small moat,
within its connectivity component of $(V,F)$. Because
\lemmaref{lemma:large_moats} states that the diameter of these components is at
most $\BO(\sigma)$ and small moats contain at most $\sigma$ terminals, this
completes in $\BO(\sigma)$ rounds.

To perform Step 3i, we identify all terminals in each moat with the moat leader
and then apply the technique from \lemmaref{lemma:transform_to_minimal}. Since
an input component $\Comp\in \Lambda$ is subset of a moat if and only if there
will be only one tuple $(\Comp(v),L(M_v))$ with $\Comp(v)=\Comp$ present
(possibly at several nodes), this will determine correctly which input
components are satisfied, after $\BO(D+k)\subseteq \BO(s+k)$ rounds. A moat is
active in growth phase $g+1$ if and only if there is a terminal whose input
component is not subset of some moat. By \lemmaref{lemma:large_moats}, small
moats have diameter at most $\sigma$ w.r.t.\ $(V,F)$, enabling to complete the
step within another $\BO(\sigma)$ rounds for small moats. For large moats, we
perform the respective convergecasts and broadcasts on the BFS tree, tagging
the messages with the moat leader's identifier. Because, by
\lemmaref{lemma:large_moats}, there are at most $\BO(\sigma)$ large moats, the
congestion at each node is bounded by $\BO(\sigma)$ and the step can be
completed in $\BO(D+\sigma)\subseteq \BO(s+\sigma)$ rounds for large moats.

Summing up the time complexities of all steps, a total of $\sO(sk_g+\sigma)$
rounds suffices to complete the growth phase. The variables
$\act_{j_{g+1}+1}(v)$, $v\in V$, have been determined in Step 3i. The variables
$\moat^{(j_{g+1})}$ are, by \lemmaref{lemma:decomp_correct}, known by the end of
Step 3a of growth phase $g$, alongside $\reg_{j_{g+1}}(v)$ and the corresponding
spanning trees, for $j\in \{1,\ldots,j_{g+1}\}$. By \lemmaref{lemma:growth_3b},
the moat leader variables reflected the moats corresponding to the respective
set of selected edges $F$ after Step 3b, which in turn matched a set
$F_c'\subseteq \bigcup_{j=1}^{j_{g+1}}F^{(j)}$ (since there were never any
candidate merges for phases $j>j_{g+1}$). By \lemmaref{lemma:filtering} and
Steps 3e to 3g, we conclude that $F$ is the edge set of the paths associated
with $\bigcup_{j=1}^{j_{g+1}}F^{(j)}$, i.e., $F=F_{i_{j_{g+1}}}$ for the
considered execution of \algref{algo:central_approx}, and leader variables
$L(M_v)$, $v\in T$, have the correct values for these moats. In summary, all
claims of the lemma hold and the proof concludes.
\end{proof}

\begin{proof}[Proof of \corollaryref{coro:growth}]
Constructing a BFS tree requires $\BO(D)$ rounds.
\lemmaref{lemma:number_merge_growth_phase} and inductive application of
\lemmaref{lemma:growth_phase} shows that Steps 2, 3, and 4 of the algorithm can
be executed in $\sO((\sum_{g=1}^{g_{\max}}k_g)s+\sigma)=\sO(ks+\sigma)$ rounds.
Moreover, the returned set $F$ equals the set $F_{i_{j_{g_{\max}}}}$ computed by
\algref{algo:central_approx}. Therefore, it is a forest, and its minimal
subforest solving the instance is, by \theoremref{theorem:2+eps_approx}, optimal
up to factor $2+\varepsilon$.
\end{proof}

\subsection{Fast Pruning Algorithm}
\label{app:prune}
The following routine assumes that for an instance of \sfic, a forest $F$ on at
most $\sigma^2:=\min\{st,n\}$ nodes solving the instance is given, where each
node knows which of its incident edges are in $F$. At the heart of the routine
are Steps 4 to 6, which heavily exploit that $F$ is a tree to ensure optimal
pipe-lining for the edge selection process.
\begin{compactenum}
\item Set $F_0:=\emptyset$ (this will be the pruned edge set). Construct an
(unweighted) BFS tree on $G$, rooted at $R$ and make the set of labels $\Lambda$
known to all nodes.
\item For each connectivity component of $(V,F)$ of diameter at most $\sigma$,
optimally solve the respective (sub)instance of \sfic. Add the respective edges
to $F_0$ and delete these components from $(V,F)$. W.l.o.g., assume that all
components of $(V,F)$ have diameter larger than $\sigma$ in the following.
\item Construct a partition of $(V,F)$ into clusters ${\cal C}$, so that
(i) $|{\cal C}|\leq \sigma$, (ii) for each $C\in {\cal C}$, the depth of the
minimal subtree of $F$ spanning $C$ is $\sO(\sigma)$, and (iii) for each $C\in
{\cal C}$, the spanning subtree induced by $F$ is directed to a root $R_C\in C$
(in the sense that each node knows its parent and the identifier of $R_C$).
\item Denote by $({\cal C},F_{\cal C})$ the
forest on ${\cal C}$ resulting from contracting each $C\in {\cal C}$ in $(V,F)$.
Make $({\cal C},F_{\cal C})$ known to all nodes.
\item Each node $u\in V$ initializes for each $e\in F_{\cal C}$
$l_e(u):=\emptyset$ and for each $C\in {\cal C}$ $l_C(u):=\emptyset$. Terminals
$v\in T$ set $l_C(v):=\{\Comp(v)\}$ (where $C$ is uniquely identified by the
identifier of $R_C$).
\item Perform the following on the BFS tree until no more messages are sent
\begin{compactitem}
\item Each node $u\in V$ sends a \emph{non-redundant} node label $(C,\Comp)$ for
$\Comp\in l_C(u)$ to its parent (if there is one). A label is \emph{redundant}
if the following holds. Start from variables $\hat{l}_e(u):=\emptyset$ and
$\hat{l}_C(u):=\emptyset$ and simulate the operations below for all messages
sent to the parent in previous rounds. If the label in question would not alter
the state of the variables further, it is redundant.
\item If $w\in V$ receives ``$(C,\Comp)$'', it sets
$l_C(w):=l_C(w)\cup\{\Comp\}$. If there is some other $C'\in {\cal C}$ with 
$\Comp\in l_{C'}(w)$, it sets $l_e(w):=l_e(w)\cup \{\Comp\}$ and
$l_{C''}(w):=l_{C''}(w)\cup \{\Comp\}$ for all edges $e$ and nodes $C''$ on the
path connecting $C$ and $C'$.\footnote{Note that different connectivity
components of $(V,F)$ must have disjoint sets of labels, since $F$ solves the
instance. Since $F$ is a forest, there is thus always a unique such path.}
\item Whenever there is for any node $u\in V$ an edge $e$ with
$\Comp,\Comp'\in l_e(u)$, for each $e\in F_{\cal C}$ with $l_e(u)\cap
\{\Comp,\Comp'\}\neq \emptyset$ set $l_e(u):=l_e(u)\cup \{\Comp,\Comp'\}$ and
for each $C\in {\cal C}$ with $l_C(u)\cap
\{\Comp,\Comp'\}\neq \emptyset$ set $l_C(u):=l_C(u)\cup \{\Comp,\Comp'\}$.
\end{compactitem}
\item Once this is done, the root $R$ of the BFS tree broadcasts the result
(using the same encoding).
\item For each edge in $e\in F_{\cal C}$ with $l_e(R)\neq \emptyset$, add $e$ to
$F_0$.
\item For each terminal $v\in T$, set $l(v):=\{\Comp(v)\}$. Nodes $u\in
V\setminus T$ set $l(v):=\emptyset$. If node $u\in V$ is the endpoint of an edge
$e\in F_{\cal C}$, $u$ sets $l(u):=l(u)\cup l_e(R)$.
\item For each tree spanning a cluster $C\in {\cal C}$, select for each
$\Comp\in \Lambda$ the edges of the minimal subtree spanning all terminals $v\in
C$ with $\Comp\in l(v)$ into $F_0$.
\item Return $F_0$.
\end{compactenum}

We start by analyzing the time complexity of the routine. The first lemma covers
the selection procedure for trees of depth at most $\sigma$ used in Steps 2 and
10.

\begin{lemma}\label{lemma:prune_2_10}
Steps 2 and 10 of the above routine can be completed in $\BO(\sigma+k)$ rounds.
\end{lemma}
\begin{proof}
Consider a tree of depth at most $\sigma$, where each node $u$ in the tree is
given a set $l(u)\subseteq \Lambda$ and the requirement is to mark all edges
that are on a path connecting some nodes $u$ and $u'$ in the tree with $l(u)\cap
l(u')\neq \emptyset$, communicating over tree edges only. This is the
requirement of Step 10, and by setting $l(u)=\{\Comp(u)\}$ for terminals $u$ and
$l(u)=\emptyset$ otherwise, we see that Step 2 can be seen as a special case.

We root the tree in $\BO(\sigma)$ rounds. Consider a fixed label $\Comp \in
\Lambda$. Each node $u$ with $\Comp\in l(u)$ a message $\Comp$ to its parent,
which is forwarded to the root; each node sends only one such message $\Comp$.
All edges traversed by a message are tentatively marked. Once this is complete,
the root $R$ checks whether it received at least two messages $\Comp$ or
satisfies that $\Comp\in l(R)$. If this is not the case, it sends an ``unmark''
message to the child sending a $\Comp$ message (if there is one). The receiving
child performs the same check w.r.t.\ its subtree, possible sending another
``unmark'' message, and so on. Clearly, removing the edges traversed by an
``unmark'' message from the set of tentatively marked edges is the minimal set
of edges connecting the nodes with $\Comp \in l(u)$. We perform this process
concurrently for all $\Comp \in \Lambda$ (tagging the unmark messages by the
respective component label), using pipelining to avoid congestion. Each of
the two phases can be completed in $\BO(\sigma+k)$ rounds, since there are at
most $k$ distinct labels and each node sends at most two messages per label.
\end{proof}

The next lemma discusses the growing of clusters. The employed technique is the
same as for Step 3b of the subroutine from \sectionref{sec:tree_selection},
analyzed in detail in \lemmaref{lemma:growth_3b}.

\begin{lemma}\label{lemma:prune_3}
Step 3 of the above routine can be completed in $\sO(\sigma)$ rounds.
\end{lemma}
\begin{proof}
Initialize the clusters to singletons. We consider a cluster \emph{small}, if it
contains fewer than $\sigma$ nodes. Otherwise it is \emph{large}. For $\lceil
\log \sigma\rceil$ iterations, perform the following.
\begin{compactenum}
\item Each small cluster selects an arbitrary outgoing edge from $F$ (this is
feasible, since after Step 2 each connectivity component contains at least
$\sigma$ nodes). Denote the set of selected edges by $F_C$.
\item Suppose $F_C'$ is the subset of edges between small clusters. Find a
maximal matching $M\subseteq F_C'$.
\item Each small cluster without an incident edge from $M$ adds the previously
selected edge to $M$, resulting in set $F_+$.
\item Merge clusters according to $F_+$ (constructing rooted spanning trees).
The new clusters select a leader and determine whether they are small or not.
\end{compactenum}
Since for each small cluster in each iteration at least one edge is selected,
the minimal number of nodes in a cluster grows by at least factor $2$ in each
iteration, implying that no small clusters remain in the end. Since there can be
at most $\sigma^2/\sigma=\sigma$ large clusters, the bound on $|{\cal C}|$
holds. Due to the construction of $F_+$, in each iteration the longest path in
the graph on the current clusters that is selected into $F_+$ has $3$ hops.
Moreover, at most one large cluster is present in each connectivity component of
the subgraph induced by $F_+$, implying that the maximal diameter of clusters
remains in $\sO(\sigma)$.

Concerning the running time, observe that the matching can be selected by
simulating the Cole-Vishkin algorithm~\cite{CV-86} on the cluster graph. Due to
the bound on the diameter of clusters, the routine can be completed in
$\sO(\sigma)$ rounds.
\end{proof}

Step 6 of our subroutine pipelines several related pieces of information, namely
(i) the inter-cluster edges to select, (ii) input components ``responsible'' for
this edge to be selected, and (iii) input components which can be identified,
because the minimal subtrees of $F$ spanning them are not disjoint (and any
subforest of $F$ solving the instance connects the terminals in the respective
different input components, too).

\begin{lemma}\label{lemma:prune_6_7}
Steps 6 and 7 of the above routine can be completed in $\sO(\sigma+k+D)$ rounds.
\end{lemma}
\begin{proof}
For each $\Comp \in \Lambda$, any non-redundant (received) message
$(C,\Comp)$ after the first implies that some edge receives a new label.
Initially, the number of different possible labels for edges is at most $k$.
Whenever an already labeled edge receives an additional label, the number of
possible different edge labels is decreased by one. The number of times an
unlabeled edge can become labeled is at most $|{\cal C}|\leq \sigma$. We
conclude that no node sends more than $k+|F_{\cal C}|< k+\sigma$ messages.

Denote by $m_u$ the number of non-redundant messages non-root node $u$ will send
and by $d_u$ the depth of the subtree rooted at $u$. We claim that after $r$
rounds, $u$ has sent $\min\{r-d_u,m_u\}$ non-redundant messages, which we show
by induction on $d_u$. The statement is trivial for $d_u=0$, i.e., leaves. For
$d_u>0$, by the induction hypothesis at the end of round $r-1$, either $u$ has
received all non-redundant messages from its children or at least one child sent
at least $r-1-(d_u-1)$ non-redundant messages. Hence, if $u$ has not yet sent
$\min\{r-d_u,m_u\}$ non-redundant messages, it will send another message in
round $r$. By the induction hypothesis, it thus has sent $\min\{r-d_u,m_u\}$
messages by the end of round $r$, and the induction step succeeds.

We conclude that Step 6 completes within $\BO(\sigma+k+D)$ rounds. By
broadcasting $\BO(\sigma+k)$ non-redundant messages over the BFS tree, it can
make the result known to all nodes, also in $\BO(\sigma+k+D)$ rounds.
\end{proof}

It remains to show that the algorithm chooses the correct set of inter-cluster
edges and the demands derived from the respective selection process in Step 9
ensures that the intra-cluster edges selected into $F_0$ in Step 10 complete the
minimal solution.

\begin{lemma}\label{lemma:prune_correct}
The set $F_0\subseteq F$ returned is minimal with the property that it solves
the instance of \sfic solved by $F$.
\end{lemma}
\begin{proof}
Clearly, Step 2 does not affect the correctness of the solution. Hence,
w.l.o.g.\ assume that $(V,F)$ contains only components of diameter larger than
$\sigma$, i.e., no deletions happen in Step 2.

Observe that the minimal subforest solving the instance is the union over all
$\Comp\in \Lambda$ of the minimal trees $T_{\Comp}\subseteq F$ spanning all
terminals $v\in T$ with $\Comp(v)=\Comp$. Note that by the initialization and
due to the rules of Step 6, $T_{\Comp}\cap F_{\cal C}$ will be labeled by
$\Comp$, i.e., $l_e(R)\supseteq \{\Comp \in
\Lambda\,|\,e\in T_{\Comp}\}$. On the other hand, if $e\notin \bigcup_{\Comp
\in \Lambda}T_{\Comp}$, the set of input labels on each side of the edge must be
disjoint. Since Step 6 will maintain this invariant, the edge will satisfy that
$l_e(R)=\emptyset$. We conclude that the edges selected into $F_0$ in Step 8 are
exactly the edges from $F_{\cal C}$ in a minimal solution.

Denote for each node $e\in F$ in the minimal solution by $T_e\subseteq F$ its
component in the minimal solution; for $u\in T_e$ for some such $e$, denote
$T_u=T_e$ ($T_u:=\emptyset$ otherwise). We claim that $l_e(R)\subseteq
\{\Comp\,|\,\exists v\in T_e:\,\Comp(v)=\Comp\}$ at the end of Step 6. To see
this, we claim that the algorithm maintains for all $u\in V$ the invariants that
$l_e(u)\subseteq \{\Comp\,|\,\exists v\in T_e:\,\Comp(v)=\Comp\}$ and
$l_C(u)\subseteq \{\Comp\,|\,\exists v\in C, w\in T_v:\,\Comp(w)=\Comp\}$. This
holds trivially after the initialization in Step 5. According to the first rule
of Step 6 and the invariants, an sent message $(C,\Comp)$ satisfies that $\Comp
\subseteq \{\Comp\,|\,\exists v\in C, w\in T_v:\,\Comp(w)=\Comp\}$. Hence, a
node $w$ receiving ``$(C,\Comp)$'' will not violate the invariant due to its
change of $l_C(w)$. If there is some $C'\in {\cal C}$ with $\Comp\in l_{C'}(w)$,
the invariant implies that $C$ and $C'$ both contain nodes that are connected by
the minimal solution to terminals $u,u'\in T$ with $\Comp(u)=\Comp(u')=\Comp$.
Since these terminals must be connected, too, the path connecting $C$ and $C'$
is part of a single connectivity component of the minimal solution, which
contains terminals labeled $\Comp$. We conclude that the invariants cannot be
violated (first) due to the second rule of Step 6. Because if $T_{\Comp}\cap
T_{\Comp'}\neq \emptyset$, they must be part of the same connectivity component
of the minimal solution, the invariants cannot be violated (first) due to the
third rule of Step 6. In summary, the invariants are upheld, yielding in
particular that $l_e(R)\subseteq \{\Comp\,|\,\exists v\in
T_e:\,\Comp(v)=\Comp\}$ for each $e$ with $l_e(R)\neq \emptyset$.

From this result, it follows that replacing the labels $\Comp(v)$, $v\in T$, by
the sets $l(u)$, $u\in V$, defined in Step 9, does not change the minimal subset
of $F$ that satisfies all constraints: if endpoint $u\in V$ of edge $e\in
F_{\cal C}$ sets $l(u):=l(u)\cup \{\Comp\}$ for $\Comp \in l_e(R)$, it follows
that $T_{\Comp}\subseteq T_e$ and therefore $u$ is connected to all terminals
$v\in T$ with $\Comp(v)=\Comp$ by the minimal solution.

Trivially, Step 10 cannot violate the minimality of the computed solution; it
thus remains to show that after Step 10, $F_0$ solves the instance. Suppose
$v,w\in T$ with $\Comp(v)=\Comp(w)$. If $v,w\in C$ for some $C$, Step 10 ensures
that $v$ and $w$ are connected by $F_0$, since $\Comp\in l(v)\cap l(w)$. Hence,
suppose that $v\in C\neq C'\ni w$. Denote by $p_{vw}$ the unique path connecting
$v$ and $w$ in $F$. We already observed that $p\cap F_{\cal C} \subseteq F_0$
and each edge $e\in p\cap F_{\cal C}$ satisfies that $\Comp\in l_e(R)$. Due to
Steps 9 and 10, it follows that $F_0$ connects $v$ and $w$.
\end{proof}

We summarize the results of our analysis of the pruning routine as follows.

\begin{corollary}\label{coro:prune}
Given an instance of \sfic and a forest $F$ on $\sigma^2$ nodes that solves
it, the above routine computes the minimal $F_0\subseteq F$ solving the
instance. It can be implemented with running time $\sO(\sigma+k+D)$.
\end{corollary}
\begin{proof}
Correctness is shown in \lemmaref{lemma:prune_correct}. Step 1 requires
$\BO(k+D)$ rounds. Step 4 can be completed in $\sO(\sigma+D)$ rounds, since due
to Step 3 $|{\cal C}|\leq \sigma$ and the nodes incident to the edges in
$F_{\cal C}$ know that these edges are in $F_{\cal C}$. Steps 5, 8, 9, and 11
require local computations only. The remaining steps can be completed within
$\sO(\sigma+k+D)$ rounds by Lemmas~\ref{lemma:prune_2_10}, \ref{lemma:prune_3},
and \ref{lemma:prune_6_7}.
\end{proof}

We conclude that executing the pruning routine on the input $F$ determined by
the algorithm from \sectionref{sec:sublinear} yields a fast factor
$(2+\varepsilon)$-approximation.

\begin{theorem}\label{theorem:2+eps_distributed}
For any constant $\varepsilon>0$, a deterministic distributed algorithm can
compute a solution for problem \sfic that is optimal up to factor
$(2+\varepsilon)$ in $\sO(sk+\sqrt{\min\{st,n\}})$ rounds.
\end{theorem}
\begin{proof}
By \corollaryref{coro:growth}, a forest solving the problem whose minimal
subforest is optimal up to factor $2+\varepsilon$ can be computed in
$\sO(sk+\sqrt{\min\{st,n\}})$. Note that the
forest is the union of at most $t-1$ paths of hop length at most $s$, and
trivially contains at most $n$ nodes. Hence, we can apply
\corollaryref{coro:prune} with $\sigma=\sqrt{\min\{st,n\}}$ to show the claim of
the theorem.
\end{proof}

Since any instance can be transformed to one with minimal inputs efficiently and
the number of different terminal decompositions that needs to be computed is
trivially bounded by $\WD$, we obtain the following stronger bound as a
corollary.

\begin{proof}[Proof of \corollaryref{coro:2+eps_distributed}]
By \lemmaref{lemma:transform_to_minimal}, we can transform the instance to a
minimal instance in $\BO(k+D)$ rounds; the minimal instance has $k_0$ input
components. The number of different possible moat sizes at which merges may
happen is bounded by $\WD$ (since edge weights are assumed to be integer and
moats grow to size at most $\WD/2$). If multiple merge phases end for the same
such value, we can complete all of them without having to recompute the terminal
decomposition. The result thus follows from
\theoremref{theorem:2+eps_distributed}.
\end{proof}

\section{Proofs for \texorpdfstring{\sectionref{sec-alg2}}{Section
\ref*{sec-alg2}}}

\subsection{Partial Construction of the Virtual Tree}\label{sec:partial}
We start out with some basic observations on the virtual tree that is
constructed by the algorithm from~\cite{KKMPT-12}.
\begin{lemma}\label{lemma:stage1_tree_s}
The following holds for the virtual tree described above.
\begin{compactenum}
\item The tree nodes corresponding to the set ${\cal S}$ of the $\sqrt{n}$ nodes
of highest rank induce a subtree.
\item For each leaf $v$, denote by $i_v\in \{0,\ldots,L\}$ the
minimal index so that $B_G(v,2^{i_v}\beta)\cap {\cal S}\neq \emptyset$. Then,
for $i\in \{0,\ldots,i_v-1\}$, there is a least-weight path from $v$ to $v_i$ of
$\tilde{\BO}(\sqrt{n})$ hops w.h.p.
\item For each leaf $v$, w.h.p.\ there is a node
$\tilde{v}_{i_v}\in {\cal S}$ for which $\Wd(v,\tilde{v}_{i_v})=\min_{w\in {\cal
S}}\{\Wd(v,w)\}$ and there is a least-weight path from $v$ to $\tilde{v}_{i_v}$
of $\tilde{\BO}(\sqrt{n})$ hops.
\end{compactenum}
\end{lemma}
\begin{proof}
The first statement follows from the fact that for each $i\in \{0,\ldots,L-1\}$,
the index of $v_{i+1}$ w.r.t.\ the random order must be larger than that of
$v_i$, since $v_{i+1}$ attains the maximum index over
$B_G(v,2^{i+1}\beta )\supseteq B_G(v,2^i\beta )$.

For the second statement, consider for any pair of nodes $v$ and $w$ a
least-hop shortest path from $v$ to $w$. If this path contains at least
$(c+3)\sqrt{n}\ln n$ hops (for a given constant $c$), it contains also at least
$(c+3)\sqrt{n}\ln n$ nodes (since least-weight paths cannot revisit nodes).
Observe that ${\cal S}$ is a uniformly random subset of the nodes. Hence, the
probability that no node from ${\cal S}$ is on the path is bounded from above
by
\begin{eqnarray*}
\binom{n-|{\cal S}|}{(c+3)\sqrt{n}\ln n}\cdot \binom{n}{(c+3)\sqrt{n}\ln
n}^{-1}&=&\frac{(n-\sqrt{n})!}{n!}\cdot \frac{(n-(c+3)\sqrt{n}\ln
n)!}{(n-(c+3)\sqrt{n}\ln n-\sqrt{n})!}\\
&< & \frac{(n-(c+3)\sqrt{n}\ln
n)^{\sqrt{n}}}{(n-\sqrt{n})^{\sqrt{n}}}\\
&< & \left(1-\frac{(c+2)\ln n}{\sqrt{n}}\right)^{\sqrt{n}}\\
&<& e^{-(c+2)\ln n}\\
&=& n^{-c-2}.
\end{eqnarray*}
By the union bound applied to all pairs of nodes $v,w\in V$, we conclude that
the probability that \emph{any} of these paths contains no node from ${\cal S}$
is at most $n^{-c}$. In other words, w.h.p., for each pair of nodes $v,w\in
V$, either a least-weight path from $v$ to $w$ with
$\tilde{\BO}(\sqrt{n})$ hops exists, or there is a node from ${\cal S}$ on a
least-weight path from $v$ to $w$, which therefore is closer to $v$ w.r.t.\
weighted distance than $w$. The second claim of the lemma follows. Regarding
the third claim, observe that the same reasoning applies if we condition on
$w\in {\cal S}$, showing that w.h.p.\ the least-weight path from $v$ the nodes
$w\in {\cal S}$ minimizing $\Wd(v,w)$ must have $\tilde{\BO}(\sqrt{n})$ hops.
\end{proof}
We leverage these insights to compute the virtual tree partially.
\begin{lemma}\label{lemma:stage1_partial}
Delete the internal nodes corresponding to the set ${\cal S}$ of the $\sqrt{n}$
nodes of highest rank from the virtual tree. W.h.p., the resulting forest can be
computed within $\tilde{\BO}(\sqrt{n}+D)$ rounds. Moreover, within this number
of rounds, each node $v\in V\setminus {\cal S}$ can learn about
$\tilde{v}_{i_v}$ and all nodes on the corresponding least-weight path can learn
the next hop on this path w.h.p. All detected least-weight paths have
$\tilde{\BO}(\sqrt{n})$ hops w.h.p.
\end{lemma}
\begin{proof}
We compute a Voronoi decomposition of $G$ w.r.t.\ to ${\cal S}$. This can be
done by, essentially, the single-source Bellmann-Ford algorithm\footnote{Connect
all nodes in ${\cal S}$ to a virtual node by edges of weight $0$ and piggy-back
the identifier of the node from ${\cal S}$ through which the constructed path to
the virtual node would pass on each message.} in time $\tilde{\BO}(\sqrt{n})$
w.h.p., since by Statement (iii) of \lemmaref{lemma:stage1_tree_s}, for each
$v\notin {\cal S}$, there is a least-weight path from $v$ to $\tilde{v}_{i_v}\in
{\cal S}$ of $\tilde{\BO}(\sqrt{n})$ hops w.h.p. Termination can be detected
over a BFS tree, requiring additional $\BO(D)$ rounds. This shows the second
claim of the lemma. As a byproduct, each node $v\notin {\cal S}$ learns $i_v$,
and the nodes on the corresponding least-weight path from $v$ to
$\tilde{v}_{i_v}$ learn the next routing hop on the path.

Now we execute the algorithm from~\cite{KKMPT-12}, however, constructing only
the forest resulting from deleting the internal nodes corresponding to nodes
from ${\cal S}$. By Statement~(i) of \lemmaref{lemma:stage1_tree_s}, this can be
done by determining, for each $v\in V\setminus {\cal S}$, the nodes $v_i$,
$i\in \{0,\ldots,i_v-1\}$, and the corresponding least-weight paths in $G$
connecting $v$ to the $v_i$. The algorithm from~\cite{KKMPT-12} requires time
$\tilde{\BO}(\tilde{s}+D)$ to do so, where $\tilde{s}$ is the maximal length of
any of the detected paths;\footnote{At the heart of the tree embedding algorithm
from~\cite{KKMPT-12} lies the construction of so-called LE lists. The algorithm
proceeds in phases of $\BO(\log n)$ rounds, where in each round, information
spreads by one hop along least-weight paths.} by Statement~(ii) of
\lemmaref{lemma:stage1_tree_s}, $\tilde{s}\in \tilde{\BO}(\sqrt{n})$ w.h.p.
\end{proof}

\subsection{Tree Construction and Edge Selection
Stage}\label{sec:tree_selection}

\subsubsection*{Time Complexity}

To prove that the first stage can be completed sufficiently fast, we show
helper lemmas concerning Steps 3a, 3c, and 3d of each phase of the stage.

\begin{lemma}\label{lemma:stage1_3a}
Fix a phase $i\in \{0,\ldots,L\}$ of the first stage. Step 3a of the phase can
be completed in $\BO(k+D)$ rounds.
\end{lemma}
\begin{proof}
The following is performed on a BFS tree.
\begin{compactitem}
\item If the third rule does not prohibit this, each node sends for each
$\Comp\in l(v)$ a message $(\Comp,v)$ to its parent.
\item Each node receiving a message forwards it to its parent, unless prohibited
by the third rule.
\item If a node ever receives a second message containing label
$\Comp$, it sends $(\Comp,\bot)$ to its parent and ignores all other messages
concerning $\Comp$.
\item Once this completes, the root of the tree can determine for which labels
$\Comp \in \Lambda$ there is only a single active terminal $v\in T$ with
$l(v)=\Comp$.
\end{compactitem}
This operation completes within $\BO(D+k)$ rounds, since no node sends more than
two messages for each label. The root broadcasts the result over the BFS tree to
all nodes, which also takes time $\BO(D+k)$.
\end{proof}

\begin{lemma}\label{lemma:stage1_3cd}
Fix a phase $i\in\{0,\ldots,L\}$ of the first stage. Steps 3c and 3d of the
phase can be implemented such that they complete within
$\sO(\min\{s,\sqrt{n}\}+k+D)$ rounds w.h.p.
\end{lemma}
\begin{proof}
Observe that if $s\leq \sqrt{n}$, for each $v\in T$ and each $i\in
\{0,\ldots,L\}$, the least-weight path from $v$ to $v_i$ determined by the tree
construction has at most $s$ hops. If $s>\sqrt{n}$ and the partial construction
was executed, by \lemmaref{lemma:stage1_partial} no detected path has more than
$\tilde{s}\in \sO(\sqrt{n})$ hops w.h.p. Therefore, all least-weight paths in
$G$ used in Step 3c have at most $\tilde{s}$ hops w.h.p.

In Step 3c, in each iteration of the sending rule, each node sends at most one
message for each node $w$ such that it is on a least-weight path from some leaf
of the virtual tree to $w$ determined by the tree construction. By the
properties of the tree, each node $v\in V$ participates in at most $\BO(\log n)$
different such paths w.h.p. Hence each iteration requires $\BO(\log n)$ rounds
w.h.p.\footnote{Note that the respective bound can be computed from $n$, which
can be determined and communicated to all nodes in $\BO(D)$ rounds. Therefore,
the iterations can be performed sequentially without the need to explicitly
synchronize their execution.}

Consider all messages $(\cdot,w)$ that are sent in phase $i$. These messages are
sent along least-weight paths, i.e., they induce a tree rooted at $w$ in $G$.
For each $\Comp\in \Lambda$, each node in the tree sends at most one message. In
each iteration of the sending rule, a node will send some message $(\Comp,w)$ if
it currently stores any message $(\Comp',w)\in \unsent\setminus \sent$. Hence,
the total number of iterations until all messages $(\Comp,w)$ are delivered is
bounded by the sum of the depth of the tree, which is bounded by $\hat{s}$
w.h.p., and $|\Lambda|=k$. Termination of Step 3c can be detected at an additive
overhead of $\BO(D)$ rounds over a BFS tree. We conclude that Step 3c can be
performed in $\sO(\tilde{s}+k+D)$ rounds w.h.p.

Concerning Step 3d, the same arguments apply: on each tree rooted at some $w$,
at most $|\hat{l}(w)|\leq |\Lambda|=k$ messages need to be sent to some node $v$
in the tree. Using the same approach as for Step 3c, this requires
$\sO(\tilde{s}+k+D)$ rounds.
\end{proof}

\begin{corollary}\label{coro:stage1_time}
The first stage can be completed in $\tilde{\BO}(\min\{s,\sqrt{n}\}+k+D)$
rounds w.h.p.
\end{corollary}
\begin{proof}
In~\cite{KKMPT-12}, the authors show that the virtual tree can be constructed in
$\tilde{\BO}(s)$ rounds w.h.p. In \lemmaref{lemma:stage1_partial}, we show that
the partial tree can be constructed in $\tilde{\BO}(\sqrt{n}+D)$ rounds w.h.p.
Therefore, Step 1 of the algorithm completes in
$\tilde{\BO}(\min\{s,\sqrt{n}\}+D)$ rounds w.h.p. As Steps 2 and 4 are local
and $L=\BO(\log \WD)=\BO(\log n)$, it is sufficient to show that each phase
can be implemented in time $\sO(\min\{s,\sqrt{n}\}+k+D)$ w.h.p. By
\lemmaref{lemma:stage1_3a}, Step 3a of each phase can be completed in $\BO(D+k)$
rounds. Step 3b requires local computations only. \lemmaref{lemma:stage1_3cd}
shows that Steps 3c and 3d can be executed in $\sO(\min\{s,\sqrt{n}\}+k+D)$
rounds w.h.p.
\end{proof}

\subsubsection*{Approximation Ratio}

We now prove that the weight of the edge set $F$ selected in the first stage is
bounded by the cost of the optimal solution on the virtual tree, which is
optimal up to factor $\BO(\log n)$ in expectation. This is facilitated by the
following definition.

\begin{definition}[$\Comp$-subtrees]
For $\Comp \in \Lambda$, denote by $T_{\Comp}$ the minimal subtree of the
virtual tree such that all terminals $v\in T$ with $\Comp(v)=\Comp$ are leaves
in the subtree.
\end{definition}

We now show that the edges selected into $F$ correspond to edges in the optimal
solution on the virtual tree, which is the edge set of the
union $\bigcup_{\Comp\in \Lambda}T_{\Comp}$. We first prove that the entries
made into the $\unsent$ variables in Step 3b can be mapped to virtual tree edges
in $\bigcup_{\Comp\in \Lambda}T_{\Comp}$.

\begin{lemma}\label{lemma:stage1_subtree}
Suppose in phase $i\in \{0,\ldots,i_v\}$ of the first stage, node $v\in V$ adds
$(l(v),v_i)$ or $(l(v),\tilde{v}_{i_v})$ to its $\unsent$ variable in Step 3b.
Then $\{v_i,v_{i-1}\}\in T_{\Comp'}$ for some $\Comp'\in \Lambda$.
\end{lemma}
\begin{proof}
Assume for contradiction that the statement is wrong and $i\in \{0,\ldots,i_v\}$
is the minimal phase in which some node $v$ violates it. Hence,
$\{v_i,v_{i-1}\}\notin T_{\Comp'}$ for any $\Comp' \in \Lambda$. For the subtree
$T_{v_{i-1}}$ of the virtual tree rooted at $v_{i-1}$, this implies that
$\Lambda$ is partitioned into the sets of labels $\{\Comp\in \Lambda\,|\,\exists
w\in T_{v_{i-1}}:\Comp(w)=\Comp\}$ and $\{\Comp\in \Lambda\,|\,\exists w\in
T\setminus T_{v_{i-1}}:\Comp(w)=\Comp\}$. By induction on phases $j\in
\{0,\ldots,i\}$, we see that at the beginning of each such phase $j$, $\Lambda$
is partitioned into the subsets $\{\Comp\in \Lambda\,|\,\exists w\in
T_{v_{i-1}}:\Comp\in l(w)\}$ and $\{\Comp\in \Lambda\,|\,\exists w\in T\setminus
T_{v_{i-1}}:\Comp\in l(w)\}$: for $0<j\leq i$, by the induction hypothesis and
Steps 3b to 3d of phase $j-1$ this would imply that there is a node $w$ on level
$j-1$ of the virtual tree that has at least one descendant from $T_{v_{i-1}}$
and one descendant outside $T_{v_{i-1}}$; this is impossible for $j-1\leq i-1$,
as the root of $T_{v_{i-1}}$ is on level $i-1$.

Due to Steps 3b and 3d in phase $i-1$, there is at most one node $w\in
T_{v_{i-1}}$ that has $l(w)\neq \emptyset$ at the end of phase $i-1$; by Step 3b
for phase $i$, it must hold that $w=v$. However, we just showed that each node
$w\notin T_{v_{i-1}}$ satisfies that $l(w')\cap l(w)=\emptyset$. Thus, $v$ sets
$l(v):=\emptyset$ in Step 3a of phase $i$, contradicting the assumption that it
adds an entry to its $\unsent$ variable in Step 3b of the phase. Therefore, our
assumption that the statement of the lemma is wrong must be false, concluding
the proof.
\end{proof}

With this lemma in place, we are ready to prove that the total weight of the
selected edge set does not exceed the weight of the optimal solution on the
virtual tree. This is done by charging the weight of a selected least-weight
path (or prefix of such a path) to the corresponding virtual tree edge given
by \lemmaref{lemma:stage1_subtree}.

\begin{lemma}\label{lemma:stage1_approx}
The weight of the set $F$ returned by the first stage is bounded from above
by the weight of an optimal solution on the virtual tree.
\end{lemma}
\begin{proof}
Suppose edge $e$ is added to $F$ in phase $i\in \{0,\ldots,L\}$.
This must have happened because in Step 3c of the phase, it was traversed by
some message $(\cdot,v_i)$ or $(\cdot,\tilde{v}_{i_v})$, where some node $v$
made the respective entry to its $\unsent$ variable in Step 3b of the phase. In
the latter case, we claim that $i_v=i$. Assuming the contrary, clearly $i_v>i$
and Steps 3b to 3d of phase $i-1$ would entail that $v$ was selected in Step 3d
of phase $i-1$ by some node $w$ for which it added an entry $(\cdot,w)$ to its
$\unsent$ variable in Step 3b of the phase. It follows that $w=\tilde{v}_{i_v}$,
and each edge on the respective least-weight path from $v$ to $\tilde{v}_{i_v}$
has been traversed by a message in Step 3c of phase $i-1$. In particular, $e$
was added to $F$ already in an earlier phase. Thus, indeed it must hold that
$i=i_v$.

Hence, $e$ is traversed by a message $(\cdot,v_i)$ or $(\cdot,\tilde{v}_i)$
in phase $i$. From \lemmaref{lemma:stage1_subtree}, we have that
$\{v_i,v_{i-1}\}\in T_{\Comp}$ for some $\Comp\in \Lambda$. Moreover, by Steps
3b and 3d of phase $i-1$, the node $v$ that made the respective entry in Step 3b
of phase $i$ is unique; there can be only one node $v$ in the subtree rooted at
$v_{i-1}$ that satisfies $l(v)\neq \emptyset$ at the beginning of phase $i$. We
``charge'' the weight of $e$ to the edge $\{v_i,v_{i-1}\}\in \bigcup_{\Comp\in
\Lambda}T_{\Comp}$. Because node $v$ is unique with the property that the
cost of edges traversed by messages $(\cdot,v_i)$ or $(\cdot,\tilde{v}_i)$ that
are charged to $\{v_i,v_{i-1}\}$ can be backtraced to an entry it made in Step
3b of phase $i$, virtual tree edge $\{v_i,v_{i-1}\}$ is in total charged at most
weight $\Wd(v,v_i)$ (if $i<i_v$) or $\Wd(v,\tilde{v}_i)$ (if $i=i_v$), the
weight of the respective least-weight paths in $G$ from $v$ to $v_i$ or
$\tilde{v}_{i_v}$, respectively. Because $\Wd(v,\tilde{v}_{i_v})\leq
\Wd(v,v_{i_v})$ and $\Wd(v,v_i)\leq \beta 2^i$, $\{v_i,v_{i-1}\}$ is in total
charged at most its own weight of $\beta 2^i$. We conclude that $W(F)=\sum_{e\in
F}W(e)$ is indeed at most the weight of the optimal solution on the virtual
tree, i.e., of the edge set of $\bigcup_{\Comp\in \Lambda}T_{\Comp}$.
\end{proof}

\subsubsection*{Feasibility}

It remains to examine what we have gained from selecting the edge set $F$ in
the first stage.

\begin{lemma}\label{lemma:stage1_connect}
For each terminal $v\in T\setminus {\cal S}$, at least one of the
following holds for the graph $(V,F)$, where $F$ is the output of the first
stage: (i) all terminals $w\in T$ with $\Comp(v)=\Comp(w)$ are in the same
connectivity component or (ii) $v$ is at most $\sO(\sqrt{n})$ hops
from a node in ${\cal S}$ w.h.p.
\end{lemma}
\begin{proof}
We claim that, for each phase $i\in \{0,\ldots,L\}$ and $\Comp \in \Lambda$,
the following holds w.h.p.: If at the beginning of the phase there are two or
more terminals $v$ with $\Comp\in l(v)$, at the end of the phase each such $v$
will be connected to a terminal $w$ with $\Comp\in l(w)$ by a path of
$\tilde{\BO}(\sqrt{n})$ hops in $(V,F)$. 

To see this, observe first that if there are two or more terminals $v$ with
$\Comp\in l(v)$ at the beginning of phase $i$, $\Comp$ will not be deleted from
the $l(v)$ variables of these nodes in Step 3a of the phase. Hence, each such
$v$ will add an entry $(\Comp,u)$, where either $u=v_i$ or $u=\tilde{v}_{i_v}$,
to its $\unsent$ variable in Step 3b of the phase. In Step 3c, all edges on the
least-weight path from $v$ to $u$ will be added to $F$. Each of the respective
paths has by Step 1 of the algorithm and \lemmaref{lemma:stage1_partial}
$\tilde{\BO}(\sqrt{n})$ hops w.h.p.

Due to Step 3c, $u$ will add $\Comp$ to $\hat{l}(u)$. In Step 3d, it will select
some node $w$ that added $(\cdot,u)$ to its $\unsent$ variable in Step 3b and
sent $\hat{l}(u)\ni \Comp$ to it; $w$ will hence set $l(w):=\hat{l}(u)\ni
\Comp$. Again, $w$ is connected to $u$ by a path of at most
$\tilde{\BO}(\sqrt{n})$ hops whose edges have been added to $F$, since in Step
3d a sequence of messages from Step 3c is backtraced. This shows the claim.

By induction on the phases, for each $\Comp\in \Lambda$, (i) each terminal $v\in
T$ with $\Comp(v)=\Comp$ is connected in $(V,F)$ to some node $w$ with
$\Comp\in l(w)$ via $\sO(\sqrt{n})$ hops w.h.p.\ at the end of the first
stage, or (ii) there is a unique node so that all terminals $v\in T$ with
$\Comp(v)=\Comp$ are connected by $F$ to this node. If (i) applies and
${\cal S}\neq \emptyset$, note that in phase $L$ all entries made to $\unsent$
variables in Step 3b were of the form $(\cdot,v_L)$, where $v_L$ is the root of
the virtual tree. Hence, all terminals $v\in T$ with $\Comp(v)=\Comp$ are
connected to the root of the virtual tree by edges in $F$. If ${\cal S}\neq
\emptyset$, the root of the virtual tree, i.e., the node of highest rank, must
be in ${\cal S}$. Hence, all entries made to $\unsent$ variables in Step 3b of
phase $L$ were of the form $(\cdot,\tilde{v}_{i_v})$, and all terminals $v\in T$
with $\Comp(v)=\Comp$ are connected to a node in~${\cal S}$.
\end{proof}

\begin{corollary}\label{coro:stage1_smalls_feasible}
If $s\leq \sqrt{n}$, the first stage solves problem $\sfic$.
\end{corollary}
\begin{proof}
Because ${\cal S}=\emptyset$ if $s\leq \sqrt{n}$, Statement~(i) of
\lemmaref{lemma:stage1_connect} applies to all terminals.
\end{proof}

\subsection{Spanner Construction and Completion Stage}\label{sec:spanner}


\subsubsection*{Running Time of the Transformation}
\begin{corollary}\label{coro:stage2_group}
Within $\sO(\sqrt{n})$ rounds, each $w\in \bigcup_{v\in {\cal S}}T_v$ can learn
the identifier of the node $v\in {\cal S}$ such that $w\in T_v$. W.h.p.,
terminals $w\notin \bigcup_{v\in {\cal S}}T_v$ satisfy that all terminals $u$ with
$\Comp(u)=\Comp(w)$ are connected in $(V,F)$.
\end{corollary}
\begin{proof}
Membership in $T_v$, $v\in {\cal S}$, can be concurrently determined for all
terminals $w\in T$ by running (essentially) the single-source Bellmann-Ford
algorithm for $\tilde{\BO}(\sqrt{n})$ rounds on the (unweighted) graph $(V,F)$,
with a virtual source connected by $0$-weight edges to nodes in ${\cal S}$ and
piggy-backing the identifiers of nodes in ${\cal S}$ on messages referring to
paths to them. This can be simulated on $G$ at no overhead, since for each
$\{v,w\}\in F$, $v$ and $w$ know that $\{v,w\}\in F$. The second statement of
the lemma directly follows from \lemmaref{lemma:stage1_connect}.
\end{proof}

As we will see, it is not necessary to construct $\hat{G}$ explicitly. However,
obviously we must determine $\hat{\Comp}$.

\begin{lemma}\label{lemma:stage2_compute}
For an $F$-reduced instance, $\hat{T}$ and $\hat{\Comp}$ can be computed and
made known to all nodes in $\tilde{\BO}(\sqrt{n}+k+D)$ rounds.
\end{lemma}
\begin{proof}
By \corollaryref{coro:stage2_group}, for each $v\in {\cal S}$ and each terminal
$w\in T_v$, $w$ can learn $v$ in $\tilde{\BO}(\sqrt{n})$ rounds. We also make
the set ${\cal S}$ global knowledge, by broadcasting it over the BFS tree; this
takes $\BO(|{\cal S}|+D)=\BO(\sqrt{n}+D)$ rounds. Next, each node
$v$ locally initializes $F(v):=\emptyset$ and $\Comp_u(v):=\bot$ for each
$u\in {\cal S}$ if $v\notin T\cap T_u$ and $\Comp_u(v):=\Comp(v)$ otherwise.
Subsequently, the following is executed on a BFS tree until no node sends any
further messages.
\begin{compactitem}
\item If node $v$ has stored an edge $e\in F(v)$ or some $\Comp_u(v)\neq
\bot$ that it has not yet sent, it sends a message with this information to its
parent.
\item If node $w$ receives a message ``$\Comp_u(v)$'' and it currently stores
$\Comp_u(w)=\bot$, it sets $\Comp_u(w):=\Comp_u(v)$.
\item If node $w$ receives a message ``$\Comp_u(v)$'' and it currently stores
$\Comp_u(w)=\Comp$, it adds edge $\{\Comp_u(v),\Comp\}$ to its set $F(w)$
unless the edge would close a cycle in $(\Lambda,F(w))$.
\item If node $w$ receives a message ``$\{\Comp,\Comp'\}$'', it adds
$\{\Comp,\Comp'\}$ to its set $F(w)$ unless the edge would close a cycle in
$(\Lambda,F(w))$.
\end{compactitem}
Since there are $|{\cal S}|=\sqrt{n}$ variables $\Comp_u(v)$ at each node $v$
and $F(v)$ remains a forest, each node sends at most $\sqrt{n}+k-1$ messages.
Since forests are matroids, we have optimal pipelining and no messages are sent
any more after $\BO(\sqrt{n}+k+D)$ rounds; this can be detected in additional
$\BO(D)$ rounds.

We claim that once the above subroutine terminated, at the root $v_R$ of the BFS
tree the connectivity components of $(\Lambda,E(v_R))$ are the same as the
connectivity components of $(\Lambda,F_{\Lambda})$. To see this, observe first
that $F(v_R)\subseteq E_{\Lambda}$, since a node $v$ adds an edge
$\{\Comp,\Comp'\}$ to its set $F(v_R)$ only if it either receives a message
``$\{\Comp,\Comp'\}$'' or it receives a message ``$\Comp$'' and stores
$\Comp_u(v)=\Comp'$ (or vice versa); this implies that some nodes $v'$ and
$w'$ must have had $\Comp_u(v')=\Comp$ and $\Comp_u(w')=\Comp'$ initially, which
by the initialization values of the variables implies that indeed
$\{\Comp,\Comp'\}\in E_{\Lambda}$. Hence, for any node $v$, any connectivity
component of $(\Lambda,F(v))$ is a subset of a connectivity component of
$(\Lambda,E_{\Lambda})$.

Now suppose that $\{\Comp,\Comp'\}\in E_{\Lambda}$. Thus, there are $u\in {\cal
S}$ and $v,w\in T_u\cap T$ so that $\Comp(v)=\Comp$ and $\Comp(w)=\Comp'$.
Consider the sequence of ancestors $v_0=v, v_1,\ldots,v_r=v_R$ of $v$ in the BFS
tree, and consider their variables $\Comp_u(v_i)$ and $F(v_i)$, $i\in
\{0,\ldots,r\}$, at the end of the computation. We will prove by induction on
$i$ that for each such $v_i$, $F(v_i)$ connects $\Comp$ to $\Comp_u(v_i)\neq
\bot$. Trivially, this holds for $v=v_0$, so assume that it holds for some $i\in
\{0,\ldots,r-1\}$ and consider $v_{i+1}$. Since $v_i$ sends $\Comp_u(v_i)$ at
some point, $\Comp_u(v_{i+1})\neq \bot$. At the latest upon reception of this
message, $\Comp_u(v_{i+1})$ and $\Comp_u(v_i)$ become connected by $F(v_{i+1})$;
since $\Comp_u(v_{i+1})$ is modified only once and $u_{i+1}$ can only add edges
to $F(v_{i+1})$, but not remove them, $\Comp_u(v_{i+1})$ and $\Comp_u(v_i)$ are
connected by $F(v_{i+1})$ when the subroutine terminates. By the induction
hypothesis, $\Comp_u(v_i)$ and $\Comp$ are connected by $F(v_i)$. Due to the
rules of the algorithm, $v_i$ will announce all edges in $F(v_i)$ at
some point to $v_{i+1}$. Whenever such a message is received, $v_{i+1}$ either
adds the edge to $F(v_{i+1})$ or its endpoints are already connected by
$F(v_{i+1})$. This shows that $\Comp_u(v_{i+1})$ eventually gets connected to
$\Comp$, i.e., the induction hypothesis holds for index $i+1$. In particular,
$F(v_R)$ connects $\Comp$ and $\Comp_u(v_R)$. Reasoning analogously for
$\Comp'$, $F(v_R)$ connects $\Comp'$ and $\Comp_u(v_R)$, and therefore also
$\Comp$ and $\Comp'$. Hence, any connectivity component of $(\Lambda,F(v_R))$ is
a superset of a connectivity component of $(\Lambda,E_{\Lambda})$.

We conclude that the connectivity components of $(\Lambda,F(v_R))$ are the same
as those of $(\Lambda,E_{\Lambda})$, as claimed. Since $F(v_R)$ is a forest, the
root can broadcast $(\Lambda,F(v_R))$ over the BFS tree in
$\BO(|\Lambda|+D)\subseteq\BO(k+D)$ rounds. From this, each node can determine
the connectivity components of $(\Lambda,E_{\Lambda})$ locally. Now, each node
can compute $\hat{T}$ (for $u\in {\cal S}$, $T_u\in \hat{T}$ iff it is not
isolated in $(\Lambda,E_{\Lambda})$) and $\hat{\Comp}$ locally, as ${\cal S}$ is
already known to all nodes.
\end{proof}

\subsubsection*{Feasibility}

\begin{lemma}\label{lemma:stage2_feasible}
Suppose $\hat{F}$ is a solution of an $F$-reduced instance, where $F$ is the set
returned by the first stage. Define $F'\subseteq E$ by selecting for each
$\hat{e}\in \hat{F}$ an edge $e\in E$ inducing it into $F'$. Then $F\cup F'$ is
a solution of the original instance w.h.p.
\end{lemma}
\begin{proof}
Suppose for $u,u'\in T$ we have that $\Comp(u)=\Comp(u')$. If $u$ or $u'$ are
not in $\bigcup_{v\in {\cal S}}T_v$, \corollaryref{coro:stage2_group} shows that
$F$ connects $u$ and $u'$ w.h.p. Hence, suppose that $u\in T_v$ and $u'\in
T_w$ for some $v,w\in {\cal S}$. This implies that
$\hat{\Comp}(T_v)=\hat{\Comp}(T_w)$. Because $\hat{F}$ solves the $F$-reduced
instance, there is a path in $(\hat{V},\hat{F})$ connecting $T_v$ and
$T_w$. By definition of $\hat{E}$ and induced edges together with the fact $F$
connects each of the sets $T_x$, $x\in {\cal S}$, $F'\cup F$ connects $u$ and
$u'$. Since $u,u'\in T$ where arbitrary with the property that
$\Comp(u)=\Comp(u')$, applying the union bound over all pairs of terminals
shows that $F\cup F'$ is a solution of the original instance w.h.p.
\end{proof}

\subsubsection*{Approximation Ratio}

\begin{lemma}\label{lemma:stage2_cost}
An optimal solution to an $F$-reduced instance has at most the weight of an
optimal solution of the original instance.
\end{lemma}
\begin{proof}
Denote by $F_o$ an optimal solution of the original instance. For each $u\in
{\cal S}$, drop all edges between nodes $v,w\in T_u$. The remaining edge set
induces an edge set $\hat{F}\subseteq \hat{E}$ of at most weight $W(F_o)$ in
$\hat{G}$, which we claim to be a solution to the reduced instance; from this
the statement of the lemma follows immediately.

Consider terminals $T_u,T_{u'}\in \hat{T}$ of the new instance with
$\hat{\Comp}(T_u)=\hat{\Comp}(T_{u'})$. By definition of $\hat{\Comp}$, this
entails that there are nodes $v\in T_u$ and $w\in T_{u'}$ and a path
$(\Comp_0=\Comp(v),\Comp_1,\ldots,\Comp_{\ell}=\Comp(w))$ in
$(\Lambda,E_{\Lambda})$. For each edge $\{\Comp_{i-1},\Comp_i\}$ on the path,
$i\in \{1,\ldots,\ell\}$, there is a node $u_i\in {\cal S}$ and terminals
$v_i,w_i\in T_{u_i}\cap T$ so that $\Comp(v_i)=\Comp_{i-1}$ and
$\Comp(w_i)=\Comp_i$. Because $F_o$ is a solution of the original instance and
$\Comp(v_i)=\Comp(w_{i-1})$ for each $i\in \{2,\ldots,\ell\}$, there is a path
in $F_o$ connecting $w_i\in T_{u_{i-1}}$ and $v_i\in T_{u_i}$. Hence, $\hat{F}$
connects $T_{u_{i-1}}$ and $T_{u_i}$. It follows that it also connects $T_{u_1}$
and $T_{u_{\ell}}$. As $\Comp_0=\Comp(v)$ and $\Comp_{\ell}=\Comp(w)$, there are
paths in $F_o$ that connect $v$ to $v_1$ and $w$ to $w_{\ell}$, respectively.
Therefore, $\hat{F}$ connects $T_u$ to $T_{u_1}$ and $T_{u'}$ to $T_{u_{\ell}}$,
respectively. Overall, $\hat{F}$ connets $T_u$ and $T_{u'}$. Since
$T_u,T_{u'}\in \hat{T}$ were arbitrary with the property that
$\hat{\Comp}(T_u)=\hat{\Comp}(T_{u'})$, we conclude that $\hat{F}$ is indeed
a solution of the $F$-reduced instance.
\end{proof}

\subsubsection*{Solving the New Instance}

\begin{lemma}\label{lemma:stage2_solve}
A solution $\hat{F}$ of the $F$-reduced instance determined by the output of the
first stage of weight $\BO(\log n)$ times the optimum can be found in
$\sO(\sqrt{n}+D)$ rounds w.h.p., in the sense that an inducing edge set
$F'\subseteq E$ is marked in $G$ that satisfies $W(F')=\hat{W}(\hat{F})$.
\end{lemma}
\begin{proof}
We use our algorithm from \cite{LenzenP13} with a minor tweak. The (unmodified)
algorithm proceeds in the following main steps.
\begin{compactitem}
\item Sample a uniformly random set ${\cal S}$ of
$|{\cal S}|\in \tilde{\Theta}(\sqrt{n})$ nodes.
\item Construct and make known to all nodes a spanner of the complete graph on
the node set $T\cup {\cal S}$, where the edge weights are the weighted distances
in $G$.
\item For each $v\in T$, make $\Comp(v)$ known to all nodes and locally solve
the instance on the spanner by a deterministic $\alpha$-approximation algorithm.
\item For each edge in the computed solution, select the edges from a
corresponding least-weight path in $G$ into the returned edge set.
\end{compactitem}
Adding the set ${\cal S}$ ensures that any least-weight path between pairs of
nodes in $T\cup {\cal S}$ that has no inner nodes from $T\cup {\cal S}$ has
$\sO(\sqrt{n})$ hops. This property is already guaranteed in $G$ and
thus also $\hat{G}$ due to the uniformly random set ${\cal S}$ of the $\sqrt{n}$
nodes of highest rank; therefore, it can be skipped.

To simulate the algorithm on $\hat{G}$, it suffices to slightly modify the
second step of the algorithm. The spanner construction iteratively grows
clusters of nodes that are connected by spanner edges, where usually the
clusters are initialized to the singletons given by the node set of the spanner.
In our setting, for each $v\in {\cal S}$, the nodes in $T_v$ are already
connected after the first stage and identified to a single node in $\hat{G}$.
To reflect this in the spanner construction, we simply initialize the clusters
to be the sets $T_v$, $v\in {\cal S}$; the algorithm then constructs a spanner
on the complete graph on $\{T_v\,|\,v\in {\cal S}\}$ with edge weights given by
distances in $\hat{G}$. The paths the algorithm detects and whose edges will be
returned in the last step of the algorithm have weight equal to the edge weights
in $\hat{G}$.

Because the third step operates on the spanner only, it does not have to be
modified. Using the (deterministic) moat-growing algorithm, which guarantees
$\alpha=2$, and parameter $k=\log n$ in the spanner construction, Theorem 5.2
from \cite{LenzenP13} shows that the returned edge set $F'$ has weight at most
$\BO(\log n)$ times the optimum of the $F$-reduced instance. The above
modifications to the algorithm do not affect the running time apart from
ensuring that the number of nodes in the spanner (and the instance of \sfic on
the spanner solved in the third step) becomes $|\hat{T}|$, so the analysis from
\cite{LenzenP13} yields a running time of $\sO(|\hat{T}|^{1+1/k}+D)\subseteq
\sO(\sqrt{n}+D)$.
\end{proof}

\subsection{Completing the Algorithm}

Finally, we can state the complete algorithm as follows.
\begin{compactenum}
\item For a sufficiently large constant $c$, run the first stage $c\log n$
times.
\item Among the computed edge sets, determine a set $F$ of minimal weight.
\item If $s\leq \sqrt{n}$, return $F$. Otherwise,
\begin{compactenum}
\item Compute the $F$-reduced instance.
\item Solve the $F$-reduced instance, resulting in edge set $F'$.
\item Return $F\cup F'$.
\end{compactenum}
\end{compactenum}

\begin{proof}[Proof of \theoremref{theorem:fast}]
The time complexity follows from the observation that checking the weight of an
edge set returned by the first stage can be done in $\BO(D)$ rounds using a BFS
tree, Lemmas~\ref{lemma:stage2_compute} and~\ref{lemma:stage2_solve}, and
Corollaries~\ref{coro:stage1_time} and~\ref{coro:stage2_group}.

In~\cite{KKMPT-12}, it is shown that the weight of the optimal solution on the
virtual tree is within factor $\BO(\log n)$ of the optimum in expectation. By
Markov's inequality, with probability at least $1/2$, this expectation is
exceeded by factor at most $2$. Hence, with probability at least $1-1/2^{c\log
n}=1-n^c$, i.e., w.h.p., at least one of the computed virtual trees exhibits an
optimal solution that is within factor $\BO(\log n)$ of the optimum for the
instance on $G$. By \lemmaref{lemma:stage1_approx}, the weight of the set $F$ is
at most that of the optimal solution on the corresponding virtual tree, implying
that the set $F$ determined in the second step of the algorithm has weight
within factor $\BO(\log n)$ of the optimum w.h.p.

For $s\leq \sqrt{n}$, by \corollaryref{coro:stage1_smalls_feasible} $F$ is a
solution, i.e., the claim of the theorem holds. For $s>\sqrt{n}$, the algorithm
proceeds to compute $F'$. By \lemmaref{lemma:stage2_solve}, $F'$ induces a
solution of the $F$-reduced instance, yielding by
\lemmaref{lemma:stage2_feasible} that $F\cup F'$ solves the original instance
w.h.p. \lemmaref{lemma:stage2_solve} also guarantees that $F'$ has weight within
factor $\BO(\log n)$ of the optimum of the $F$-reduced instance, which by
\lemmaref{lemma:stage2_cost} implies that $W(F')$ weighs also at most $\BO(\log
n)$ times optimum of the original instance. We conclude that $W(F\cup F')$ is
optimal up to factor $\BO(\log n)$ w.h.p. Applying the union bound over the
various statements that hold w.h.p., the statement of the theorem follows for
$s>\sqrt{n}$.
\end{proof}

\end{document}